\documentclass[12pt,reqno]{amsart}
\usepackage{amsthm,amsfonts,amssymb,euscript}

\newtheorem{proposition}{Proposition}
\newtheorem{theorem}{Theorem}
\newtheorem{remark}{Remark}

\usepackage[margin=1in,dvips]{geometry}
\usepackage{bbm}
\usepackage{enumerate}
\usepackage{graphics}

\def\ub {\underline{u}}
\def\th {\theta}
\def\Lb {\underline{L}}
\def\Hb {\underline{H}}
\def\chib {\underline{\chi}}
\def\chih {\hat{\chi}}
\def\chibh {\hat{\underline{\chi}}}
\def\omegab {\underline{\omega}}
\def\etab {\underline{\eta}}
\def\betab {\underline{\beta}}
\def\alphab {\underline{\alpha}}

\def\hot{\widehat{\otimes}}

\def\rhoc{\check{\rho}}
\def\sigmac{\check{\sigma}}
\def\Psic{\check{\Psi}}
\def\kappab{\underline{\kappa}}
\def\mub{\underline{\mu}}

\def\a {\alpha}
\def\b {\beta}
\def\ab {\alphab}
\def\bb {\betab}
\def\nab {\nabla}

\renewcommand{\div}{\mbox{div }}
\newcommand{\curl}{\mbox{curl }}
\newcommand{\trchb}{\mbox{tr} \chib}
\def\trch{\mbox{tr}\chi}
\newcommand{\tr}{\mbox{tr }}

\newcommand{\eps}{{\epsilon} \mkern-8mu /\,}
\newcommand{\Ls}{{\mathcal L} \mkern-10mu /\,}

\begin{document}

\title{Local Propagation of Impulsive Gravitational Waves}

\author{Jonathan Luk}
\address{Department of Mathematics, Princeton University, Princeton NJ 08544}
\email{jluk@math.princeton.edu}

\author{Igor Rodnianski}
\address{Department of Mathematics, Princeton University, Princeton NJ 08544 and Department of Mathematics, MIT, Cambridge, MA 02139}
\email{irod@math.princeton.edu, irod@math.mit.edu}

\begin{abstract}
In this paper, we initiate the rigorous mathematical study of the problem of impulsive gravitational spacetime waves. We construct such spacetimes as solutions to the characteristic initial value problem of the Einstein vacuum equations with a data curvature delta singularity. We show that in the resulting spacetime, the delta singularity propagates along a characteristic hypersurface, while away from that hypersurface the spacetime remains smooth. Unlike the known explicit examples of impulsive gravitational spacetimes, this work in particular provides the first construction of an impulsive gravitational wave of compact extent and does not require any symmetry assumptions. The arguments in the present paper also extend to the problem of existence and uniqueness of solutions to a larger class of non-regular characteristic data.
\end{abstract}

\maketitle

\tableofcontents

\section{Introduction}

In this paper, we initiate the study of the characteristic initial value problem for impulsive gravitational spacetimes in general relativity. These were considered to be solutions of the vacuum Einstein equations
$$R_{\mu\nu}=0$$
with a delta singularity in the Riemann curvature tensor supported on a null hypersurface. Their historical origin can be traced back to the cylindrical waves of Einstein-Rosen \cite{EinsteinRosen}, the plane waves of Brinkmann \cite{Brinkmann} and the explicit impulsive gravitational spacetimes of Penrose \cite{Penrose72}.

Impulsive gravitational waves have often been studied within the class of plane fronted gravitational waves. This class of explicit solutions to the vacuum Einstein equations has been first studied by Brinkmann \cite{Brinkmann} and interest in them has been revived in later decades by the work of Bondi-Pirani-Robinson \cite{BPR}. They were later classified geometrically by Jordon-Ehlers-Kundt \cite{JEK}. Among this class are the \emph{pp}-waves (plane fronted waves with parallel rays) that were discovered by Brinkmann \cite{Brinkmann}, for which the metric takes the form
$$g=2d\ub dr+H(\ub,X,Y)d\ub^2+dX^2+dY^2,$$
and the Einstein vacuum equations imply that 
\begin{equation}\label{H}
\frac{\partial^2 H}{\partial X^2}+\frac{\partial^2 H}{\partial Y^2}=0.
\end{equation}
These include the special case of sandwich waves, where $H$ is compactly supported in $\ub$. Since (\ref{H}) is linear, it follows that \emph{pp}-waves enjoy a principle of linear superposition.

\emph{pp}-waves have a plane symmetry and originally impulsive gravitational waves have been thought of as a limiting case of the \emph{pp}-wave with the function $H$ admitting a delta singularity in the variable $\ub$. Precisely, explicit impulsive gravitational spacetimes were discovered and studied by Penrose \cite{Penrose72}, who gave the metric in the following double null coordinate form:
$$g=-2dud\ub+(1-\ub\Theta(\ub))dx^2+(1+\ub\Theta(\ub))dy^2,$$
where $\Theta$ is the Heaviside step function. In the Brinkmann coordinate system, the metric has the \emph{pp}-wave form and an obvious delta singularity:
$$g=-2d\ub dr-\delta(\ub)(X^2-Y^2)d\ub^2+dX^2+dY^2,$$
where $\delta(\ub)$ is the Dirac delta. Despite the presence of the delta singularity for the metric in Brinkmann coordinates, the corresponding spacetime is Lipschitz and it is only the Riemann curvature tensor (specifically, the only non-trivial $\alpha$ component of it) that has a delta function supported on the plane null hypersurface $\ub=0$. This spacetime turns out to possess remarkable global geometric properties \cite{Penrose65}. In particular, it exhibits a strong focusing property such that the whole null cone emanating from a point in the past of $\ub=0$ refocuses to a single line in the future of $\ub=0$ (See Figure 1). In \cite{Penrose65}, this property forms the basis of Penrose's argument that global hyperbolicity fails in this spacetime. Penrose's impulsive spacetime is plane symmetric, non-asymptotically flat and the delta curvature singularity is supported on the 3-dimensional infinite plane $\{(u,\ub,x,y) : \ub=0\}$. It has been long debated whether the strong focussing property and the resulting lack of global hyperbolicity is directly tied to the infinite extent of the impulsive gravitational wave (See Yurtsever \cite{Yurtsever}).

\begin{figure}[htbp]
\begin{center}
 
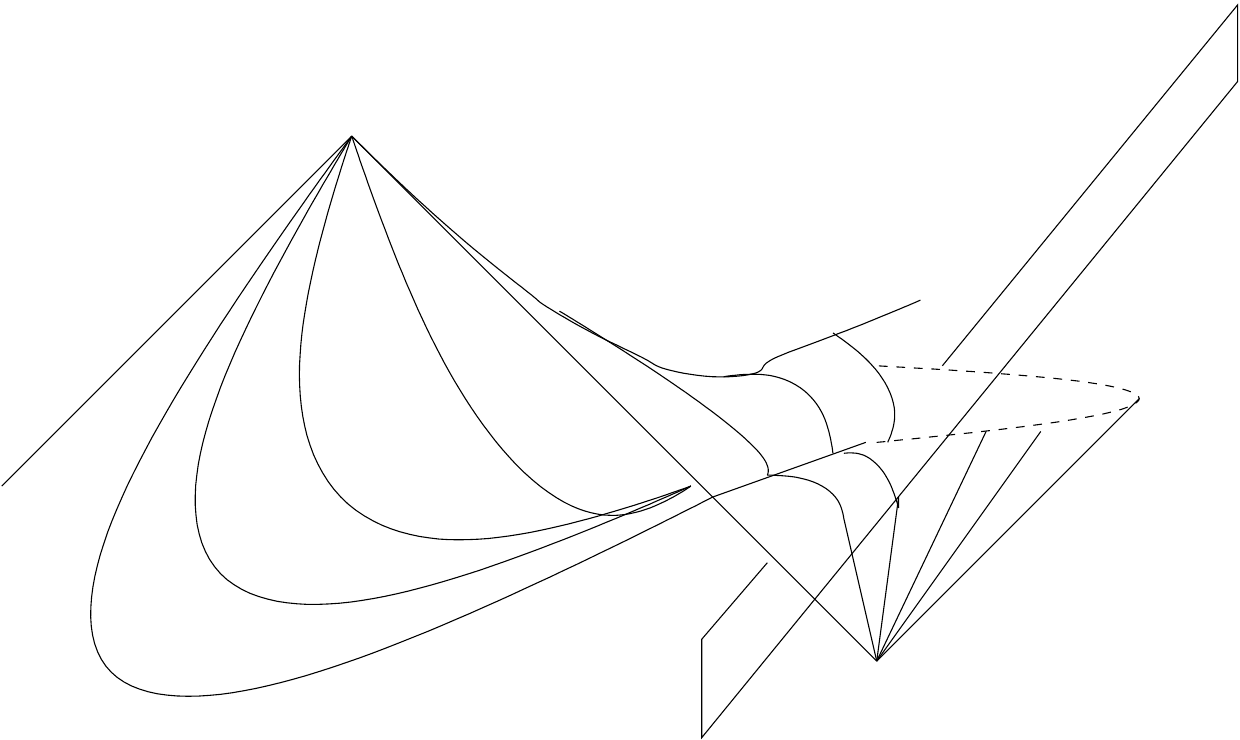
 
\caption{Focusing in Penrose's impulsive gravitational spacetime}
\end{center}
\end{figure}

While the interest in impulsive gravitational spacetimes had been high due to the availability of explicit solutions, their global geometric properties, superposition properties and their limiting relation to general spacetimes (see Penrose \cite{Penrose72} , Aichelburg-Sexl \cite{AichSexl}), the first study of general spacetimes satisfying the Einstein equations and admitting possible three dimensional delta singularities was undertaken by Taub \cite{Taub}, who derived a system of consistency relations linking the metric, curvature tensor and the geometry of the singular hypersurface. Impulsive gravitational spacetimes also arise as high-frequency limiting spacetimes considered by Choquet-Bruhat \cite{Bruhat}. We refer the readers to \cite{Gr}, \cite{GrPo}, \cite{BaHo}, \cite{Bicak} and the references therein for a more detailed exposition on the physics literature.

In this paper, we begin the study of impulsive gravitational spacetimes viewed in the context of an (characteristic) initial value problem. We consider the data, prescribed on an outgoing null hypersurface $u=0$ and assume that the $\alpha$ component of the curvature has a delta singularity supported on a two dimensional surface $S_{0,\ub_s}$, which can be thought of as the intersection of the hypersurfaces $u=0$ and $\ub=\ub_s$. Observe that in the Penrose's explicit impulsive solution, the $\alpha$ curvature component has precisely this type of behavior when restricted to $u=0$. Unlike the explicit impulsive spacetimes, which have only been constructed in plane symmetry and thus have infinite spatial extent, we will consider the case of $S_{0,\ub_s}$ with compact topology, more precisely a sphere. The data on the incoming null hypersurface is prescribed to be smooth but otherwise without any smallness assumptions.

With this data, we show that a unique spacetime satisfying the vacuum Einstein equations can be constructed locally\footnote{Notice that while the Riemann curvature tensor admits a delta function singularity, we show that the Ricci curvature tensor is well-defined in $L^2$, allowing us to make sense of the vacuum Einstein equations.}. Moreover, the delta singularity propagates along a null hypersurface emanating from the initial singularity on $S_{0,\ub_s}$ and the spacetime is smooth away from this null hypersurface (See Figure 2). 

\begin{figure}[htbp]
\begin{center}
 
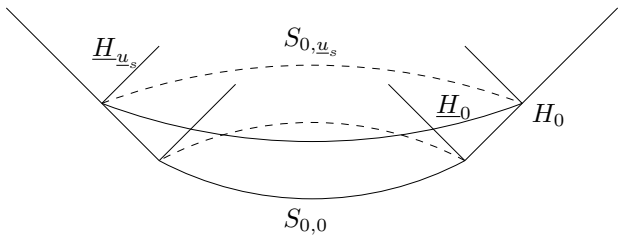

\caption{Propagation of Singularity}
\end{center}
\end{figure}

Our main result on the propagation of an impulsive gravitational wave is described by
\begin{theorem}\label{giwthmv1}
Suppose the following hold for the initial data set:
\begin{itemize}
\item The data on $\Hb_0$ is smooth. 
\item The data on $H_0$ is smooth except across a two sphere $S_{0,\ub_s}$, where the traceless part of the second fundamental form $\chih$ has a jump discontinuity.
\end{itemize}
Then
\begin{enumerate}[(a)]
\item Given such initial data and $\epsilon$ sufficiently small, there exists a unique spacetime $(\mathcal M,g)$ endowed with a double null foliation $u$, $\ub$ that solves the characteristic initial value problem for the vacuum Einstein equations in the region $0\leq u\leq u_*$, $0\leq \ub\leq\ub_*$ whenever $u_*, \ub_* \leq \epsilon$.
\item Define $\Hb_{\ub_s}$ to be the incoming null hypersurface emanating from $S_{0,\ub_s}$. Then the curvature components $\alpha_{AB}=R(e_A,e_4,e_B,e_4)$ are measures with a singular atom supported on the hypersurface $\Hb_{\ub_s}$. All other components of the curvature tensor can be defined in $L^2$. Moreover, the solution is smooth away from $\Hb_{\ub_s}$.
\end{enumerate}
\end{theorem}

For an impulsive gravitational wave, standard local existence shows that the spacetime is smooth in $0\leq \ub <\ub_s$. The problem of constructing an impulsive gravitational wave lies in making sense of the solution in the whole region with $\ub_s\leq \ub\leq \ub_*$ and showing that the singularity propagates along the characteristic hypersurface $\Hb_{\ub_s}$ and that the spacetime is smooth after the impulse.

\begin{remark}[Larger class of initial data]
The proof introduces a new type of energy estimates for the Einstein vacuum equations which allows some components of the Riemann curvature tensor not to be in $L^2$. This discovery allows us to consider the problem of local existence and uniqueness for a larger class of non-regular data. This includes data with a Riemann curvature tensor that can only be understood as a conormal distribution. 
\end{remark}

\begin{remark}[Uniqueness]
Penrose's construction of explicit impulsive gravitational spacetimes is based on the gluing approach, which has later been used to generate other explicit solutions of the Einstein equations. In the general case of the characteristic initial value problem for an impulsive gravitational wave, an appropriate adaptation of the gluing philosophy would allow us to construct {\bf weak} solutions with undetermined uniqueness. In this paper, however, we construct {\bf strong} solutions in the sense that {\bf uniqueness} can also be established. To prove uniqueness, we establish a priori estimates for a larger class of admissible initial data. This allows us to show that the solution is unique among all $C^0$ limits of smooth solutions to the vacuum Einstein equations.
\end{remark}

\begin{remark}[Propagation of singularity]
Our theorem gives a precise description of the propagation of the initial singularity. Such problems have been extensively studied for semilinear equations (see for example \cite{Beals}, \cite{Metivier}). Our result can be formally compared to the works of Majda on the propagation of shocks \cite{Maj1}, \cite{Maj2} for systems of conservation laws and the subsequent \cite{Alinhac}, \cite{Metivier2}, \cite{Metivier3}, \cite{CS}. In these works, a short time existence, uniqueness and regularity result was established for initial data with a jump discontinuity across a surface and a precise description of the propagation of the singularity was also given. We establish an analogous result of propagation of singularity for a nonlinear system of quasilinear hyperbolic equations. However, contrary to \cite{Maj1}, \cite{Maj2} the singularity that is considered in the present paper is not a shock, as it propagates along the characteristics. Moreover, unlike in \cite{Maj1}, \cite{Maj2}, \cite{Alinhac},\cite{Metivier2}, \cite{Metivier3}, \cite{CS}, where the problem is reformulated as an initial-boundary value problem and uniqueness is known only within the class of piecewise smooth solutions, our solution is also unique among limits of smooth solutions. In order to achieve this, the special structure of the Einstein equations in the double null foliation gauge has been heavily exploited. 
\end{remark}

The construction of spacetimes from non-regular characteristic initial data consistent with that of an impulsive gravitational wave had been known only under symmetry assumptions \cite{ChrSph1}, \cite{LeSm}, \cite{LeSte2}. The work of Christodoulou, who solved the characteristic initial value problem for data with bounded variation for the spherically symmetric Einstein-scalar field system \cite{ChrSph1}, can be thought of as a first result in that direction. In particular, in this work, the second derivatives of the scalar field, which formally is analogous to the curvature, is allowed to have a delta singularity. The study of data with bounded variation turned out to have important consequences in the global structure of spacetimes and the resolution of the cosmic censorship conjecture for the spherically symmetric Einstein-scalar field system \cite{ChrSph2}, \cite{ChrSph3}. The construction of distributional solutions for the vacuum Einstein equations that include \emph{plane} impulsive gravitational wave was carried out in \cite{LeSm}, \cite{LeSte2}. The present paper is the first work that provides a consistent study of the initial value problem of impulsive gravitational spacetimes, including their existence, uniqueness and propagation of singularity/regularity.

Going beyond spacetimes which represent a single impulsive gravitational wave, colliding impulsive gravitational waves had been studied by Khan-Penrose \cite{KhanPenrose} and Szekeres \cite{Szekeres}. In these explicit solutions, the spacetimes possess two null hypersurfaces with curvature delta singularity with a transverse intersection, representing the nonlinear interaction of two impulsive gravitational waves. The study of the characteristic initial value problem for the colliding impulsive gravitational waves will be carried out by the authors in a subsequent paper.

\subsection{First Version of the Theorem}

Our general approach is based on energy estimates and transport equations in the double null foliation gauge. This general approach in the double null foliation gauge has been carried out in \cite{KN}, \cite{Chr} and \cite{KlRo}.

The spacetime in question will be foliated by families of outgoing and incoming null hypersurfaces $H_u$ and $\Hb_{\ub}$ respectively. Their intersection is assumed to be a 2-sphere denoted by $S_{u,\ub}$. Define a null frame $\{e_1,e_2,e_3,e_4\}$, where $e_3$ and $e_4$ are null, as indicated in Figure 3, and $e_1$, $e_2$ are tangent to the two spheres $S_{u,\ub}$. $e_4$ is tangent to $H_u$ and $e_3$ is tangent to $\Hb_{\ub}$.

\begin{figure}[htbp]
\begin{center}
 
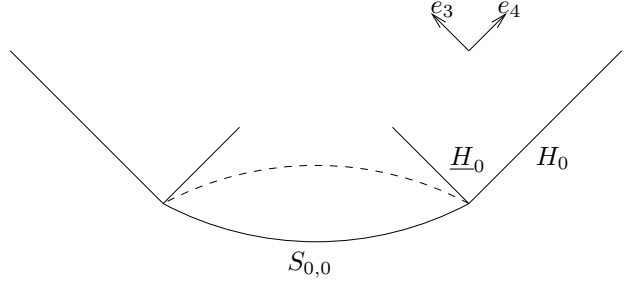
 
\caption{The Basic Setup and the Null Frame}
\end{center}
\end{figure}
Decompose the Riemann curvature tensor with respect to this frame:
\begin{equation*}
\begin{split}
\a_{AB}&=R(e_A, e_4, e_B, e_4),\quad \, \,\,   \ab_{AB}=R(e_A, e_3, e_B, e_3),\\
\b_A&= \frac 1 2 R(e_A,  e_4, e_3, e_4) ,\quad \bb_A =\frac 1 2 R(e_A,  e_3,  e_3, e_4),\\
\rho&=\frac 1 4 R(e_4,e_3, e_4,  e_3),\quad \sigma=\frac 1 4  \,^*R(e_4,e_3, e_4,  e_3)
\end{split}
\end{equation*}
In the context of impulsive gravitational spacetimes, the $\alpha$ component can only be understood as a measure. 

\noindent Define also the following Ricci coefficients with respect to the null frame:
\begin{equation*}
\begin{split}
&\chi_{AB}=g(D_A e_4,e_B),\, \,\, \quad \chib_{AB}=g(D_A e_3,e_B),\\
&\eta_A=-\frac 12 g(D_3 e_A,e_4),\quad \etab_A=-\frac 12 g(D_4 e_A,e_3)\\
&\omega=-\frac 14 g(D_4 e_3,e_4),\quad\,\,\, \omegab=-\frac 14 g(D_3 e_4,e_3),\\
&\zeta_A=\frac 1 2 g(D_A e_4,e_3)
\end{split}
\end{equation*}

Define $\chih$ (resp. $\chibh$) to be the traceless part of $\chi$ (resp. $\chib$). For the problem of the propagation of impulsive gravitational waves, we prescribe initial data on $H_0$ such that $\chih$ has a jump discontinuity across $S_{0,\ub_s}$ but smooth otherwise. On $\Hb_0$, we prescribe the initial data to be smooth but without any smallness assumptions.

As mentioned before, we will prove local existence and uniqueness for a class of data more general than that for the impulsive gravitational wave. More precisely, we require that $\chih$ and its angular derivatives are merely bounded. However, we do not require the derivative of $\chih$ in the $e_4$ direction even to be defined.

\begin{theorem}\label{rdthmv1}
Suppose the characteristic initial data are smooth on $\Hb_0$. On $H_0$, the characteristic initial data are determined by $\chi$ and $\zeta_A$. Let $\th^A$ be the coordinates on the two sphere foliating $H_0$. Suppose, in every coordinate patch
$$\sum_{i\leq 3}|(\frac{\partial}{\partial\th})^i\chih_{AB}|,\sum_{k\leq 1}\sum_{i\leq 3} |(\frac{\partial}{\partial\ub})^k(\frac{\partial}{\partial\th})^i(\trch,\zeta_A)|\leq C.$$
Then there exists $\epsilon$ sufficiently small such that the unique solution to the vacuum Einstein equations $(\mathcal M,g)$ endowed with a double null foliation $u$, $\ub$ exists in $0\leq u\leq \epsilon$, $0\leq \ub\leq \epsilon$. Associated to the spacetime a coordinate system $(u,\ub,\th^1,\th^2)$ exists, relative to which the spacetime is in particular Lipschitz and retains higher regularity in the angular directions.
\end{theorem}

\begin{remark}
We will also explicitly construct a class of initial data that satisfy the assumptions of Theorem \ref{giwthmv1} in Section \ref{initialcondition}. This construction of initial data will indicate an easy modification leading to an even more general class satisfying the assumptions of Theorem \ref{rdthmv1}.
\end{remark}

This Theorem in particular implies Theorem \ref{giwthmv1}(a). An additional argument, based on the estimates derived in the proof of Theorem \ref{rdthmv1}, will be carried out in Section \ref{limitgiw} to prove part (b) of Theorem \ref{giwthmv1}. The statement of Theorem \ref{rdthmv1} will be made precise in Section \ref{thmstatement}. In particular, the smoothness assumptions of Theorem \ref{rdthmv1} can be weakened and we will define in what sense the vacuum Einstein equations are satisfied.

\begin{remark}[Comparison to local existence results]
Without symmetry assumptions, all known proofs of existence of spacetimes satisfying the Einstein equations are based on $L^2$-type estimates for the curvature tensor and its derivatives or the metric components and their derivatives. Even with the recent resolution of the $L^2$ curvature conjecture by Klainerman-Rodnianski-Szeftel \cite{L21}, \cite{L22}, \cite{L23}, \cite{L24}, \cite{L25}, the Riemann curvature tensor has to be at least in $L^2$. (The classical local existence result of Hughes-Kato-Marsden \cite{HKM} requires the metric to be in $H^s$, with $s> \frac 52$.)

In this paper, we prove local existence and uniqueness under the assumption that the spacetime is merely Lipschitz, which in terms of differentiability is even one derivative weaker than the $L^2$ curvature conjecture. Of course the Lipschitz assumption refers to the worst possible behavior observed in our data and our result heavily relies on the structure of the Einstein equations which allows us to efficiently exploit the better behavior of the other components.
\end{remark}

\subsection{Strategy of the Proof}

For an impulsive gravitational wave, the curvature tensor can only be defined as a measure and is not in $L^2$. This is one of the main challenges of this work.  Let $\Psi$ denote the curvature components and $\psi$ denote the Ricci coefficients. The key new observation of this paper is that the $L^2$-type energy estimates for the components of the Riemann curvature tensor
$$\int_{H_u} \Psi^2+\int_{\Hb_{\ub}} \Psi^2\leq \int_{H_0} \Psi^2+\int_{\Hb_0} \Psi^2 +\int_0^{\ub} \int_0^u\int_{S_{u
',\ub'}} \psi\Psi\Psi du' d\ub'$$
coupled together with the null transport equations for the Ricci coefficients 
$$\nab_3\psi=\Psi+\psi\psi,\quad\nab_4\psi=\Psi+\psi\psi$$
can be renormalized and closed avoiding the $L^2$-non-integrable components of curvature.

\subsubsection{Renormalized Energy Estimates}\label{eeoutline}

The difficulty in carrying out the above argument is that in our setting, not all curvature components are defined in $L^2$. Even at the level of initial data, $\alpha$ is only defined as a measure. Therefore, we need to prove $L^2$ energy estimates without involving $\alpha$. To this end, we need to introduce and estimate the renormalized curvature components. The idea of renormalizing the curvature components has been introduced in \cite{KlRo:causal}. Unlike this work, there the renormalization was used to obtain estimates for the Ricci coefficients.

The classical way to derive energy estimates for the Einstein equations is via the Bel-Robinson tensor. In view of our renormalization, we avoid the use of the Bel-Robinson tensor and instead prove the energy estimates directly from the Bianchi identities (\ref{eq:null.Bianchi}). We note that the derivation of the energy estimates directly from the Bianchi equations without using the Bel Robinson tensor has also appeared in the work of Holzegel \cite{Holzegel} in a different setting. The challenge and motivation of this method is the derivation of estimates not involving the singular component of curvature $\alpha$.

To illustrate this, we first prove the energy estimates for $\beta$ on $H_u$ and for $(\rho,\sigma)$ on $\Hb_{\ub}$ by considering the following set of Bianchi equations:
$$\nabla_4\rho=\div\beta - \frac 12 \chibh\cdot\alpha+...$$
$$\nab_4\sigma=-\div ^*\beta+\frac 12\chibh\cdot\alpha+...$$
$$\nab_3\beta=\nab\rho+\nab^*\sigma+...$$
However, the curvature component $\alpha$ still appears in the nonlinear terms in these equations. In order to deal with this problem, we renormalize $\rho$ and $\sigma$. Define
$$\rhoc=\rho-\frac 12 \chih\cdot\chibh,\quad \sigmac=\sigma+\frac 12 \chih\wedge\chibh.$$
Using the equation
$$\nabla_4\chih=-\alpha+...,$$
we notice that the first two equations become
$$\nabla_4\rhoc=\div\beta+...,\quad\nab_4\sigmac=-\div ^*\beta+...$$
At the same time, the equation for $\beta$ can be re-written in terms of of $\rhoc$ and $\sigmac$:
$$\nab_3\beta=\nab\rhoc+\nab^*\sigmac+\psi\nab\psi+...$$
Now we have a set of renormalized Bianchi equations that does not contain $\alpha$. Using these equations, we derive the renormalized energy estimate
$$\int_{H_u} \beta^2+\int_{\Hb_{\ub}} (\rhoc,\sigmac)^2\leq \int_{H_0} \beta^2+\int_{\Hb_0} (\rhoc,\sigmac)^2 +\int_0^{\ub} \int_0^u\int_{S_{u
',\ub'}} \psi\Psi\Psi+\psi\psi\nab\psi du' d\ub',$$
in which $\alpha$ does not appear in the error term.

The same philosophy can be applied for the remaining curvature components ($\rho,\sigma,\betab,\alphab$). As a consequence, we obtain a set of $L^2$ curvature estimates which do not explicitly couple to the singular curvature component $\alpha$. We say explicitly that there is still a remaining possibility that the Ricci coefficients $\psi$ appearing in the nonlinear error terms for the energy estimates may depend on $\alpha$.

\subsubsection{Estimates for the Ricci Coefficients}\label{Riccioutline}

In order to close the estimates, it is necessary to show that all the Ricci coefficients can be estimated without \emph{any} knowledge of $\alpha$. The most dangerous component is $\chih$, which naively would have to be estimated using the transport equation
$$\nabla_4\chih+\trch\chih=-2\omega\chih-\alpha.$$
We take an alternate route and estimate $\chih$ from the equation
$$\nab_3\chih+\frac 1 2 \trchb \chih=\nab\widehat{\otimes} \eta+2\omegab \chih-\frac 12 \trch \chibh +\eta\widehat{\otimes} \eta,$$
or
$$\div\chih=\frac 12 \nabla \trch - \frac 12 (\eta-\etab)\cdot (\chih -\frac 1 2 \trch) -\beta.$$

As we shall see, the loss of information of $\alpha$ is accompanied by the loss of information of $\beta$ on $\Hb_{\ub}$. This presents yet an additional challenge in estimating the Ricci coefficients.

\subsubsection{Higher Order Energy Estimates}\label{higherorderoutline}

Another difficulty arises considering the fact that in order to close the energy estimates, we need to prove higher derivative estimates for the curvature components. However, the derivatives of some curvature components along $e_4$ are not defined in $L^2$ initially. We will therefore only use angular covariant derivatives $\nab$ as commutators and will prove estimates only for the $L^2$ norms of the angular covariant derivatives of the renormalized curvature components. We will show firstly that this procedure does not introduce terms that cannot be estimated in $L^2$ (in particular $\alpha$ will not appear) and secondly that all the energy estimates can be closed only using the estimates of the angular derivatives of the renormalized curvature components.

\subsubsection{Existence and Uniqueness}

Since we work with initial data with very low regularity, the a priori estimates that we have described above in Sections \ref{eeoutline}, \ref{Riccioutline} and \ref{higherorderoutline} do not immediately imply the existence and uniqueness of the solutions. Instead, we need to approximate the data by a sequence of smooth data and show first that they have a common domain of existence and second that they converge in this domain. The fact that the solutions to the sequence of smooth data have a common domain of existence follows from the a priori estimates that do not involve $\alpha$ and its derivatives as outlined above. Since the approximating data are smooth, we can conclude that the approximating solutions are also smooth and exist in a common domain.

Once we have proved the uniform a priori estimates, we proceed to prove that the sequence of solutions converges. To this end, we consider the equations for the difference of the Ricci coefficients and curvature components. We identify the spacetimes in this sequence by the value of their coordinate functions and derive equations for the difference of the Ricci coefficients and curvature components. Our a priori estimates heavily relies on the structure of the Einstein equations which allows us to eliminate any dependence of the $\alpha$ component. A priori, there is no reason to think that this structure is preserved when considering difference of these solutions, which is necessary to show that our end result is a strong solution to the Einstein equations. It is a remarkable fact that it turns out that in the equations for the difference of the Ricci coefficients and curvature components, $\alpha$ indeed does not appear. Thus the a priori estimates we have proved are sufficient to control the difference of the Ricci coefficients and curvature components from the difference of the initial value. This proves that the sequence of solutions converge. Even though we face the standard challenge of loss of derivative in the quasilinear equations, our estimates are still sufficient to show that the sequence converges.

The constructed limiting spacetime is not smooth. In particular, second derivative of the metric in the $e_4$ direction is not even defined. We will, however, show that the Ricci curvature tensor is better behaved than a general second derivative of the metric and that the limiting spacetime satisfies the vacuum Einstein equations $R_{\mu\nu}=0$ in the $L^2$ sense.

The estimates for the difference of the Ricci coefficients and curvature components imply that the constructed spacetime is the unique solution to the Einstein equations among the class of spacetimes admitting a double null foliation and satisfying strong enough a priori bounds. Moreover, we can prove uniqueness of the constructed solution among all limits of smooth spacetimes, i.e, any spacetime that arises as a $C^0$ limit of smooth solutions to the vacuum Einstein equations with initial data converging to the given initial data must coincide with the constructed spacetime.

It can be observed that the above argument for the a priori estimates, as well as that for showing existence and uniqueness, does not use the fact that initially, $\alpha$ is a measure whose singular support is on $S_{0,\ub_s}$. In effect, the argument avoids $\alpha$ completely, and can be used for data such that $\alpha$ is much rougher. Since the argument used only estimates on the Ricci coefficients and the curvature components other than $\alpha$, it can be used to handle initial data satisfying only the assumptions of Theorem \ref{rdthmv1}.

\subsubsection{Regularity and Propagation of Singularity}

In the setting of Theorem \ref{giwthmv1}, i.e., that of an impulsive gravitational wave, the theorem gives a precise description of the propagation of singularity. The a priori estimates imply that all the curvature components other than $\alpha$ are bounded. Here, we are interested in proving two additional statements: firstly, $\alpha$ is a measure that indeed has a delta singularity on $\Hb_{\ub_s}$; secondly, the spacetime is smooth away from $\Hb_{\ub_s}$.

To show that $\alpha$ is a measure, we approximate the data for $\alpha$ by a sequence of smooth data $\alpha_n$ such that $\alpha_n$ is of size $2^n$ in a region $|\ub-\ub_s|\leq 2^{-n}$. The spacetimes constructed for such data are smooth and therefore allow us to use the previously avoided $L^2$ estimate for the $\alpha$ component of curvature. This estimate imply that in the constructed spacetime, in the region $|\ub-\ub_s|\geq 2^{-n}$, the $L^2_{\ub}$ norm of $\alpha_n$ are uniformly bounded; while in the region $|\ub-\ub_s|\leq 2^{-n}$, the $L^2_{\ub}$ norm of $\alpha_n$ is bounded by $2^{\frac n2}$. By Cauchy-Schwarz, we have that the $L^1_{\ub}$ norms of $\alpha_n$ are uniformly bounded. This allows us to show that in the limiting spacetime, the curvature component $\alpha$ is a measure.

To show that the singular part of $\alpha$ is a delta function supported on the null hypersurface $\Hb_{\ub_s}$, we notice that
$$\alpha=-\nab_4\chih-\trch \chih-2 \omega \chih.$$
Therefore, it suffices to show that $\chih$ has a jump discontinuity across $\Hb_{\ub_s}$ and smooth everywhere else. This can be proved using the equation
$$\nab_3\chih+\frac 1 2 \trchb \chih=\nab\widehat{\otimes} \eta+2\omegab \chih-\frac 12 \trch \chibh +\eta\widehat{\otimes} \eta$$
and the fact that on the initial hypersurface $H_0$, $\chih$ has a jump discontinuity across $S_{0,\ub_s}$ and smooth everywhere else.

In order to show that the spacetime is smooth away from $\Hb_{\ub_s}$, we will estimate the higher regularity of all the curvature components in that region. We first use the Bianchi equations
$$\nab_3\beta+\trchb\beta=\nabla\rho + 2\omegab \beta +^*\nabla\sigma +2\chih\cdot\betab+3(\eta\rho+^*\eta\sigma).$$
Integrating this equation and using Gronwall's inequality, we obtain for any $\ub\neq\ub_s$ that
$$\sup_u\sum_{i\leq I}||\nab^i\beta||_{L^2(S_{u,\ub})}\leq (\sum_{i\leq I}||\nab^i\beta||_{L^2(S_{0,\ub})}+...)\exp(\int \sum_{i\leq I}||\nab^i(\trch,\omegab)||_{L^\infty(S)}).$$
The regularity of $\beta$ is thus inherited from the initial data, which is smooth away from $\ub=\ub_s$.
Once we have estimates for $\beta$, we consider the equation
$$\nab_3\alpha+\frac 12 \trchb \alpha=\nabla\hot \beta+ 4\omegab\alpha-3(\chih\rho+^*\chih\sigma)+
(\zeta+4\eta)\hot\beta.$$
Integrating as before, we see that $\alpha$ also inherits the regularity from the initial data. Higher derivatives estimates for $\alpha$ can be derived analogously by differentiating this equation. The other components of curvature can be controlled in a similar fashion. Notice that this procedure results in a loss of derivatives. In particular, in order to control the $N$-th derivative of the curvature, one needs $\sim 2N$ derivatives initially.

\subsection{Outline of the Paper}

Finally, we give the outline of the remainder of the paper. In the next Section, we give a careful introduction of the setting, describing the double null foliation, the coordinate system, the equations and relevant notations. We will state a more precise version of Theorem \ref{rdthmv1}, which we will call Theorem \ref{rdthmv2}. In Section \ref{initialcondition}, we provide a construction of the initial data set satisfying the conditions in Theorem \ref{giwthmv1} and exhibit a sequence of smooth data approximating the data with a curvature delta singularity. In Sections \ref{estimates} and \ref{convergence}, we prove Theorem \ref{rdthmv2}. In Section \ref{estimates}, we prove that for smooth initial data satisfying the assumptions of Theorem \ref{rdthmv2}, a unique spacetime exists in a region depending only on the constants in the assumptions. This in particular implies that the approximating sequence of initial data constructed in Section \ref{initialcondition} gives rise to spacetimes with a common and uniform region of existence (identified by a choice of a double null coordinate system). In Section \ref{convergence}, we study the equations for the difference of two spacetimes. This allows us to show convergence of solutions with converging initial data and conclude the existence part of Theorem \ref{rdthmv2}. In Section \ref{limit}, we examine the regularity of the limiting spacetime and show that it is a solution to the vacuum Einstein equations. In Section \ref{uniquenesssec}, we conclude the uniqueness part of Theorem \ref{rdthmv2}. Finally, in Section \ref{limitgiw}, we return to the proof of part (b) of Theorem \ref{giwthmv1}, giving a precise description of the propagation of singularity.\\

\noindent{\bf Acknowledgments:} The authors would like to thank Mihalis Dafermos for valuable discussions. We also thank Dejan Gajic, Joe Keir, Jan Sbierski, Martin Taylor, as well as an anonymous referee, for helpful comments. J. Luk is supported by the NSF Postdoctoral Fellowship DMS-1204493. I. Rodnianski is supported by the NSF grant DMS-1001500 and the FRG grant DMS-1065710.

\section{Setting and Equations}

\begin{figure}[htbp]
\begin{center}
 
\input{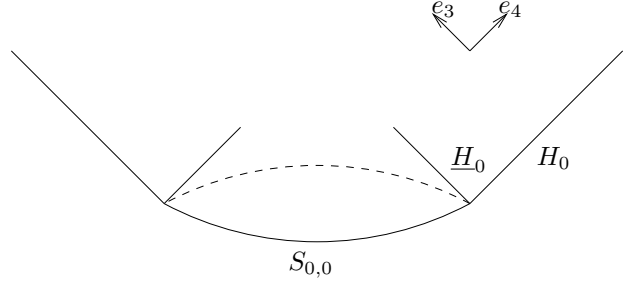}
 
\caption{The Basic Setup and the Null Frame}
\end{center}
\end{figure}

Our setting is the characteristic initial value problem with data given on the two characteristic hypersurfaces $H_0$ and $\Hb_0$ intersecting at the sphere $S_{0,0}$ (see Figure 4). The spacetime will be a solution to the Einstein equations constructed in a neighborhood of $S_{0,0}$ bounded by the two hypersurfaces.

\subsection{Double Null Foliation}
For a spacetime in a neighborhood of $S_{0,0}$, we define a double null foliation as follows: Let $u$ and $\ub$ be solutions to the eikonal equation
$$g^{\mu\nu}\partial_\mu u\partial_\nu u=0,\quad g^{\mu\nu}\partial_\mu\ub\partial_\nu \ub=0,$$
satisfying the initial conditions $u=0$ on $H_0$ and $\ub=0$ on $\Hb_0$.
Let
$$L'^\mu=-2g^{\mu\nu}\partial_\nu u,\quad \Lb'^\mu=-2g^{\mu\nu}\partial_\nu \ub.$$ 
These are null and geodesic vector fields. Define
$$2\Omega^{-2}=-g(L',\Lb').$$
Define
$$e_3=\Omega\Lb'\mbox{, }e_4=\Omega L'$$
to be the normalized null pair such that 
$$g(e_3,e_4)=-2$$
and
$$\Lb=\Omega^2\Lb'\mbox{, }L=\Omega^2 L'$$
to be the so-called equivariant vector fields.

We will denote the level sets of $u$ as $H_u$ and the level sets of $\ub$ and $\Hb_{\ub}$. By virtue of the eikonal equations, $H_u$ and $\Hb_{\ub}$ are null hypersurface. Notice that the sets defined by fixed values of $(u,\ub)$ are 2-spheres. We denote such spheres by $S_{u,\ub}$. They are intersections of the hypersurfaces $H_u$ and $\Hb_{\ub}$. The integral flows of $L$ and $\Lb$ respect the foliation $S_{u,\ub}$.

\subsection{The Coordinate System}\label{coordinates}
On a spacetime in a neighborhood of $S_{0,0}$, we define a coordinate system $(u,\ub,\th^1,\th^2)$ as follows:
On the sphere $S_{0,0}$, define a coordinate system $(\th^1,\th^2)$ for the sphere such that on each coordinate patch the metric $\gamma$ is smooth, bounded and positive definite. Then we define the coordinates on the initial hypersurfaces by requiring 
$$\frac{\partial}{\partial u}\th^A=0\mbox{ on $\Hb_0$, and }\frac{\partial}{\partial \ub}\th^A=0\mbox{ on $H_0$}.$$
We now define the coordinate system in the spacetime in a neighborhood of $S_{0,0}$ by letting $u$ and $\ub$ to be solutions to the eikonal equations:
$$g^{\mu\nu}\partial_\mu u\partial_\nu u=0,\quad g^{\mu\nu}\partial_\mu\ub\partial_\nu \ub=0,$$
and define $\th^1, \th^2$ by
$$\Ls_L \th^A=0,$$ 
where $\Ls_L$ denote the restriction of the Lie derivative to $TS_{u,\ub}$ (See \cite{Chr}).
Relative to the coordinate system, the null pair $e_3$ and $e_4$ can be expressed as
$$e_3=\Omega^{-1}\left(\frac{\partial}{\partial u}+b^A\frac{\partial}{\partial \th^A}\right), e_4=\Omega^{-1}\frac{\partial}{\partial \ub},$$
for some $b^A$ such that $b^A=0$ on $\Hb_0$, while the metric $g$ takes the form
$$g=-2\Omega^2(du\otimes d\ub+d\ub\otimes du)+\gamma_{AB}(d\th^A-b^Adu)\otimes (d\th^B-b^Bdu).$$ 

\subsection{The Equations}\label{sec.eqns}
We will recast the Einstein equations as a system for Ricci coefficients and curvature components associated to a null frame $e_3$, $e_4$ defined above and an orthonormal frame ${e_1,e_2}$ tangent to the 2-spheres $S_{u,\ub}$. Using the indices $A,B$ to denote $1,2$, we define the Ricci coefficients relative to the null fame:
 \begin{equation}
\begin{split}
&\chi_{AB}=g(D_A e_4,e_B),\, \,\, \quad \chib_{AB}=g(D_A e_3,e_B),\\
&\eta_A=-\frac 12 g(D_3 e_A,e_4),\quad \etab_A=-\frac 12 g(D_4 e_A,e_3)\\
&\omega=-\frac 14 g(D_4 e_3,e_4),\quad\,\,\, \omegab=-\frac 14 g(D_3 e_4,e_3),\\
&\zeta_A=\frac 1 2 g(D_A e_4,e_3)
\end{split}
\end{equation}
where $D_A=D_{e_{(A)}}$. We also introduce the  null curvature components,
 \begin{equation}
\begin{split}
\a_{AB}&=R(e_A, e_4, e_B, e_4),\quad \, \,\,   \ab_{AB}=R(e_A, e_3, e_B, e_3),\\
\b_A&= \frac 1 2 R(e_A,  e_4, e_3, e_4) ,\quad \bb_A =\frac 1 2 R(e_A,  e_3,  e_3, e_4),\\
\rho&=\frac 1 4 R(e_4,e_3, e_4,  e_3),\quad \sigma=\frac 1 4  \,^*R(e_4,e_3, e_4,  e_3)
\end{split}
\end{equation}
Here $\, ^*R$ denotes the Hodge dual of $R$.  We denote by $\nab$ the 
induced covariant derivative operator on $S_{u,\ub}$ and by $\nab_3$, $\nab_4$
the projections to $S_{u,\ub}$ of the covariant derivatives $D_3$, $D_4$, see
precise definitions in \cite{KN}. 

Observe that,
\begin{equation}
\begin{split}
&\omega=-\frac 12 \nab_4 (\log\Omega),\qquad \omegab=-\frac 12 \nab_3 (\log\Omega),\\
&\eta_A=\zeta_A +\nab_A (\log\Omega),\quad \etab_A=-\zeta_A+\nab_A (\log\Omega)
\end{split}
\end{equation}

Let $\phi^{(1)}\cdot\phi^{(2)}$ denote an arbitrary contraction of the tensor product of $\phi^{(1)}$ and $\phi^{(2)}$ with respect to the metric $\gamma$. We also define
$$(\phi^{(1)}\hot\phi^{(2)})_{AB}:=\phi^{(1)}_A\phi^{(2)}_B+\phi^{(1)}_B\phi^{(2)}_A-\gamma_{AB}(\phi^{(1)}\cdot\phi^{(2)}) \quad\mbox{for one forms $\phi^{(1)}_A$, $\phi^{(2)}_A$,}$$
$$\phi^{(1)}\wedge\phi^{(2)}:=\eps^{AB}(\gamma^{-1})^{CD}\phi^{(1)}_{AC}\phi^{(2)}_{BD}\quad\mbox{for symmetric two tensors $\phi^{(1)}_{AB}$, $\phi^{(2)}_{AB}$}.$$
For totally symmetric tensors, the $\div$ and $\curl$ operators are defined by the formulas
$$(\div\phi)_{A_1...A_r}:=\nabla^B\phi_{BA_1...A_r},$$
$$(\curl\phi)_{A_1...A_r}:=\eps^{BC}\nabla_B\phi_{CA_1...A_r},$$
where $\eps$ is the volume form associated to the metric $\gamma$.
Define also the trace to be
$$(\mbox{tr}\phi)_{A_1...A_{r-1}}:=(\gamma^{-1})^{BC}\phi_{BCA_1...A_{r-1}}.$$

We separate the trace and traceless part of $\chi$ and $\chib$. Let $\chih$ and $\chibh$ be the traceless parts of $\chi$ and $\chib$ respectively. Then $\chi$ and $\chib$ satisfy the following null structure equations:
\begin{equation}
\label{null.str1}
\begin{split}
\nab_4 \trch+\frac 12 (\trch)^2&=-|\chih|^2-2\omega \trch\\
\nab_4\chih+\trch \chih&=-2 \omega \chih-\alpha\\
\nab_3 \trchb+\frac 12 (\trchb)^2&=-2\omegab \trchb-|\chibh|^2\\
\nab_3\chibh + \trchb\,  \chibh&= -2\omegab \chibh -\alphab\\
\nab_4 \trchb+\frac1 2 \trch \trchb &=2\omega \trchb +2\rho- \chih\cdot\chibh +2\div \etab +2|\etab|^2\\
\nab_4\chibh +\frac 1 2 \trch \chibh&=\nab\widehat{\otimes} \etab+2\omega \chibh-\frac 12 \trchb \chih +\etab\widehat{\otimes} \etab\\
\nab_3 \trch+\frac1 2 \trchb \trch &=2\omegab \trch+2\rho- \chih\cdot\chibh+2\div \eta+2|\eta|^2\\
\nab_3\chih+\frac 1 2 \trchb \chih&=\nab\widehat{\otimes} \eta+2\omegab \chih-\frac 12 \trch \chibh +\eta\widehat{\otimes} \eta
\end{split}
\end{equation}
The other Ricci coefficients satisfy the following null structure equations:
\begin{equation}
\label{null.str2}
\begin{split}
\nabla_4\eta&=-\chi\cdot(\eta-\etab)-\b\\
\nabla_3\etab &=-\chib\cdot (\etab-\eta)+\bb\\
\nabla_4\omegab&=2\omega\omegab+\frac 34 |\eta-\etab|^2-\frac 14 (\eta-\etab)\cdot (\eta+\etab)-
\frac 18 |\eta+\etab|^2+\frac 12 \rho\\
\nabla_3\omega&=2\omega\omegab+\frac 34 |\eta-\etab|^2+\frac 14 (\eta-\etab)\cdot (\eta+\etab)- \frac 18 |\eta+\etab|^2+\frac 12 \rho\\
\end{split}
\end{equation}
The Ricci coefficients also satisfy the following constraint equations
\begin{equation}
\label{null.str3}
\begin{split}
\div\chih&=\frac 12 \nabla \trch - \frac 12 (\eta-\etab)\cdot (\chih -\frac 1 2 \trch) -\beta,\\
\div\chibh&=\frac 12 \nabla \trchb + \frac 12 (\eta-\etab)\cdot (\chibh-\frac 1 2   \trchb) +\betab\\
\curl\eta &=-\curl\etab=\sigma +\frac 1 2\chibh \wedge\chih\\
K&=-\rho+\frac 1 2 \chih\cdot\chibh-\frac 1 4 \trch \trchb
\end{split}
\end{equation}
with $K$ the Gauss curvature of the surfaces $S$.
The null curvature components satisfy the following null Bianchi equations:
\begin{equation}
\label{eq:null.Bianchi}
\begin{split}
&\nab_3\alpha+\frac 12 \trchb \alpha=\nabla\hot \beta+ 4\omegab\alpha-3(\chih\rho+^*\chih\sigma)+
(\zeta+4\eta)\hot\beta,\\
&\nab_4\beta+2\trch\beta = \div\alpha - 2\omega\beta +  \eta \alpha,\\
&\nab_3\beta+\trchb\beta=\nabla\rho + 2\omegab \beta +^*\nabla\sigma +2\chih\cdot\betab+3(\eta\rho+^*\eta\sigma),\\
&\nab_4\sigma+\frac 32\trch\sigma=-\div^*\beta+\frac 12\chibh\cdot ^*\alpha-\zeta\cdot^*\beta-2\etab\cdot
^*\beta,\\
&\nab_3\sigma+\frac 32\trchb\sigma=-\div ^*\betab+\frac 12\chih\cdot ^*\alphab-\zeta\cdot ^*\betab-2\eta\cdot 
^*\betab,\\
&\nab_4\rho+\frac 32\trch\rho=\div\beta-\frac 12\chibh\cdot\alpha+\zeta\cdot\beta+2\etab\cdot\beta,\\
&\nab_3\rho+\frac 32\trchb\rho=-\div\betab- \frac 12\chih\cdot\alphab+\zeta\cdot\betab-2\eta\cdot\betab,\\
&\nab_4\betab+\trch\betab=-\nabla\rho +^*\nabla\sigma+ 2\omega\betab +2\chibh\cdot\beta-3(\etab\rho-^*\etab\sigma),\\
&\nab_3\betab+2\trchb\betab=-\div\alphab-2\omegab\betab+\etab \cdot\alphab,\\
&\nab_4\alphab+\frac 12 \trch\alphab=-\nabla\hot \betab+ 4\omega\alphab-3(\chibh\rho-^*\chibh\sigma)+
(\zeta-4\etab)\hot \betab
\end{split}
\end{equation}
where $^*$ denotes the Hodge dual on $S_{u,\ub}$.

In the sequel, we will use capital Latin letters $A\in \{1,2\}$ for indices on the spheres $S_{u,\ub}$ and Greek letters $\mu\in\{1,2,3,4\}$ for indices in the whole spacetime.

In the following it will be useful to apply a schematic notation. We will let $\phi$ denote an arbitrary tensorfield, $\psi$ a Ricci coefficient and $\Psi$ a null curvature component different from $\alpha$. We will simply write $\psi\psi$ or $\psi^2$ (or $\psi\Psi$, etc.) to denote an arbitrary contraction. Moreover, we will denote by $\nab^i\psi^j$ the sum of all terms which are products of $j$ factors, with each factor being $\nab^{i_k}\psi$ and that the sum of all $i_k$'s being $i$, i.e., 
$$\nab^i\psi^j=\displaystyle\sum_{i_1+i_2+...+i_j}\underbrace{\nab^{i_1}\psi\nab^{i_2}\psi...\nab^{i_j}\psi}_\text{j factors}.$$
The use of the schematic notation is reserved for the cases when the precise nature of the contraction is not important to the argument. In particular, when using this schematic notation, we will neglect all constant factors.

\subsection{Integration and Norms}

Let $U$ be a coordinate patch on $S_{0,0}$ and $p_U$ be a partition of unity in $D_U$ such that $p_U$ is supported in $D_U$. Given a function $\phi$, the integration on $S_{u,\ub}$ is given by the formula:
$$\int_{S_{u,\ub}} \phi :=\sum_U \int_{-\infty}^{\infty}\int_{-\infty}^{\infty}\phi p_U\sqrt{\det\gamma}d\th^1 d\th^2.$$
Let $D_{u_*,\ub_*}$ by the region $0\leq u\leq u_*$, $0\leq \ub\leq \ub_*$. The integration on $D_{u,\ub}$ is given by the formula
\begin{equation*}
\begin{split}
\int_{D_{u,\ub}} \phi :=&\sum_U \int_0^u\int_0^{\ub}\int_{-\infty}^{\infty}\int_{-\infty}^{\infty}\phi p_U\sqrt{-\det g}d\th^1 d\th^2d\ub du\\
=&2\sum_U \int_0^u\int_0^{\ub}\int_{-\infty}^{\infty}\int_{-\infty}^{\infty}\phi p_U\Omega^2\sqrt{-\det \gamma}d\th^1 d\th^2d\ub du.
\end{split}
\end{equation*}
Since there are no canonical volume forms on $H_u$ and $\Hb_{\ub}$, we define integration by
$$\int_{H_{u}} \phi :=\sum_U \int_0^{\epsilon}\int_{-\infty}^{\infty}\int_{-\infty}^{\infty}\phi2 p_U\Omega\sqrt{\det\gamma}d\th^1 d\th^2d\ub,$$
and
$$\int_{H_{\ub}} \phi :=\sum_U \int_0^\epsilon\int_{-\infty}^{\infty}\int_{-\infty}^{\infty}\phi2p_U\Omega\sqrt{\det\gamma}d\th^1 d\th^2du.$$

With these definitions of integration, we can define the norms that we will use. Let $\phi$ be a tensorfield. For $1\leq p<\infty$, define
$$||\phi||_{L^p(S_{u,\ub})}^p:=\int_{S_{u,\ub}} <\phi,\phi>_\gamma^{p/2},$$
$$||\phi||_{L^p(H_u)}^p:=\int_{H_{u}} <\phi,\phi>_\gamma^{p/2},$$
$$||\phi||_{L^p(\Hb_{\ub})}^p:=\int_{\Hb_{\ub}} <\phi,\phi>_\gamma^{p/2}.$$
Define also the $L^\infty$ norm by
$$||\phi||_{L^\infty(S_{u,\ub})}:=\sup_{\th\in S_{u,\ub}} <\phi,\phi>_\gamma^{1/2}(\th).$$

\subsection{Precise Statement of the Main Theorem}\label{thmstatement}
With the notations introduced in this Section, we give a precise version of the statement of Theorem \ref{rdthmv1}:
\begin{theorem}\label{rdthmv2}
Suppose the initial data set for the characteristic initial value problem is given on $H_0$ for $0\leq \ub\leq \ub_*$ and on $\Hb_0$ for $0\leq u\leq u_*$ such that
$$c\leq |\det\gamma \restriction_{S_{u,0}} |, |\det\gamma \restriction_{S_{0,\ub}} |\leq C,$$
$$\sum_{i\leq 3}\left(|(\frac{\partial}{\partial\th})^i\gamma \restriction_{S_{u,0}}|+|(\frac{\partial}{\partial\th})^i\gamma \restriction_{S_{0,\ub}}|\right)\leq C,$$
$$\mathcal O_0:= \sum_{i\leq 3} \sup_{\ub}||\nabla^i\psi||_{L^2(S_{0,\ub})}+\sum_{i\leq 3} \sup_{u}||\nabla^i\psi||_{L^2(S_{u,0})}\leq C,$$
$$\mathcal R_0:=\sum_{i\leq 2}\left(\sum_{\Psi\in\{\beta,\rho,\sigma,\betab\}}\sup_{\ub}||\nab^i\Psi||_{L^2(S_{0,\ub})}+\sum_{\Psi\in\{\rho,\sigma,\betab,\alphab\}}\sup_u||\nab^i\Psi||_{L^2(S_{u,0})}\right)\leq C.$$
Then for $\epsilon$ sufficiently small depending only on $C$ and $c$, there exists a spacetime $(\mathcal M,g)$ endowed with a double null foliation $u$, $\ub$ that solves the characteristic initial value problem to the vacuum Einstein equations in $0\leq u\leq u_*$, $0\leq \ub\leq \ub_*$ for $u_*,\ub_*\leq \epsilon$. The metric is continuous and takes the form
$$g=-2\Omega^2(du\otimes d\ub+d\ub\otimes du)+ \gamma_{AB}(d\th^A-b^A du)\otimes(d\th^B-b^B du).$$
$(\mathcal M,g)$ is a $C^0$ limit of smooth solutions to the vacuum Einstein equations and is the unique spacetime solving the characteristic initial value problem among all $C^0$ limits of smooth solutions. Moreover, 
$$\frac{\partial}{\partial \th}g,\frac{\partial}{\partial u}g\in C^0_u C^0_{\ub} L^4(S),$$
$$\frac{\partial^2}{\partial \th^2}g,\frac{\partial^2}{\partial u\partial\th}g,\frac{\partial^2}{\partial u^2}g\in C^0_u C^0_{\ub} L^2(S),$$
$$\frac{\partial}{\partial \ub}g, \frac{\partial}{\partial\ub}((\gamma^{-1})^{AB}\frac{\partial}{\partial\ub}(\gamma)_{AB}) \in L^\infty_u L^\infty_{\ub} L^\infty(S),$$
$$\frac{\partial^2}{\partial \th \partial \ub}g,\frac{\partial^2}{\partial u\partial\ub}g,\frac{\partial^2}{\partial \ub^2}b^A\in L^\infty_u L^\infty_{\ub} L^4(S).$$
In the $(u,\ub,\th^1,\th^2)$ coordinates, the Einstein equations are satisfied in $L^\infty_uL^\infty_{\ub}L^2(S)$. Furthermore, the higher angular differentiability in the data results in higher angular differentiability.
\end{theorem}
We prove this Theorem in two steps. First, we show that if the initial data are smooth and the bounds for the initial data in the assumption of Theorem \ref{rdthmv2} hold, then an $\epsilon$ can be chosen depending only on the bounds in the assumption of Theorem \ref{rdthmv2} such that a smooth spacetime solving the Einstein equations exists in $0\leq u\leq\epsilon$ and $0\leq \ub\leq \epsilon$. This is formulated as Theorem \ref{timeofexistence}, and is proved in Section \ref{estimates}. Then, we show that a sequence of solutions to the Einstein equations with a converging sequence of smooth initial data, satisfying the assumptions of Theorem \ref{rdthmv2} with uniform constants $C$ and $c$, converges. The limit spacetime satisfies the Einstein equations and has the properties stated in Theorem \ref{rdthmv2}. This is formulated as Theorem \ref{convergencethm2} and is proved in Section \ref{convergence}. We furthermore show that the solution is unique and that the regularity in the angular directions persists in Section \ref{convergence}.

\section{The Initial Data}\label{initialcondition}

In this Section, we construct data satisfying the constraint equations such that $\chih$ has a jump discontinuity across a two sphere. We also construct a sequence of smooth data satisfying the constraint equations approaching the data with discontinuous $\chih$. We derive precise bounds for the Ricci coefficients and the curvature components for the initial data in this approximating sequence.

We fix
$$\Omega=1$$
identically on the initial hypersurfaces $H_0$ and $\Hb_0$. 
On $H_0$, $\gamma$ and $\chi$ have to satisfy the equations
\begin{equation}\label{con1}
\Ls_L \gamma=2\chi,\quad \Ls_L \trch= -\frac 12 (\trch)^2-|\chih|_\gamma^2,
\end{equation}
while on $\Hb_0$, $\gamma$ and $\chib$ have to satisfy the equations
\begin{equation}\label{con2}
\Ls_{\Lb}\gamma=2\chib,\quad\Ls_{\Lb} \trchb= -\frac 12 (\trchb)^2-|\chibh|_\gamma^2
\end{equation}
Here $\Ls$ denotes the restriction of the Lie derivative to $TS_{u,\ub}$.

Following \cite{Chr}, we can obtain initial data satisfying the above constraint equations by prescribing $\gamma_{AB}$, $\zeta_A$, $\trch$ and $\trchb$ on the two sphere $S_{0,0}$ and prescribing the conformal class of the metric $\hat{\gamma}_{AB}$ satisfying $\sqrt{\det \hat{\gamma}_{AB}}=1$ on each of the initial hypersurfaces. Relative to the coordinate system $(\th^1,\th^2)$, we require that on $S_{0,0}$,
$$\sum_{i\leq I+3}|(\frac{\partial}{\partial\th})^i(\gamma_{AB}, \zeta_A, tr\chi, tr\chib)|\leq C.$$
On the initial incoming hypersurface $\Hb_0$, in the coordinate system $(u,\th^1,\th^2)$ as in Section \ref{coordinates}, we require the conformal class of the metric $\hat{\gamma}_{AB}$ to be smooth, satisfying $\sqrt{\det \hat{\gamma}_{AB}}=1$ and obeying the estimates
$$\sum_{j\leq J+1}\sum_{i\leq I+3}|(\frac{\partial}{\partial u})^j(\frac{\partial}{\partial\th})^i\hat{\gamma}_{AB}|\leq C.$$

On the initial outgoing hypersurface $H_0$, in the coordinate system $(u,\th^1,\th^2)$ as in Section \ref{coordinates}, we now prescribe the conformal class of the metric $\hat{\gamma}_{AB}$ satisfying $\sqrt{\det \hat{\gamma}_{AB}}=1$ such that its $\ub$ derivative has a jump discontinuity. To this end we define smooth matrices $((\hat{\gamma})_1)_{AB}$ and $((\hat{\gamma})_2)_{AB}$, such that $((\hat{\gamma})_1)_{AB}$ is positive definite with determinant equals to $1$ and both $((\hat{\gamma})_1)_{AB}$ and $((\hat{\gamma})_2)_{AB}$ satisfy
$$\sum_{k\leq K+1}\sum_{i\leq I+3}|(\frac{\partial}{\partial\ub})^k(\frac{\partial}{\partial\th})^i((\hat{\gamma})_j)_{AB}|\leq C.$$
Fix $\ub_s\leq\frac{\epsilon}{2}$, where $\epsilon>0$ is a small parameter depending on the constants $C$ above and will be determine later. 
Let
$$\underline{\hat{\gamma}}_{AB}=((\hat{\gamma})_1)_{AB}+(\ub-\ub_s)((\hat{\gamma})_2)_{AB}\mathbbm 1_{\{\ub\geq\ub_s\}}.$$
For $\epsilon$ sufficiently small depending on $C$, $\underline{\hat{\gamma}}_{AB}$ is positive definite for $0\leq\ub\leq\epsilon$. We then define 
$$\hat{\gamma}_{AB}=\frac{1}{\sqrt{\det \underline{\hat{\gamma}}}}\underline{\hat{\gamma}}_{AB}.$$

According to the procedure in \cite{Chr}, there exists $\epsilon$ sufficiently small depending only on $C$ such that there exists initial data for $0\leq u\leq \epsilon$ on $\Hb_0$ satisfying \eqref{con2} and initial data for $0\leq \ub\leq \epsilon$ on $H_0$ obeying \eqref{con1}. We refer the readers to \cite{Chr} for details.

We note that according to \cite{Chr}, in order to obtain the initial data set on $H_0$, we need to solve for the conformal factor $\Phi$ defined by 
$$\gamma_{AB}=\Phi^2 \hat{\gamma}_{AB},$$
which obeys the ODE
\begin{equation}\label{PhiODE}
\frac{\partial^2\Phi}{\partial \ub^2}+\frac 18 ({\hat\gamma}^{-1})^{AC}({\hat\gamma}^{-1})^{BD}\frac{\partial}{\partial \ub}\hat\gamma_{AB}\frac{\partial}{\partial \ub}\hat\gamma_{CD}\Phi=0,
\end{equation}
with initial data
$$\Phi |_{S_{0,0}}=1.$$
The time of existence $\epsilon$ for this ODE depends only on the size of $({\hat\gamma}^{-1})^{AC}({\hat\gamma}^{-1})^{BD}\frac{\partial}{\partial \ub}\hat\gamma_{AB}\frac{\partial}{\partial \ub}\hat\gamma_{CD}$. Therefore, even though $\frac{\partial}{\partial \ub}\hat\gamma_{AB}$ is discontinuous, we can prescribe the discontinuity at $\ub=\ub_s$ such that $\ub_s=\frac{\epsilon}{2}$.

Given the conformal part of the metric $\hat{\gamma}$ and the conformal factor $\Phi$, we can identify
\begin{equation}\label{chihPhi}
\chih_{AB}=\frac{1}{2}\Phi^2\frac{\partial}{\partial \ub}\hat\gamma_{AB},
\end{equation}
and
\begin{equation}\label{trchPhi}
\trch=\frac{2}{\Phi}\frac{\partial\Phi}{\partial\ub}.
\end{equation}
With this identification, $\chih$ has a jump discontinuity at $\ub=\ub_s$ while $\trch$ is continuous.

\subsection{Approximation Procedure}

We now introduce an approximation procedure to construct a sequence of smooth initial data approaching the data described above. Consider a $C^\infty_0(\mathbb R)$ function $\tilde{h}_0$ that is supported in $[-1,1]$ and is identically 1 in $[-\frac{1}{2},\frac{1}{2}]$. Let
$$\tilde{h}(x)=\left\{\begin{array}{clcr}\tilde{h}_0(x)-\tilde{h}_0(2x)&x\ge 0\\0&x< 0\end{array}\right.$$
Note that $\tilde{h}$ is smooth. Now let
$$h_n(x)=\mathbbm 1_{\{x\geq 0\}}+\sum_{j=-\infty}^{n}\tilde{h}(2^j x).$$
We note that $h_n$ is supported in $\{x\geq 2^{-(n+1)}\}$ and $h'_n:=h_n-h_{n-1}$ is supported in $\{2^{-(n+1)}\leq x\leq 2^{-n}\}$. Moreover, $h_{n}\to\mathbbm 1_{\{x\geq 0\}}$ in $L^p$ for every $p<\infty$ as $n\to\infty$ and $x h_n\to x\mathbbm 1_{\{x\geq 0\}}$ in $L^p$ for every $p\leq\infty$ as $n\to\infty$. At the $n$-th step, we define 
$$(\gamma_n,\zeta_n,\trch_n)=(\gamma,\zeta,\trch)\mbox{  on $S_{0,0}$}$$
and
$$(\chibh_n,\trchb_n)=(\chibh,\trchb)\mbox{  on $\Hb_{0}$}$$
as before. Define
$$(\underline{\hat{\gamma}}_n)_{AB}=((\hat{\gamma})_1)_{AB}+(\ub-\ub_s)h_n(\ub-\ub_s)((\hat{\gamma})_2)_{AB}\mbox{  on $H_{0}$,}$$
and
$$(\hat{\gamma}_n)_{AB}=\frac{1}{\sqrt{\det (\underline{\hat{\gamma}})_n}}(\underline{\hat{\gamma}}_n)_{AB}\mbox{  on $H_{0}$}.$$
Notice that for $\epsilon$ sufficiently small, we have uniform upper bounds for each component of  $\hat{\underline{\gamma}}$ and their angular derivatives and a uniform lower bound for $\det\hat{\underline{\gamma}}$. Moreover, it is easy to see that for 
$${\hat{\gamma}}'_n:={\hat{\gamma}}_n-{\hat{\gamma}}_{n-1},$$
we have
$$\sum_{i\leq I+3}|(\frac{\partial}{\partial \ub})^k(\frac{\partial}{\partial\th})^i{\hat{\gamma}}'_n|\leq C2^{(-1+k)n},\quad\mbox{ for }\ub_s\leq \ub\leq\ub_s+2^{-n}, k\leq K+1, $$
$$(\frac{\partial}{\partial \ub})^k(\frac{\partial}{\partial\th})^i{\hat{\gamma}}'_n=0,\quad\mbox{ for }\ub\geq\ub_s+2^{-n},\mbox{ for $i\leq I+3$, $k\leq K+1$}, $$
Thus
$$(\int_0^\epsilon|(\frac{\partial}{\partial \ub})^k(\frac{\partial}{\partial\th})^i{\hat{\gamma}}'_n|^p d\ub)^{\frac 1p}\leq C2^{(-1+k-\frac 1p) n}\mbox{ for }k\leq K+1.$$
Moreover, we know that ${\hat{\gamma}}_n\to {\hat{\gamma}}$ in the sense that
$$(\int_0^\epsilon|(\frac{\partial}{\partial \ub})^k(\frac{\partial}{\partial\th})^i({\hat{\gamma}}_n-{\hat{\gamma}})|^p d\ub)^{\frac 1p}\leq C2^{(-1+k-\frac 1p) n}\mbox{ for }k\leq K+1.$$

With the definition of $(\hat{\gamma}_n)_{AB}$, we can solve (\ref{PhiODE}) for $\Phi_n$ and use (\ref{chihPhi}) and (\ref{trchPhi}) to solve for $\chih_n$ and $(\trch)_n$. This then allows us to use the null structure equations and the Bianchi identities to solve for all Ricci coefficients and curvature components, as well as their derivatives. In the following we will use the notation that a subscript $n$ denotes that metric component, Ricci coefficient or curvature component associated to $\hat{\gamma}_n$. Moreover, we will use the notation
$$\Phi'_n=\Phi_n-\Phi_{n-1}.$$
We now derive the bounds for each of the $\Phi_n$'s and $\Phi'_n$'s in the initial data. By uniqueness of the constraint ODEs, all 
$$\Phi'_n=0 \mbox{ for }\ub\leq\ub_s.$$
We therefore only derive bounds for $\ub_s\leq\ub\leq \epsilon$.

From (\ref{PhiODE}), we have
$$c\leq \Phi_n\leq C,\quad \sum_{k\leq K+2}\sum_{i\leq I+3}|(\frac{\partial}{\partial \ub})^k(\frac{\partial}{\partial\th})^i\Phi_n|\leq C$$
thus
$$\sum_{k\leq K}\sum_{i\leq I+3}|(\frac{\partial}{\partial \th})^i(\chih,\trch)_n|\leq C.$$
Then using the equation (\ref{PhiODE}) again, 
$$\sum_{k\leq K+2}\sum_{i\leq I+3}|(\frac{\partial}{\partial \ub})^k(\frac{\partial}{\partial \th})^i\Phi'_n|\leq C2^{(-2+k)n},\quad\mbox{ for }\ub_s\leq\ub\leq\ub_s+2^{-n},$$
and
$$(\frac{\partial}{\partial \ub})^k(\frac{\partial}{\partial \th})^i\Phi'_n=0,\quad\mbox{ for }\ub\geq\ub_s+2^{-n}\mbox{ and }2\leq k\leq K+2,\mbox{ for all }i\leq I+3.$$
Then by (\ref{chihPhi}) and (\ref{trchPhi}), we have for $\chih'_n$,
$$\sum_{k\leq K}\sum_{i\leq I+3}|(\frac{\partial}{\partial \ub})^k(\frac{\partial}{\partial \th})^i(\chih'_n)_{AB}|\leq C2^{kn},\quad\mbox{ for }\ub_s\leq\ub\leq\ub_s+2^{-n},$$
$$\sum_{k\leq K}\sum_{i\leq I+3}|(\frac{\partial}{\partial \ub})^k(\frac{\partial}{\partial \th})^i(\chih'_n)_{AB}|\leq C2^{-n},\quad\mbox{ for }\ub\geq\ub_s+2^{-n};$$
and for $\trch'_n$:
$$\sum_{k\leq K}\sum_{i\leq I+3}|(\frac{\partial}{\partial \ub})^k(\frac{\partial}{\partial \th})^i\trch'_n|\leq C2^{(-1+k)n},\quad\mbox{ for }\ub_s\leq\ub\leq\ub_s+2^{-n}$$
$$\sum_{k\leq K}\sum_{i\leq I+3}|(\frac{\partial}{\partial \ub})^k(\frac{\partial}{\partial \th})^i\trch'_n|\leq C2^{-n},\quad\mbox{ for }\ub\geq\ub_s+2^{-n}.$$
From the equation for $\zeta$ on $H_0$ together with the Codazzi equation, we have
$$\nabla_4\zeta=\div\chih-\frac 12 \nab tr\chi-\chih\zeta-\frac 32 \trch\zeta.$$
Notice that since $\Omega=1$ on $H_0$, we can write
$$\nab_4\zeta=\frac{\partial}{\partial \ub}\zeta+\chi\zeta.$$
Thus, we can integrate to get the estimate
$$\sum_{k\leq K+1}\sum_{i\leq I+2}|(\frac{\partial}{\partial \ub})^k(\frac{\partial}{\partial\th})^i\zeta_n|\leq C.$$
For the difference $\zeta'_n=\zeta_n-\zeta_{n-1}$,
$$\sum_{k\leq K+1}\sum_{i\leq I+2}|(\frac{\partial}{\partial \ub})^k(\frac{\partial}{\partial\th})^i\zeta'_n|\leq C2^{-n}.$$
Notice that we have used the control on the metric to relate $\div$ to the coordinates derivatives. Since on $H_0$,
$$\zeta=\eta=-\etab,$$
we have
$$\sum_{k\leq K+1}\sum_{i\leq I+2}|(\frac{\partial}{\partial \ub})^k(\frac{\partial}{\partial\th})^i(\eta_n,\etab_n)|\leq C.$$
and
$$\sum_{k\leq K+1}\sum_{i\leq I+2}|(\frac{\partial}{\partial \ub})^k(\frac{\partial}{\partial\th})^i(\eta'_n,\etab'_n)|\leq C2^{-n}.$$
By the equation of $\trchb$ on $H_0$
$$\nabla_4(\trchb)+\frac 12\trch \trchb=2\rho-\chih\cdot\chibh-2\div\zeta+|\zeta|^2_\gamma,$$
and the Gauss equation
$$K=-\rho+\frac 12 \chih\cdot\chibh-\frac 14 \trch\trchb,$$
we have
$$\nab_4(\trchb)+\trch \trchb=-2K-2\div\zeta+|\zeta|^2_\gamma.$$
Since the estimates for the Gauss curvature $K$ can be derived from the estimates for $\gamma$ and its derivatives alone, by the estimates derived above and the initial estimate for $\trchb$ on $S_{0,0}$,
$$\sum_{k\leq K+1}\sum_{i\leq I+1}|(\frac{\partial}{\partial\ub})^k(\frac{\partial}{\partial\th})^i\trchb_n|\leq C.$$
and
$$\sum_{k\leq K+1}\sum_{i\leq I+1}|(\frac{\partial}{\partial\ub})^k(\frac{\partial}{\partial\th})^i \trchb'_n|\leq C2^{-n}.$$
By the equation of $\chibh$ on $H_0$
$$\nabla_4 \chibh=-\frac{1}{2}\trch\chibh-\nabla\hot\zeta+\zeta\hot\zeta-\frac{1}{2}\trchb\chih,$$
we have
$$\sum_{k\leq K+1}\sum_{i\leq I+1}|(\frac{\partial}{\partial\ub})^k(\frac{\partial}{\partial\th})^i\chibh_n|\leq C.$$
and
$$\sum_{k\leq K+1}\sum_{i\leq I+1}|(\frac{\partial}{\partial\ub})^k(\frac{\partial}{\partial\th})^i\chibh'_n|\leq C2^{-n}.$$
By the following equation on $H_0$
$$\nabla_4\chih=-\trch \chih-\alpha,$$
we have
$$\sum_{k\leq K-1}\sum_{i\leq I+3}|(\frac{\partial}{\partial\ub})^k(\frac{\partial}{\partial\th})^i\alpha_n|\leq C2^{(1+k)n},\quad\mbox{ for }\ub_s\leq \ub\leq \ub_s+2^{-n},$$
$$\sum_{k\leq K-1}\sum_{i\leq I+3}|(\frac{\partial}{\partial\ub})^k(\frac{\partial}{\partial\th})^i\alpha_n|\leq C2^{-n},\quad\mbox{ for }\ub\geq \ub_s+2^{-n},$$
By the Codazzi equation on $H_0$
$$\div\chi=\frac{1}{2}\nabla \trch-\zeta\cdot(\chih-\frac{1}{2}\trch)-\beta,$$
we have
$$\sum_{k\leq K}\sum_{i\leq I+2}|(\frac{\partial}{\partial\ub})^k(\frac{\partial}{\partial\th})^i\beta_n|\leq C \mbox{ for }\ub_s\leq \ub\leq \ub_s+2^{-n},$$
$$\sum_{k\leq K}\sum_{i \leq I+2}|(\frac{\partial}{\partial\ub})^k(\frac{\partial}{\partial\th})^i\beta'_n|\leq C2^{-n} \mbox{ for }\ub\geq \ub_s+2^{-n}.$$
By the Gauss equation on $H_0$
$$K=-\rho+\frac{1}{2}\chih\cdot\chibh-\frac{1}{4}tr\chi tr\chib,$$
we have
$$\sum_{k\leq K}\sum_{i\leq I+1}|(\frac{\partial}{\partial\ub})^k(\frac{\partial}{\partial\th})^i\rho_n|\leq C \mbox{ for }\ub_s\leq \ub\leq \ub_s+2^{-n},$$
$$\sum_{k\leq K}\sum_{i\leq I+1}|(\frac{\partial}{\partial\ub})^k(\frac{\partial}{\partial\th})^i\rho'_n|\leq C2^{-n} \mbox{ for }\ub\geq \ub_s+2^{-n}.$$
By the equation on $H_0$
$$\curl \zeta=\sigma+\frac 12\chibh\wedge\chih,$$
we have
$$\sum_{k\leq K}\sum_{i\leq I+1}|(\frac{\partial}{\partial\ub})^k(\frac{\partial}{\partial\th})^i\sigma_n|\leq C \mbox{ for }\ub_s\leq \ub\leq \ub_s+2^{-n},$$
$$\sum_{k\leq K}\sum_{i\leq I+1}|(\frac{\partial}{\partial\th})^i\sigma'_n|\leq C2^{-n} \mbox{ for }\ub\geq \ub_s+2^{-n}.$$
On $S_{0,0}$, using the Codazzi equation
$$\div\chibh=\frac{1}{2}\nabla \trchb+\zeta\cdot(\chih-\frac{1}{2}\trch)+\beta,$$
we have
$$\sum_{i\leq I+2}|(\frac{\partial}{\partial\th})^i\betab_n|\leq C \mbox{ on $S_{0,0}$}.$$
By the following Bianchi equation on $H_0$
$$\nabla_4\betab=-\nab\rho+^*\nab\sigma-tr\chi\betab+2\chibh\cdot\beta-3(^*\zeta\sigma-\zeta\rho),$$
we have
$$\sum_{k\leq K+1}\sum_{i\leq I}|(\frac{\partial}{\partial\ub})^k(\frac{\partial}{\partial\th})^i\betab_n|\leq C,$$
$$\sum_{k\leq K+1}\sum_{i\leq I}|(\frac{\partial}{\partial\ub})^k(\frac{\partial}{\partial\th})^i\betab'_n|\leq C2^{-n}.$$
By the null structure equation on $H_0$
$$\nabla_4\omegab=3|\zeta|^2+\frac 12\rho,$$
we have 
$$\sum_{k\leq K+1}\sum_{i=0}^{I+1}|(\frac{\partial}{\partial\ub})^k(\frac{\partial}{\partial\th})^i\omegab_n|\leq C.$$
$$\sum_{k\leq K+1}\sum_{i=0}^{I+1}|(\frac{\partial}{\partial\ub})^k(\frac{\partial}{\partial\th})^i\omegab'_n|\leq C2^{-n}.$$
Together with $\omega=0$ which follows from $\omega=-\frac{1}{2}\nabla_4(\log\Omega)$, we have the bounds in $L^\infty$ for all the Ricci coefficients and null curvature components on $H_0$.

A similar procedure allows us to solve for the Ricci coefficients and the null curvature components on $\Hb_0$. Recall that for an impulsive gravitational wave, we take the data on $\Hb_0$ to be fixed, i.e., $\psi_n=\psi$ and $\Psi_n=\Psi$ on $\Hb_0$.

Finally, we see that on $H_0$, we can define $\nab_3^j\psi_n$ and $\nab_3^j\Psi_n$ for $j\leq J$ by taking the values of $\nab_3^j\psi_n$ and $\nab_3^j\Psi_n$ restricted to $S_{0,0}$ and solve the ODEs as above. Notice that they can be solved as long as $I$ is sufficiently large depending on $J$.

This allows us to conclude
\begin{proposition}\label{dataprop}
For every $I\geq 2$, $J,K\geq 0$, there exists characteristic initial data for an impulsive gravitational wave such that on $\Hb_0$:
\begin{equation}\label{data1}
\sum_{j\leq J}\sum_{k\leq K}\sum_{i\leq I+1}\sup_{\ub}||\nabla_3^j\nab_4^k\nabla^i\psi||_{L^2(S_{0,\ub})}\leq C_{I,J,K},
\end{equation}
\begin{equation}\label{data2}
\sum_{j\leq J}\sum_{k\leq K}\sum_{i\leq I}\sum_{\Psi\in\{\rho,\sigma,\betab,\alphab\}}\sup_u||\nabla_3^j\nab_4^k\nabla^i\Psi||_{L^2(S_{u,0})}\leq C_{I,J,K},\end{equation}
On $H_0$, $\chih$ has a jump discontinuity across $S_{0,\ub_s}$, and
\begin{equation}\label{data3}
\sum_{j\leq J}\sum_{i\leq I+1}\sup_{\ub}||\nab_3^j\nabla^i\psi||_{L^2(S_{0,\ub})}\leq C_{I,J},
\end{equation}
\begin{equation}\label{data4}
\sum_{j\leq J}\sum_{i\leq I}\sum_{\Psi\in\{\beta,\rho,\sigma,\betab\}}\sup_{\ub}||\nab_3^j\nabla^i\Psi||_{L^2(S_{0,\ub})}\leq C_{I,J},
\end{equation}
and for $H_0\setminus\{|\ub-\ub_s|\geq\eta\}$ for every $\eta>0$,
$$\sum_{j\leq J}\sum_{k\leq K}\sum_{i\leq I+1}\sup_{\ub}||\nab_3^j\nab_4^k\nabla^i\psi||_{L^2(S_{0,\ub})}\leq C_{I,J,K},$$
$$\sum_{j\leq J}\sum_{k\leq K}\sum_{i\leq I}\sum_{\Psi\in\{\beta,\rho,\sigma,\betab\}}\sup_{\ub}||\nab_3^j\nab_4^k\nabla^i\Psi||_{L^2(S_{0,\ub})}\leq C_{I,J,K},$$
and for $\epsilon$ sufficiently small, $\gamma$ is positive definite on $H_0$ and $\Hb_0$ with bounded angular derivatives up to the $(I+2)$-nd derivative. 

There exists a sequence of of smooth initial data approaching the characteristic initial data described above such that (\ref{data1})-(\ref{data4}) hold for $\psi_n$ and $\Psi_n$, 
and for $H_0\setminus\{\ub_s\leq\ub\leq \ub_s+2^{-n}\}$,
\begin{equation}\label{data5}
\sum_{j\leq J}\sum_{k\leq K}\sum_{i\leq I+1}\sup_{\ub}||\nab_3^j\nab_4^k\nabla^i\psi_n||_{L^2(S_{0,\ub})}\leq C_{I,J,K},
\end{equation}
\begin{equation}\label{data6}
\sum_{j\leq J}\sum_{k\leq K}\sum_{i\leq I}\sum_{\Psi\in\{\beta,\rho,\sigma,\betab\}}\sup_{\ub}||\nab_3^j\nab_4^k\nabla^i\Psi_n||_{L^2(S_{0,\ub})}\leq C_{I,J,K},
\end{equation}
and for $H_0\cap\{\ub_s\leq\ub\leq \ub_s+2^{-n}\}$,
\begin{equation}\label{data7}
\sum_{j\leq J}\sum_{i\leq I}\sup_{\ub}||\nab_3^j\nabla^i\alpha_n||_{L^2(S_{0,\ub})}\leq C_{I,J,K}2^n.
\end{equation}
This sequence of data converges in the sense that on $\Hb_0$, the data is fixed; and on $H_0$, for $\psi'_n=\psi_n-\psi_{n-1}$ and $\Psi'_n=\Psi_n-\Psi_{n-1}$,
\begin{equation}\label{data8}
\sum_{j\leq J}\sum_{i\leq I+1}\sup_{\ub}||\nab_3^j\nabla^i\psi'_n||_{L^2(S_{0,\ub})}\leq C_{I,J}2^{-n},
\end{equation}
\begin{equation}\label{data9}
\sum_{j\leq J}\sum_{i\leq I}\sum_{\Psi\in\{\beta,\rho,\sigma,\betab\}}\sup_{\ub}||\nab_3^j\nabla^i\Psi'_n||_{L^2(S_{0,\ub})}\leq C_{I,J}2^{-n},
\end{equation}
on $H_0\setminus\{\ub_s\leq\ub\leq \ub_s+2^{-n}\}$,
\begin{equation}\label{data10}
\sum_{j\leq J}\sum_{k\leq K}\sum_{i\leq I+1}\sup_{\ub}||\nab_3^j\nab_4^k\nabla^i\psi'_n||_{L^2(S_{0,\ub})}\leq C_{I,J,K}2^{-n},
\end{equation}
\begin{equation}\label{data11}
\sum_{j\leq J}\sum_{k\leq K}\sum_{i\leq I}\sum_{\Psi\in\{\beta,\rho,\sigma,\betab\}}\sup_{\ub}||\nab_3^j\nab_4^k\nabla^i\Psi'_n||_{L^2(S_{0,\ub})}\leq C_{I,J,K}2^{-n}.
\end{equation}
\end{proposition}
\begin{remark}
In the context of Theorem \ref{rdthmv2}, corresponding sequences of $\psi_n$ and $\Psi_n$ can be constructed. As opposed to the case of an impulsive gravitational wave, such sequences will satisfy (\ref{data1})-(\ref{data4}) and (\ref{data8})-(\ref{data9}), but they will not obey the improved estimates of Proposition \ref{dataprop} which holds on $H_0\setminus\{\ub_s\leq\ub\leq\ub_s+2^{-n}\}$. This additional regularity is only relevant for the argument that the impulsive gravitational wave is smooth away from $\Hb_{\ub_s}$.
\end{remark}
Proposition \ref{dataprop} in particular constructs the set of initial data satisfying the assumptions of Theorem \ref{giwthmv1}. In the following, we will focus on the case $I=2$, $J=K=0$. This level of regularity is sufficient to prove all the estimates and to construct a spacetime metric satisfying the Einstein equations.

\section{A Priori Estimates}\label{estimates}

In this section, we begin the proof of Theorem \ref{rdthmv2}. We show that for any smooth characteristic initial data satisfying the estimates in the assumption of Theorem \ref{rdthmv2}, there is an $\epsilon$ depending only on the constants in the estimates in the assumptions such that the unique resulting spacetime admits a double null foliation $(u, \ub)$ and remains smooth with precise estimates on its Ricci coefficients and curvature components in $0\leq u\leq u_*$ and $0\leq \ub\leq\ub_*$ for $u, \ub\leq \epsilon$. In particular, no assumptions have to be made on the size of $\alpha$.

In the context of Theorem \ref{rdthmv2}, our Theorem in this Section asserts that a region of existence can be found independent of $n$ for any smooth approximating sequence of data satisfying the bounds in the assumptions of Theorem \ref{rdthmv2}.

We first define the norms that we will work with. Let
$$\mathcal R=\sum_{i\leq 2}\left(\sum_{\Psi\in\{\beta,\rho,\sigma,\betab\}}\sup_u||\nabla^i\Psi||_{L^2(H_u)} +\sum_{\Psi\in\{\rho,\sigma,\betab,\alphab\}}\sup_{\ub}||\nabla^i\Psi||_{L^2(\underline{H}_{\ub})}\right),$$
$$\mathcal O_{i,p}=\sup_{u,\ub}||\nabla^i\psi||_{L^p(S_{u,\ub})},$$
and
$$\tilde{\mathcal O}_{i,2}=\sum_{\psi_H\in\{\trch,\chih,\eta,\etab,\omega,\trchb\}}\sup_u||\nabla^i\psi_H||_{L^2(H_u)}+ \sum_{\psi_{\Hb}\in\{\trch,\eta,\etab,\omegab,\trchb,\chibh\}}\sup_{\ub}||\nabla^i\psi_{\Hb}||_{L^2(\Hb_{\ub})}.$$
We write
$$\mathcal O=\sum_{i=0}^{2} \mathcal O_{i,2}+\sum_{i=0}^{1}\mathcal O_{i,4}+\mathcal O_{0,\infty}.$$
The following is the main theorem in this section:
\begin{theorem}\label{timeofexistence}
Suppose the initial data set for the characteristic initial value problem is smooth and satisfies
$$c\leq |\det\gamma \restriction_{S_{u,0}} |\leq C,\quad \sum_{i\leq 3}|(\frac{\partial}{\partial\th})^i\gamma \restriction_{S_{u,0}}|\leq C,$$
$$\mathcal O_0:= \sum_{i\leq 2} \sup_{\ub}||\nabla^i\psi||_{L^2(S_{0,\ub})}+\sum_{i\leq 2} \sup_{u}||\nabla^i\psi||_{L^2(S_{u,0})}+||\nabla^3\psi||_{L^2(H_0)}+||\nabla^3\psi||_{L^2(\Hb_0)}\leq C,$$
$$\mathcal R_0:=\sum_{i\leq 2}\left(\sum_{\Psi\in\{\beta,\rho,\sigma,\betab\}}||\nab^i\Psi||_{L^2(H_0)}+\sum_{\Psi\in\{\rho,\sigma,\betab,\alphab\}}||\nab^i\Psi||_{L^2(\Hb_0)}\right)\leq C.$$
Then for $\epsilon$ sufficiently small depending only on $C$ and $c$, there exists a unique spacetime $(\mathcal M, g)$ endowed with a double null foliation $(u,\ub)$ that solves the Einstein equations in the region $0\leq u\leq u_*$, $0\leq \ub\leq \ub_*$, whenever $u_*, \ub_*\leq \epsilon$. Moreover, in this region, the following norms are bounded above by a constant $C'$ depending only on $C$ and $c$:
$$\mathcal O, \tilde{\mathcal O}_{3,2}, \mathcal R < C'.$$
\end{theorem}
It follows by \cite{Rendall} that for every smooth initial data set, a unique smooth spacetime satisfying the Einstein equations exists in a small neighborhood of $S_{0,0}$. In order to show that the size of this region only depends on the constants in the assumption of Theorem \ref{timeofexistence}, it suffices to establish a priori control in the region $0\leq u\leq\epsilon$, $0\leq \ub\leq \epsilon$. The Theorem then follows from a standard last slice argument (See \cite{L}).

In this section, all estimates will be proved under the following bootstrap assumption:
\begin{equation}\tag{A}\label{BA1}
\mathcal O_{0,\infty}\leq \Delta_0.
\end{equation}

We now outline the steps in proving a priori estimates. In Section \ref{metric}, we will first estimate the metric components under the bootstrap assumption (\ref{BA1}). In Section \ref{transportsec}-\ref{elliptic}, we will derive some preliminary estimates and formulae. In Section \ref{transportsec}, we provide Propositions which gives $L^p$ estimates for general quantities satisfying transport equations. They will be used to control the Ricci coefficients and curvature components. In Section \ref{Embedding}, we establish Sobolev Embedding Theorems in our setting. In Section \ref{commutation}, we state some formulae for the commutators $[\nabla_4,\nabla]$ and $[\nabla_3,\nabla]$. They will then be used to obtain higher order estimates for general quantities satisfying transport equations. In Section \ref{elliptic}, we prove elliptic estimates for general quantities obeying Hodge systems. 

After these preliminary estimates, we estimate the Ricci coefficients and the curvature components. This proceeds in three steps:\\

\noindent {\bf STEP 1}: In Section \ref{Riccisec}, we prove that $\mathcal R<\infty$ and $\tilde{\mathcal O}_{3,2}<\infty$ together imply that $\mathcal O\leq C(\mathcal O_0)$, where $C(\mathcal O_0)$ is a constant depending only on the initial norm $\mathcal O_0$. These estimates are proved via the null structure equations.\\

\noindent {\bf STEP 2}: In Section \ref{Ricciellipticsec}, we use elliptic estimates to prove that $\mathcal R<\infty$ implies $\tilde{\mathcal O}_{3,2}\leq C(\mathcal O_0,\mathcal R_0,\mathcal R)$, where $C(\mathcal O_0,\mathcal R_0,\mathcal R)$ is a constant depending both on the initial norm and $\mathcal R$. \\

\noindent {\bf STEP 3}: In Section \ref{energyestimatessec}, we derive the energy estimates and prove the boundedness of $\mathcal R$ by the initial data.
\\

These steps provide bounds for the Ricci coefficients up to three angular derivatives and the curvature components up to two angular derivatives depending only on the constants in the assumptions of Theorem \ref{timeofexistence}. In Section \ref{Propregsec}, we show that this implies that higher regularity propagates and that the spacetime remains smooth in a region depending only on the constants in the assumptions of Theorem \ref{timeofexistence}. This allows us to conclude the proof of Theorem \ref{timeofexistence} in Section \ref{EndofProof}.

Finally, in Section \ref{AddEst}, we return to the sequence of smooth spacetimes arising from the sequence of data approximating the data of an impulsive gravitational wave as described in Section \ref{initialcondition}. We show that these spacetimes obey estimates in addition to those given in Theorem \ref{timeofexistence}. These bounds will be applied in Section \ref{limitgiw} to establish the regularity and singularity properties of an impulsive gravitational spacetime.

\subsection{Estimates for Metric Components}\label{metric}
We first show that we can control $\Omega$ under the bootstrap assumption (\ref{BA1}):
\begin{proposition}\label{Omega}
There exists $\epsilon_0=\epsilon_0(\Delta_0)$ such that for every $\epsilon\leq\epsilon_0$,
$$\frac 12\leq \Omega\leq 2.$$
\end{proposition}
\begin{proof}
Consider the equation
\begin{equation}\label{Omegatransport}
 \omega=-\frac{1}{2}\nabla_4\log\Omega=\frac{1}{2}\Omega\nabla_4\Omega^{-1}=\frac{1}{2}\frac{\partial}{\partial \ub}\Omega^{-1}.
\end{equation}
Fix $\ub$. Notice that both $\omega$ and $\Omega$ are scalars and therefore the $L^\infty$ norm is independent of the metric. We can integrate equation (\ref{Omegatransport}) using the fact that $\Omega^{-1}=1$ on $H_0$ to obtain
$$||\Omega^{-1}-1||_{L^\infty(S_{u,\ub})}\leq C\int_0^{u}||\omega||_{L^\infty(S_{u',\ub})}du\leq C\Delta_0\epsilon.$$
This implies both the upper and lower bounds for $\Omega$ for sufficiently small $\epsilon$.
\end{proof}

We then show that we can control $\gamma$ under the bootstrap assumption (\ref{BA1}):
\begin{proposition}\label{gamma}
Consider a coordinate patch $U$ on $S_{0,0}$ and define $U_{u,0}$ to be a coordinate patch on $S_{u,0}$ given by the one-parameter diffeomorphism generated by $\Lb$. Define $U_{u,\ub}$ to be the image of $U_{u,0}$ under the one-parameter diffeomorphism generated by $L$. Define also $D_U=\bigcup_{0\leq u\leq \epsilon ,0\leq \ub\leq \epsilon} U_{u,\ub}$. For $\epsilon$ small enough depending on initial data and $\Delta_0$, there exists $C$ and $c$ depending only on initial data such that the following pointwise bounds hold for $\gamma$ in $\mathcal D_U$:
$$c\leq \det\gamma\leq C. $$
Moreover, in $D_U$,
$$|\gamma_{AB}|,|(\gamma^{-1})^{AB}|\leq C.$$
\end{proposition}
\begin{proof}
The first variation formula states that
$$\Ls_L\gamma=2\Omega\chi.$$
In coordinates, this means
$$\frac{\partial}{\partial \ub}\gamma_{AB}=2\Omega\chi_{AB}.$$
From this we derive that 
$$\frac{\partial}{\partial \ub}\log(\det\gamma)=\Omega\trch.$$
Define $\gamma_0(u,\ub,\th^1,\th^2)=\gamma(0,\ub,\th^1,\th^2)$. 
$$|\det\gamma-\det(\gamma_0)|\leq C\int_0^{\ub}|\trch|d\ub'\leq C\Delta_0\epsilon.$$
This implies that the $\det \gamma$ is bounded above and below. Let $\Lambda$ be the larger eigenvalue of $\gamma$. Clearly,
\begin{equation}\label{La}
\Lambda\leq C\sup_{A,B=1,2}\gamma,
\end{equation}
and 
$$\sum_{A,B=1,2}|\chi_{AB}|^2\leq C\Lambda ||\chi||_{L^\infty(S_{u,\ub})}.$$
Then
$$|\gamma_{AB}-(\gamma_0)_{AB}|\leq C\int_0^{\ub}|\chi_{AB}|d\ub'\leq C\Lambda\Delta_0\epsilon.$$
Using the upper bound (\ref{La}), we thus obtain the upper bound for $|\gamma_{AB}|$. The upper bound for $|(\gamma^{-1})^{AB}|$ follows from the upper bound for $|\gamma_{AB}|$ and the lower bound for $\det\gamma$.
\end{proof}

A consequence of the previous Proposition is an estimate on the surface area of each two sphere $S_{u,\ub}$.
\begin{proposition}\label{area}
$$\sup_{u,\ub}|\mbox{Area}(S_{u,\ub})-\mbox{Area}(S_{u,0})|\leq C\Delta_0\epsilon.$$
\end{proposition}
\begin{proof}
This follows from the fact that $\sqrt{\det\gamma}$ is pointwise only slightly perturbed if $\epsilon$ is chosen to be appropriately small.
\end{proof}
With the estimate on the volume form, we can now show that the $L^p$ norms defined with respect to the metric and the $L^p$ norms defined with respect to the coordinate system are equivalent.
\begin{proposition}\label{eqnorm}
Given a covariant tensor $\phi_{A_1...A_r}$ on $S_{u,\ub}$, we have
$$\int_{S_{u,\ub}} <\phi,\phi>_{\gamma}^{p/2} \sim \sum_{i=1}^r\sum_{A_i=1,2}\iint |\phi_{A_1...A_r}|^p \sqrt{\det\gamma} d\th^1 d\th^2.$$
\end{proposition}
We can also control $b$ under the bootstrap assumption, thus controlling the full spacetime metric: 
\begin{proposition}\label{b}
In the coordinate system $(u,\ub,\th^1,\th^2)$,
$$|b^A|\leq C\Delta_0\epsilon.$$
\end{proposition}
\begin{proof}
$b^A$ satisfies the equation
$$\frac{\partial b^A}{\partial \ub}=-4\Omega^2\zeta^A.$$
This can be derived from 
$$[L,\Lb]=\frac{\partial b^A}{\partial \ub}\frac{\partial}{\partial \th^A}.$$
Now, integrating and using Proposition \ref{eqnorm} gives the result.
\end{proof}

\subsection{Estimates for Transport Equations}\label{transportsec}
The estimates for the Ricci coefficients and the null curvature components are derived from the null structure equations and the null Bianchi equations respectively. In order to use the equations, we need a way to obtain estimates from the null transport type equations. Such estimates require the boundedness of $\trch$ and $\trchb$, which is consistent with our bootstrap assumption (\ref{BA1}). Below, we state two Propositions which provide $L^p$ estimates for general quantities satisfying transport equations either in the $e_3$ or $e_4$ direction.
\begin{proposition}\label{transport}
There exists $\epsilon_0=\epsilon_0(\Delta_0)$ such that for all $\epsilon \leq \epsilon_0$ and for every $2\leq p<\infty$, we have
\[
 ||\phi||_{L^p(S_{u,\ub})}\leq C(||\phi||_{L^p(S_{u,\ub'})}+\int_{\ub'}^{\ub} ||\nabla_4\phi||_{L^p(S_{u,\ub''})}d{\ub''}),
\]
\[
 ||\phi||_{L^p(S_{u,\ub})}\leq C(||\phi||_{L^p(S_{u',\ub})}+\int_{u'}^{u} ||\nabla_3\phi||_{L^p(S_{u'',\ub})}d{u''}).
\]
\end{proposition}

\begin{proof}

The following identity holds for any scalar $f$:
\[
 \frac{d}{d\ub}\int_{\mathcal S_{u,\ub}} f=\int_{\mathcal S_{u,\ub}} \left(\frac{df}{d\ub}+\Omega \trch f\right)=\int_{\mathcal S_{u,\ub}} \Omega\left(e_4(f)+ \trch f\right).
\]
Similarly, we have
\[
 \frac{d}{du}\int_{\mathcal S_{u,\ub}} f=\int_{\mathcal S_{u,\ub}} \Omega\left(e_3(f)+ \trchb f\right).
\]
Hence, taking $f=|\phi|_{\gamma}^p$, we have
\begin{equation}\label{Lptransport}
\begin{split}
 ||\phi||^p_{L^p(S_{u,\ub})}=&||\phi||^p_{L^p(S_{u,\ub'})}+\int_{\ub'}^{\ub}\int_{S_{u,\ub''}} p|\phi|^{p-2}\Omega\left(<\phi,\nabla_4\phi>_\gamma+ \frac{1}{p}\trch |\phi|^2_{\gamma}\right)d{\ub''}\\
 ||\phi||^p_{L^p(S_{u,\ub})}=&||\phi||^p_{L^p(S_{u',\ub})}+\int_{u'}^{u}\int_{S_{u'',\ub}} p|\phi|^{p-2}\Omega\left(<\phi,\nabla_3\phi>_\gamma+ \frac{1}{p}\trchb |\phi|^2_{\gamma}\right)d{u''}
\end{split}
\end{equation}
The Proposition is proved using Cauchy-Schwarz on the sphere and the $L^\infty$ bounds for $\Omega$ and $\trch$ ($\trchb$) which are provided by Proposition \ref{Omega} and the bootstrap assumption (\ref{BA1}) respectively.
\end{proof}
The above estimates also hold for $p=\infty$:
\begin{proposition}\label{transportinfty}
There exists $\epsilon_0=\epsilon_0(\Delta_0)$ such that for all $\epsilon \leq \epsilon_0$, we have
\[
 ||\phi||_{L^\infty(S_{u,\ub})}\leq C\left(||\phi||_{L^\infty(S_{u,\ub'})}+\int_{\ub'}^{\ub} ||\nabla_4\phi||_{L^\infty(S_{u,\ub''})}d{\ub''}\right)
\]
\[
 ||\phi||_{L^\infty(S_{u,\ub})}\leq C\left(||\phi||_{L^\infty(S_{u',\ub})}+\int_{u'}^{u} ||\nabla_3\phi||_{L^\infty(S_{u'',\ub})}d{u''}\right).
\]
\end{proposition}
\begin{proof}
This follows simply from integrating along the integral curves of $L$ and $\Lb$, and the estimate on $\Omega$ in Proposition \ref{Omega}.
\end{proof}

\subsection{Sobolev Embedding}\label{Embedding}
Under the bootstrap assumption (\ref{BA1}), the Sobolev Embedding hold on a 2-sphere $S_{u,\ub}$. 
\begin{proposition}\label{L4}
There exists $\epsilon_0=\epsilon_0(\Delta_0)$ such that as long as $\epsilon\leq \epsilon_0$, we have
$$||\phi||_{L^4(S_{u,\ub})}\leq C\sum_{i=0}^1||\nabla^i\phi||_{L^2(S_{u,\ub})}. $$
\end{proposition}
\begin{proof}
We first prove this for scalars. Given our coordinate system, the desired Sobolev Embedding follows from standard Sobolev Embedding Theorems and the lower and upper bound of the volume form. Thus, the proposition holds for scalars. Now, for $\phi$ being a tensor, let $f=\sqrt{|\phi|_{\gamma}^2+\delta^2}$. Then
$$||f||_{L^4(S_{u,\ub})}\leq C\left(||f||_{L^2(S_{u,\ub})}+||\frac{<\phi,\nabla\phi>_{\gamma}}{\sqrt{|\phi|_{\gamma}^2+\delta^2}}||_{L^2(S_{u,\ub})}\right)\leq C\left(||f||_{L^2(S_{u,\ub})}+||\nabla\phi||_{L^2(S_{u,\ub})}\right).$$
The Proposition can be achieved by sending $\delta\to 0$.
\end{proof}
We can also prove the Sobolev Embedding Theorem for the $L^\infty$ norm: 
\begin{proposition}\label{Linfty}
There exists $\epsilon_0=\epsilon_0(\Delta_0)$ such that as long as $\epsilon\leq \epsilon_0$, we have
$$||\phi||_{L^\infty(S_{u,\ub})}\leq C\left(||\phi||_{L^2(S_{u,\ub})}+||\nabla\phi||_{L^4(S_{u,\ub})}\right). $$
As a consequence,
$$||\phi||_{L^\infty(S_{u,\ub})}\leq C\sum_{i=0}^2||\nabla^i\phi||_{L^2(S_{u,\ub})}. $$
\end{proposition}
\begin{proof}
The first statement follows from coordinate considerations as in Proposition \ref{L4}. The second statement follows from applying the first and Proposition \ref{L4}.
\end{proof}

\subsection{Commutation Formulae}\label{commutation}
We have the following formula from \cite{KN}:
\begin{proposition}
The commutator $[\nabla_4,\nabla]$ acting on an $(0,r)$ S-tensor is given by
\begin{equation*}
 \begin{split}
[\nabla_4,\nabla_B]\phi_{A_1...A_r}=&[D_4,D_B]\phi_{A_1...A_r}+(\nabla_B\log\Omega)\nabla_4\phi_{A_1...A_r}-(\gamma^{-1})^{CD}\chi_{BD}\nabla_C\phi_{A_1...A_r} \\
&-\sum_{i=1}^r (\gamma^{-1})^{CD}\chi_{BD}\etab_{A_i}\phi_{A_1...\hat{A_i}C...A_r}+\sum_{i=1}^r (\gamma^{-1})^{CD}\chi_{A_iB}\etab_{D}\phi_{A_1...\hat{A_i}C...A_r}.
 \end{split}
\end{equation*}
Similarly, the commutator $[\nabla_3,\nabla]$ acting on an $(0,r)$ S-tensor is given by
\begin{equation*}
 \begin{split}
[\nabla_3,\nabla_B]\phi_{A_1...A_r}=&[D_3,D_B]\phi_{A_1...A_r}+(\nabla_B\log\Omega)\nabla_3\phi_{A_1...A_r}-(\gamma^{-1})^{CD}\chib_{BD}\nabla_C\phi_{A_1...A_r} \\
&-\sum_{i=1}^r (\gamma^{-1})^{CD}\chib_{BD}\eta_{A_i}\phi_{A_1...\hat{A_i}C...A_r}+\sum_{i=1}^r (\gamma^{-1})^{CD}\chib_{A_iB}\eta_{D}\phi_{A_1...\hat{A_i}C...A_r}.
 \end{split}
\end{equation*}
\end{proposition}

By induction, we get the following schematic formula for repeated commutations:
\begin{proposition}\label{commuteeqn}
Suppose $\nabla_4\phi=F_0$. Let $\nabla_4\nabla^i\phi=F_i$.
Then
\begin{equation*}
\begin{split}
F_i\sim &\sum_{i_1+i_2+i_3=i}\nabla^{i_1}(\eta,\underline{\eta})^{i_2}\nabla^{i_3} F_0+\sum_{i_1+i_2+i_3+i_4=i}\nabla^{i_1}(\eta,\underline{\eta})^{i_2}\nabla^{i_3}\chi\nabla^{i_4} \phi\\
&+\sum_{i_1+i_2+i_3+i_4=i-1} \nabla^{i_1}(\eta,\underline{\eta})^{i_2}\nabla^{i_3}\beta\nabla^{i_4} \phi.
\end{split}
\end{equation*}
where we have applied our schematic notation and by $\nabla^{i_1}(\eta,\underline{\eta})^{i_2}$ we mean the sum of all terms which is a product of $i_2$ factors, each factor being $\nabla^j \eta$ or $\nabla^j\underline{\eta}$ for some $j$ and that the sum of all $j$'s is $i_1$, i.e., $$\nabla^{i_1}(\eta,\underline{\eta})^{i_2}=\displaystyle\sum_{j_1+...+j_{i_2}=i_1}\sum_{\psi_1,...,\psi_{i_2}\in\{\eta,\etab\}}\nabla^{j_1}\psi_1...\nabla^{j_{i_2}}\psi_{i_2}.$$ Similarly, suppose $\nabla_3\phi=G_{0}$. Let $\nabla_3\nabla^i\phi=G_{i}$.
Then
\begin{equation*}
\begin{split}
G_{i}\sim &\sum_{i_1+i_2+i_3=i}\nabla^{i_1}(\eta,\underline{\eta})^{i_2}\nabla^{i_3} G_{0}+\sum_{i_1+i_2+i_3+i_4=i}\nabla^{i_1}(\eta,\underline{\eta})^{i_2}\nabla^{i_3}\underline{\chi}\nabla^{i_4} \phi\\
&+\sum_{i_1+i_2+i_3+i_4=i-1} \nabla^{i_1}(\eta,\underline{\eta})^{i_2}\nabla^{i_3}\underline{\beta}\nabla^{i_4} \phi.
\end{split}
\end{equation*}

\end{proposition}
\begin{proof}
The proof is by induction. We will prove it for the first statement. The proof of second one is analogous. The formula obviously holds for $i=0$. Assume that the statement is true for $i<i_0$.
\begin{equation*}
\begin{split}
F_{i_0}=&[\nabla_4,\nabla]\nabla^{i_0-1}\phi+\nabla F_{i_0-1}\\
\sim &\chi\nabla^{i_0}\phi+(\eta,\underline{\eta})\nabla_4\nabla^{i_0-1}\phi+\beta\nabla^{i_0-1}\phi+\chi(\eta,\underline{\eta})\nabla^{i_0-1}\phi\\
&+\sum_{i_1+i_2+i_3=i_0}\nabla^{i_1}(\eta,\underline{\eta})^{i_2}\nabla^{i_3} F_{0}\\
&+\sum_{i_1+i_2+i_3+i_4=i_0}\nabla^{i_1}(\eta,\underline{\eta})^{i_2}\nabla^{i_3}\chi\nabla^{i_4} \phi\\
&+\sum_{i_1+i_2+i_3+i_4=i_0-1} \nabla^{i_1}(\eta,\underline{\eta})^{i_2}\nabla^{i_3}\beta\nabla^{i_4} \phi.
\end{split}
\end{equation*}
First notice that the first, third and fourth terms are acceptable. Then plug in the formula for $F_{i_0-1}=\nabla_4\nabla^{i_0-1}\phi$ and get the result.
\end{proof}
The following further simplified version is useful for our estimates in the next section:
\begin{proposition}
Suppose $\nabla_4\phi=F_0$. Let $\nabla_4\nabla^i\phi=F_i$.
Then
\begin{equation*}
\begin{split}
F_i\sim &\sum_{i_1+i_2+i_3=i}\nabla^{i_1}\psi^{i_2}\nabla^{i_3} F_0+\sum_{i_1+i_2+i_3+i_4=i}\nabla^{i_1}\psi^{i_2}\nabla^{i_3}\chi\nabla^{i_4} \phi.\\
\end{split}
\end{equation*}
Similarly, suppose $\nabla_3\phi=G_{0}$. Let $\nabla_3\nabla^i\phi=G_{i}$.
Then
\begin{equation*}
\begin{split}
G_{i}\sim &\sum_{i_1+i_2+i_3=i}\nabla^{i_1}\psi^{i_2}\nabla^{i_3} G_{0}+\sum_{i_1+i_2+i_3+i_4=i}\nabla^{i_1}\psi^{i_2}\nabla^{i_3}\underline{\chi}\nabla^{i_4} \phi.
\end{split}
\end{equation*}
\end{proposition}
\begin{proof}
We replace $\beta$ and $\betab$ using the Codazzi equations, which schematically looks like
$$\beta=\nabla\chi+\psi\chi,$$
$$\betab=\nabla\chib+\psi\chib.$$
\end{proof}

\subsection{General Elliptic Estimates for Hodge Systems}\label{elliptic}
We recall the definition of the divergence and curl of a symmetric covariant tensor of an arbitrary rank:
$$(\div\phi)_{A_1...A_r}=\nabla^B\phi_{BA_1...A_r},$$
$$(\curl\phi)_{A_1...A_r}=\eps^{BC}\nabla_B\phi_{CA_1...A_r},$$
where $\eps$ is the volume form associated to the metric $\gamma$.
Recall also that the trace is defined to be
$$(tr\phi)_{A_1...A_{r-1}}=(\gamma^{-1})^{BC}\phi_{BCA_1...A_{r-1}}.$$
The following elliptic estimate is standard (See for example \cite{CK} or \cite{Chr}):
\begin{proposition}\label{ellipticthm}
Let $\phi$ be a totally symmetric $r+1$ covariant tensorfield on a 2-sphere $(\mathbb S^2,\gamma)$ satisfying
$$\div\phi=f,\quad \curl\phi=g,\quad \mbox{tr}\phi=h.$$
Suppose also that
$$\sum_{i\leq 1}||\nabla^i K||_{L^2(S)}< \infty.$$
Then for $i\leq 3$,
$$||\nabla^{i}\phi||_{L^2(S)}\leq C(\sum_{k\leq 1}||\nabla^k K||_{L^2(S)})(\sum_{j=0}^{i-1}(||\nabla^{j}f||_{L^2(S)}+||\nabla^{j}g||_{L^2(S)}+||\nabla^{j}h||_{L^4(S)}+||\phi||_{L^2(S)})).$$
\end{proposition}
For the special case that $\phi$ a symmetric traceless 2-tensor, we only need to know its divergence:
\begin{proposition}\label{elliptictraceless}
Suppose $\phi$ is a symmetric traceless 2-tensor satisfying
$$\div\phi=f.$$
Suppose moreover that $$\sum_{i\leq 1}||\nabla^i K||_{L^2(S)}< \infty.$$
Then, for $i\leq 3$,
$$||\nabla^{i}\phi||_{L^2(S)}\leq C(\sum_{k\leq 1}||\nabla^k K||_{L^2(S)})(\sum_{j=0}^{i-1}(||\nabla^{j}f||_{L^2(S)}+||\phi||_{L^2(S)})).$$
\end{proposition}
\begin{proof}
In view of Proposition \ref{ellipticthm}, this Proposition follows from
$$\curl\phi=^*f.$$
This is a direct computation using that fact that $\phi$ is both symmetric and traceless.
\end{proof}

\subsection{$L^p(S)$ Estimates for the Ricci Coefficients}\label{Riccisec}
We continue to work under the bootstrap assumptions (\ref{BA1}). In this Subsection, we pursue step 1 of the proof as outlined in the beginning of the Section. We show if that the curvature norm $\mathcal R$ and the third order Ricci coefficient norm $\tilde{\mathcal O}_{3,2}$ are bounded, then so is the Ricci coefficient norm $\mathcal O$. In particular, our bootstrap assumption (\ref{BA1}) and all the estimates in the last section are verified as long as $\mathcal R$ and $\tilde{\mathcal O}_{3,2}$ are controlled.
\begin{proposition}\label{Ricci}
Suppose
$$\mathcal R<\infty,\quad\tilde{\mathcal O}_{3,2}<\infty.$$
Then there exists $\epsilon_0=\epsilon_0(\mathcal O_0,\mathcal R,\tilde{\mathcal O}_{3,2})$ such that for $\epsilon\leq \epsilon_0$, 
\[
 \mathcal O\leq C(\mathcal O_0),
\]
where $\mathcal O_0$ is the initial norm for the Ricci coefficients defined in Theorem \ref{timeofexistence}.
\end{proposition}
\begin{proof}
To prove the Proposition, we need to estimate $\displaystyle\sum_{i\leq 1}{\mathcal O}_{i,4}$ and $\displaystyle\sum_{i\leq 2}{\mathcal O}_{i,2}$. First, we estimate the $\mathcal O_{i,4}$ norms. Using the null structure equations for all Ricci coefficients except $\chih$, we have
\begin{equation*}
\sum_{\psi\in\{\trchb,\chibh,\etab,\omega\}}\nabla_3\psi=\sum_{\Psi\in\{\rho,\sigma,\betab,\alphab\}}\Psi+\psi\psi,
\end{equation*}
and
\begin{equation*}
\sum_{\psi\in\{\trch,\eta,\omegab\}}\nabla_4\psi=\sum_{\Psi\in\{\beta,\rho,\sigma,\betab\}}\Psi+\psi\psi,
\end{equation*}
where the $\psi$ on the right hand sides denote any Ricci coefficients. We will also use the null structure equations commuted with angular derivatives:
\begin{equation}\label{Riccieqn}
\begin{split}
\sum_{\psi\in\{\trchb,\chibh,\etab,\omega\}}\nabla_3\nabla^i\psi=&\sum_{\Psi\in\{\rho,\sigma,\betab,\alphab\}}\sum_{i_1+i_2+i_3=i}\nabla^{i_1}\psi^{i_2}\nabla^{i_3}\Psi+\sum_{i_1+i_2+i_3+i_4=i}\nabla^{i_1}\psi^{i_2}\nabla^{i_3}\psi\nabla^{i_4}\psi,\\
\sum_{\psi\in\{\trch,\eta,\omegab\}}\nabla_4\nabla^i\psi=&\sum_{\Psi\in\{\beta,\rho,\sigma,\betab\}}\sum_{i_1+i_2+i_3=i}\nabla^{i_1}\psi^{i_2}\nabla^{i_3}\Psi+\sum_{i_1+i_2+i_3+i_4=i}\nabla^{i_1}\psi^{i_2}\nabla^{i_3}\psi\nabla^{i_4}\psi,
\end{split}
\end{equation}
where we have applied the notation defined at the end of Section \ref{sec.eqns} to represent the sum of products of derivatives of $\psi$ by $\nab^{i_1}\psi^{i_2}$.
By Sobolev Embedding in Proposition \ref{L4} and the bootstrap assumption (\ref{BA1}), we have
$$\sum_{i_1+i_2+i_3\leq 1}||\nabla^{i_1}\psi^{i_2}\nabla^{i_3}\Psi||_{L^4(S_{u,\ub})}\leq \Delta_0\sum_{i=0}^2||\nabla^i\Psi||_{L^2(S_{u,\ub})}.$$
\begin{equation}\label{Riccinl}
\sum_{i_1+i_2+i_3+i_4\leq 1}||\nabla^{i_1}\psi^{i_2}\nabla^{i_3}\psi\nabla^{i_4}\psi||_{L^4(S_{u,\ub})}\leq (\Delta_0+\Delta_0^2)\sum_{i=0}^1\mathcal O_{i,4}.
\end{equation}
Hence, by Proposition \ref{transport} applied to (\ref{Riccinl}), we obtain
\begin{equation}\label{RicciL41}
\begin{split}
&\sum_{i=0}^{1}\mathcal O_{i,4}[\trch,\eta,\etab,\omega,\omegab,\trchb,\chibh]\\
\leq &\mathcal O_0+C\epsilon^{\frac{1}{2}}\Delta_0\mathcal R+C(\Delta_0)\epsilon \sum_{i=0}^{1}\mathcal O_{i,4}.
\end{split}
\end{equation}
In order to close this estimate, we need to estimate $\chih$. We like to avoid the use of the equation $\nabla_4\chih=-\alpha+\psi\psi$ since we do not have an estimate for singular curvature component $\alpha$. Instead, we use the equation
$$\nabla_3\chih=\nabla\hot\eta+\psi\psi.$$
Commuting with angular derivatives, we get
\begin{equation}\label{Riccieqn2}
\nabla_{3}\nabla^i\chih=\sum_{i_1+i_2+i_3=i}\nabla^{i_1}\psi^{i_2}\nabla^{i_3+1}\eta+\sum_{i_1+i_2+i_3+i_4=i}\nabla^{i_1}\psi^{i_2}\nabla^{i_3}\psi\nabla^{i_4}\psi.
\end{equation}
Applying Proposition \ref{transport} to (\ref{Riccieqn2}) and using the estimates in (\ref{Riccinl}), we have
\begin{equation*}
\sum_{i=0}^1||\nabla^i\chih||_{L^4(S_{u,\ub})}\leq \mathcal O_0+C\int_0^u ||\nabla^2\eta||_{L^4(S_{u',\ub})}du'+C(\Delta_0)\epsilon \sum_{i=0}^{1}\mathcal O_{i,4}.
\end{equation*}
By Cauchy-Schwarz and the Sobolev Embedding in Proposition \ref{L4},
$$\int_0^u ||\nabla^2\eta||_{L^4(S_{u',\ub})}du' \leq C\epsilon \mathcal O_{2,2}+C\epsilon^{\frac 12} \tilde{\mathcal O}_{3,2}.$$
Therefore,
\begin{equation}\label{RicciL42}
\sum_{i=0}^1||\nabla^i\chih||_{L^4(S_{u,\ub})}\leq C\epsilon \mathcal O_{2,2}+C\epsilon^{\frac 12} \tilde{\mathcal O}_{3,2}+C(\Delta_0)\epsilon \sum_{i=0}^{1}\mathcal O_{i,4}.
\end{equation}
Putting (\ref{RicciL41}) and (\ref{RicciL42}) together, we get
\begin{equation}\label{RicciL4}
\begin{split}
\sum_{i=0}^{1}\mathcal O_{i,4}\leq \mathcal O_0+C\epsilon^{\frac{1}{2}}\Delta_0\mathcal R+C\epsilon \mathcal O_{2,2}+C\epsilon^{\frac 12}\tilde{\mathcal O}_{3,2}+ C(\Delta_0)\epsilon \sum_{i=0}^{1}\mathcal O_{i,4}. \\
\end{split}
\end{equation}
By choosing $\epsilon$ sufficiently small depending on $\Delta_0$, $\mathcal O_0$, $\mathcal R$ and $\tilde{\mathcal O}_{3,2}$, we have
$$\sum_{i=0}^{1}\mathcal O_{i,4}\leq C(\mathcal O_0)+C\epsilon \mathcal O_{2,2},$$
where $C(\mathcal O_0)$ does not depend on $\Delta_0$, once $\epsilon$ is chosen to be small enough.

We now move to estimate the $\mathcal O_{i,2}$ norms. Applying Proposition \ref{transport} to (\ref{Riccieqn}), we obtain
\begin{equation}\label{RicciL21}
\begin{split}
&\sum_{i=0}^{2}\mathcal O_{i,2}[\trch,\eta,\etab,\omega,\omegab,\trchb,\chibh]\\
\leq &\mathcal O_0+C\epsilon^{\frac{1}{2}}\Delta_0\mathcal R+\epsilon\Delta_0 \sum_{i=0}^{2}\mathcal O_{i,2}+\epsilon(\sum_{i=0}^{1}\mathcal O_{i,4})^2.
\end{split}
\end{equation}
For $\chih$, we apply Proposition \ref{transport} to (\ref{Riccieqn2}) to get
\begin{equation}\label{RicciL22}
\sum_{i=0}^2||\nabla^i\chih||_{L^2(S_{u,\ub})}\leq C(\mathcal O_0)+ C\epsilon^{\frac{1}{2}}\tilde{\mathcal O}_{3,2}+C\epsilon\Delta_0 \sum_{i=0}^{2}\mathcal O_{i,2}+\epsilon(\sum_{i=0}^{1}\mathcal O_{i,4})^2.
\end{equation}
(\ref{RicciL21}) and (\ref{RicciL22}) together give
$$\sum_{i=0}^2\mathcal O_{i,2}\leq \mathcal O_0+C\epsilon^{\frac{1}{2}}\Delta_0\mathcal R+\epsilon\Delta_0 \sum_{i=0}^{2}\mathcal O_{i,2}+\epsilon(\sum_{i=0}^{1}\mathcal O_{i,4})^2.$$
Choosing $\epsilon$ to be appropriately small depending on $\Delta_0$, $\mathcal O_0$, $\mathcal R$ and $\tilde{\mathcal O}_{3,2}$, we can absorb $\displaystyle\epsilon\Delta_0 \sum_{i=0}^{2}\mathcal O_{i,2}$ to the left hand side to obtain
\begin{equation}\label{RicciL2}
\sum_{i=0}^2\mathcal O_{i,2}\leq C(\mathcal O_0)+\epsilon(\sum_{i=0}^{1}\mathcal O_{i,4})^2.
\end{equation}
For $\epsilon$ sufficiently small, (\ref{RicciL4}) and (\ref{RicciL2}) imply that
$$\sum_{i=0}^1\mathcal O_{i,4}+\sum_{i=0}^2\mathcal O_{i,2}\leq C(\mathcal O_0).$$

Finally, we use Sobolev embedding in Proposition \ref{Linfty} to get the necessary bounds for $\mathcal O_{0,\infty}$ from the control for $\displaystyle\sum_{i=0}^1\mathcal O_{i,4}$. Moreover, $\Delta_0$ can be chosen so that $\Delta_0\gg C(\mathcal O_0)$ and the bootstrap assumption (\ref{BA1}) is verified as long as $\mathcal R$ and $\tilde{\mathcal O}_{3,2}$ are controlled. Furthermore, the since the choice of $\Delta_0$ depend only on $\mathcal O_0$, the $\epsilon$ can be set to depend only on $\mathcal O_0$, $\mathcal R$ and $\tilde{\mathcal O}_{3,2}$. 
\end{proof}

\subsection{$L^2(S)$ Estimates for Curvature}
In this Subsection, we prove that $\beta,\rho,\sigma,\betab$ and their first derivatives can be controlled in $L^2(S)$. The argument relies on Proposition \ref{transport} and the Bianchi equations (\ref{eq:null.Bianchi}). The results of this Section will give the bounds for the Gauss curvature and its first derivatives appearing in Proposition \ref{ellipticthm}.
\begin{proposition}\label{RS}
Suppose
$$\mathcal R<\infty,\quad\tilde{\mathcal O}_{3,2}<\infty.$$
Then there exists $\epsilon_0=\epsilon_0(\mathcal O_0,\mathcal R,\tilde{\mathcal O}_{3,2})$ such that for $\epsilon\leq \epsilon_0$, 
$$\sum_{i\leq 1}\sum_{\Psi\in\{\beta,\rho,\sigma,\betab\}}||\nabla^i\Psi||_{L^2(S_{u,\ub})}\leq C(\mathcal R_0).$$
\end{proposition}
\begin{proof}
By the Bianchi equations, we have
$$\nabla_3(\beta,\rho,\sigma,\betab)=\sum_{\Psi\in\{\rho,\sigma,\betab,\alphab\}}\nabla\Psi+\psi\sum_{\Psi\in\{\beta,\rho,\sigma,\betab,\alphab\}}\Psi$$
Commuting the above with $\nabla$, we have
$$\nabla_3\nabla(\beta,\rho,\sigma,\betab)=\sum_{\Psi\in\{\rho,\sigma,\betab,\alphab\}}\nabla^2\Psi+\sum_{\Psi\in\{\beta,\rho,\sigma,\betab,\alphab\}}(\Psi\nabla\psi+\psi\nabla\Psi).$$
Estimating directly, we get
\begin{equation*}
\begin{split}
&\sum_{i\leq 1}||\nabla_3\nabla^i(\beta,\rho,\sigma,\betab)||_{L^2(\Hb_{\ub})}\\
\leq &\sum_{i\leq 2}\sum_{\Psi\in\{\rho,\sigma,\betab,\alphab\}}||\nabla^i\Psi||_{L^2(\Hb_{\ub})}+C\sup_u\sum_{i_1\leq 1}||\nabla^{i_1}\psi||_{L^4(S_{u,\ub})}\sum_{i_2\leq 2}\sum_{\Psi\in\{\rho,\sigma,\betab,\alphab\}}||\nabla^{i_2}\Psi||_{L^2(\Hb_{\ub})}\\
&+C\epsilon\sup_u\sum_{i_1\leq 1}||\nabla^{i_1}\psi||_{L^4(S_{u,\ub})}\sum_{i_2\leq 2}\sum_{\Psi\in\{\beta,\rho,\sigma,\betab\}}||\nabla^{i_2}\Psi||_{L^2(S_{u,\ub})}.
\end{split}
\end{equation*}
Therefore, by the definition of $\mathcal R$ and Proposition \ref{Ricci}, we have
$$\sum_{i=0}^1||\nabla_3\nabla^i(\beta,\rho,\sigma,\betab)||_{L^2(\Hb_{\ub})}\leq C(\mathcal O_0)\mathcal R+C(\mathcal O_0)\epsilon\sum_{i\leq 1}||\nabla^i(\beta,\rho,\sigma,\betab)||_{L^2(S_{u,\ub})}.$$
Therefore, by Proposition \ref{transport}, we have
\begin{equation*}
\sum_{i\leq 1}||\nabla^i(\beta,\rho,\sigma,\betab)||_{L^2(S_{u,\ub})}\leq C(\mathcal R_0)+ \epsilon^{\frac 12} C(\mathcal O_0)\mathcal R+C(\mathcal O_0)\epsilon^{\frac 32}\sum_{i\leq 1}||\nabla^i(\beta,\rho,\sigma,\betab)||_{L^2(S_{u,\ub})}.
\end{equation*}
The conclusion follows from choosing $\epsilon$ sufficiently small and absorbing the last term to the left hand side.
\end{proof}
We now derive $L^2(S)$ estimate for $K$ and its derivative, which will in particular allow us to apply Proposition \ref{ellipticthm}.
\begin{proposition}\label{Kest}
Suppose
$$\mathcal R<\infty,\quad\tilde{\mathcal O}_{3,2}<\infty.$$
Then there exists $\epsilon_0=\epsilon_0(\mathcal O_0,\mathcal R,\tilde{\mathcal O}_{3,2})$ such that for $\epsilon\leq \epsilon_0$, 
\[
\sum_{i\leq 1}||\nabla^i K||_{L^2(S)}\leq C(\mathcal O_0, \mathcal R_0)
\]
\end{proposition}
\begin{proof}
We use the equation
$$K=-\rho+\frac{1}{2} \chih\cdot\chibh-\frac 1 4 tr\chi\cdot \trchb,$$
and bound term by term using Propositions \ref{Ricci} and \ref{RS}.
\end{proof}

\subsection{Elliptic Estimates for Ricci Coefficients}\label{Ricciellipticsec}

This Subsection contains step 2 of the proof of the a priori estimates outlined in the beginning of the Section. We bound the third derivatives of the Ricci coefficients. They cannot be estimated by transport equations alone as before because that would require the control of three derivatives of curvature, which we do not have. Instead, we need to use elliptic estimates, following ideas in \cite{CK}, \cite{KN}, \cite{Chr}, \cite{KlRo}. We hope to find carefully chosen combination of $\Theta:=\nabla\psi+\Psi$ such that $\Theta$ satisfies transport equations
\begin{equation}\label{thtransport3}
\nabla_3\Theta=\psi\nabla\psi_{\Hb}+\psi\psi\psi+\psi\Psi_{\Hb},
\end{equation}
or
\begin{equation}\label{thtransport4}
\nabla_4\Theta=\psi\nabla\psi_H+\psi\psi\psi+\psi\Psi_H,
\end{equation}
where $\psi_H\in\{\trch,\chih,\eta,\etab,\omega,\trchb\}$, $\psi_{\Hb}\in\{\trch,\eta,\etab,\omegab,\trchb,\chibh\}$, $\Psi_H\in\{\beta,\rho,\sigma,\betab\}$ and $\Psi_{\Hb}\in\{\rho,\sigma,\betab,\alphab\}$. These transport equations will then be coupled to the estimates from Hodge systems of the type
$$\mathcal D\psi=\Theta+\Psi+\psi\psi$$
to recover the control for the third derivatives of the Ricci coefficients. In order to apply this strategy, it is desirable for (\ref{thtransport3}) and (\ref{thtransport4}) to have two special structures. Firstly, $\nabla\Psi$ should not appear on the right hand side of the equation so that $\Theta$ can be estimated by these transport equations without a loss of derivative. Secondly, we hope that the terms $\psi\nab\psi_H$ and $\psi\Psi_H$ do not enter the $\nabla_3\Theta$ equation (respectively, $\psi\nab\psi_{\Hb}$ and $\psi\Psi_{\Hb}$ do not enter the $\nabla_4\Theta$ equation). If this is satisfied for equations (\ref{thtransport3}) and (\ref{thtransport4}), then the integrated error terms can be controlled by the flux term that are bounded on the corresponding hypersurface. Indeed, the third derivatives of most the Ricci coefficients can be obtained in this way by constructing approximate $\Theta$ variables which satisfy either (\ref{thtransport3}) or (\ref{thtransport4}).

An extra difficulty arises when we estimate $\nab^3\omega$. More precisely, the $\Theta$ variable that allows us to retrieve estimates for derivatives of $\omega$ satisfies an equation of the form
\begin{equation}\label{probTh}
\nabla_3\Theta=\psi\nabla\psi_{\Hb}+\psi\psi\psi+\psi\Psi_{\Hb}+\psi\beta,
\end{equation}
which has an undesirable term $\psi\beta$. Since we prove estimates without any information on the singular curvature component $\alpha$, we automatically lose control of $\nab^2\beta$ in $L^2(\Hb)$. Therefore, when obtaining the estimates from this transport equation using Proposition \ref{transport}, the error term arising from $\psi\beta$ in (\ref{probTh}) cannot be bounded. In order to overcome this problem, we prove a weaker bound for $\nab^3\omega$. More precisely, we will only control it in $L^2(H)$ but will not obtain any bounds of its $L^2(\Hb)$ norm. We will moreover show that this is sufficient to close the energy estimates.

We first bound the $\Theta$ variables that satisfy equations of the form (\ref{thtransport3}) and (\ref{thtransport4}).
\begin{proposition}\label{Theta}
Let $\Theta$ satisfy the equation
$$\nabla_3\Theta=\psi\nabla\psi_{\Hb}+\psi\psi\psi+\psi\Psi_{\Hb},$$
or
$$\nabla_4\Theta=\psi\nabla\psi_H+\psi\psi\psi+\psi\Psi_H,$$
where $\psi_H\in\{\trch,\chih,\eta,\etab,\omega,\trchb\}$, $\psi_{\Hb}\in\{\trch,\eta,\etab,\omegab,\trchb,\chibh\}$, $\Psi_H\in\{\beta,\rho,\sigma,\betab\}$ and $\Psi_{\Hb}\in\{\rho,\sigma,\betab,\alphab\}$. 
Suppose $$\mathcal R<\infty,\quad \tilde{\mathcal O}_{3,2}<\infty.$$ 
Then there exists $\epsilon_0=\epsilon_0(\mathcal O_0,\mathcal R,\tilde{\mathcal O}_{3,2})$ such that for $\epsilon\leq \epsilon_0$.
$$\sum_{i=0}^2\sup_{u,\ub}||\nabla^i\Theta||_{L^2(S_{u,\ub})}\leq C(\mathcal O_0).$$
\end{proposition}
\begin{proof}
Notice that the assumption allows us to choose $\epsilon_0$ such that the conclusion of Proposition \ref{Ricci} holds. For $\Theta$ satisfying
$$\nabla_3\Theta=\psi\nabla\psi_{\Hb}+\psi\psi\psi+\psi\Psi_{\Hb},$$
we commute with angular derivatives to obtain
$$\nabla_{3}\nabla^{i}\Theta=\sum_{i_1+i_2+i_3+i_4=i}(\nabla^{i_1}\psi^{i_2}\nabla^{i_3}\psi\nabla^{i_4+1}\psi_{\Hb}+\nabla^{i_1}\psi^{i_2+1}\nabla^{i_3}\psi\nabla^{i_4}\psi+\nabla^{i_1}\psi^{i_2}\nabla^{i_3}\psi\nabla^{i_4}\Psi_{\Hb}).$$
Take $i\leq 2$. We estimate term by term. For the first two terms, except for the factor with $\nab^3\psi_{\Hb}$, can be estimated using the $\mathcal O$ norm, which is bounded in view of Proposition \ref{Ricci}. The term $\psi\nab^3\psi_H$ can be controlled by the $\tilde{\mathcal O}_{3,2}$ norm. We have
$$\sum_{i_1+i_2+i_3+i_4\leq 2}||\nabla^{i_1}\psi^{i_2}\nabla^{i_3}\psi\nabla^{i_4+1}\psi_{\Hb}+\nabla^{i_1}\psi^{i_2+1}\nabla^{i_3}\psi\nabla^{i_4}\psi||_{L^1_uL^2(S_{u,\ub})}\leq C(\mathcal O_0)(1+\epsilon^{\frac 12}\tilde{\mathcal O}_{3,2}).$$
The curvature term can be estimated using Proposition \ref{Ricci} and the definition of the curvature norm, since the component $\beta$ does not appear:
$$\sum_{i_1+i_2+i_3+i_4\leq 2}||\nabla^{i_1}\psi^{i_2}\nabla^{i_3}\psi\nabla^{i_4}\Psi_{\Hb}||_{L^1_uL^2(S_{u,\ub})}\leq C(\mathcal O_0)(1+\epsilon^{\frac 12}\mathcal R).$$
We then estimate using Proposition \ref{transport},
$$\sum_{i=0}^2\sup_u||\nabla^i\Theta||_{L^2(S_{u,\ub})}\leq C(\mathcal O_0)+C(\mathcal O_0)\epsilon^{\frac{1}{2}}(\mathcal R+\tilde{\mathcal O}_{3,2}).$$
Similarly, for $\Theta$ satisfying
$$\nabla_4\Theta=\psi\nabla\psi_{H}+\psi\psi\psi+\psi\Psi_H,$$
we commute with angular derivatives to obtain
$$\nabla_{3}\nabla^{i}\Theta=\sum_{i_1+i_2+i_3+i_4=i}(\nabla^{i_1}\psi^{i_2}\nabla^{i_3}\psi\nabla^{i_4+1}\psi_{H}+\nabla^{i_1}\psi^{i_2+1}\nabla^{i_3}\psi\nabla^{i_4}\psi+\nabla^{i_1}\psi^{i_2}\nabla^{i_3}\psi\nabla^{i_4}\Psi_H),$$
We then estimate as above and use Proposition \ref{transport} to get,
$$\sum_{i=0}^2\sup_{\ub}||\nabla^i\Theta||_{L^2(S_{u,\ub})}\leq C(\mathcal O_0)+C(\mathcal O_0)\epsilon^{\frac{1}{2}}(\mathcal R+\tilde{\mathcal O}_{3,2}).$$
The Proposition follows from choosing $\epsilon$ sufficiently small depending on $\mathcal R$ and $\tilde{\mathcal O}_{3,2}$.
\end{proof}
We now write down the explicit $\Theta$ quantities following \cite{KlRo} and show that the estimates for some of the Ricci coefficients can be recovered from $\Theta$ via elliptic equations.
\begin{proposition}\label{ellipticTheta}
Assume 
$$\mathcal R<\infty,\quad \tilde{\mathcal O}_{3,2}< \infty.$$
Let $\Theta$ denote any of the following quantities:
$$\nab\trch,\quad \nab\trchb,\quad \mu:=-\div\eta-\rhoc,\quad \mub:=-\div\etab-\rhoc.$$ 
Then $\Theta$ satisfies (\ref{thtransport3}) or (\ref{thtransport4}). Moreover, there exists $\epsilon_0=\epsilon_0(\mathcal O_0,\mathcal R,\tilde{\mathcal O}_{3,2})$ such that for $\epsilon\leq \epsilon_0$, we have
$$\sup_{u,\ub}||\nab^3(\trch,\trchb)||_{L^2(S_{u,\ub})}\leq C(\mathcal O_0),$$
$$\sup_u||\nabla^3(\chih,\eta,\etab)||_{L^2(H_u)} \leq C(\mathcal O_0)(\epsilon^{\frac 12}+\mathcal R),$$
$$\sup_{\ub}||\nabla^3(\chibh,\eta,\etab)||_{L^2(\Hb_{\ub})} \leq C(\mathcal O_0)(\epsilon^{\frac 12}+\mathcal R),$$
where, as before, $\psi_H\in\{\trch,\chih,\eta,\etab,\omega,\trchb\}$, $\psi_{\Hb}\in\{\trch,\eta,\etab,\omegab,\trchb,\chibh\}$.
\end{proposition}
\begin{proof}
Consider the following equations:
$$\nabla_3 \trchb=\psi_{\Hb}\psi_{\Hb},$$
$$\nabla_4 \trch=\psi_H\psi_H.$$
These equations do not involve a curvature term. After commuting with an angular derivative, we have
$$\nabla_{3}\nabla\trchb=\psi\nabla\psi_{\Hb}+\psi\psi\psi,$$
$$\nabla_{4}\nabla\trch=\psi\nabla\psi_H+\psi\psi\psi.$$
Thus $\nabla\trch$ and $\nabla\trchb$ satisfy the requirement to be one of the $\Theta$ quantities.
Thus, by definition,
$$||\nabla^3\trch||_{L^2(S_{u,\ub})} \leq C||\nabla^2\Theta||_{L^2(S_{u,\ub})}.$$
Moreover,
\begin{equation}\label{elliptictrch}
||\nabla^3\trch||_{L^2(H_u)} \leq C\epsilon^{\frac 12}\sup_{\ub}||\nabla^2\Theta||_{L^2(S_{u,\ub})}.
\end{equation}
and
\begin{equation}\label{elliptictrchb}
||\nabla^3\trchb||_{L^2(\Hb_{\ub})} \leq C\epsilon^{\frac 12}\sup_u||\nabla^2\Theta||_{L^2(S_{u,\ub})}.
\end{equation}
From the estimates for $\nabla^3\trch$ and $\nabla^3\trchb$, we can obtain the estimates for $\nabla^3\chih$ and $\nabla^3\chibh$. This is because of the Codazzi equations:
$$\div\chih=\frac{1}{2}\nabla\trch-\beta+\psi\psi,$$
$$\div\chibh=\frac{1}{2}\nabla\trchb+\betab+\psi\psi,$$
Since $\chih$ and $\chibh$ are symmetric traceless 2-tensors and we have the estimates for the Gauss curvature in Proposition \ref{Kest}, we can apply the elliptic estimates for symmetric traceless 2-tensors in Proposition \ref{elliptictraceless}
\begin{equation*}
\begin{split}
||\nabla^3\chih||_{L^2(H_u)} 
\leq &C(||\nabla^3\trch||_{L^2(H_u)}+||\nabla^2\beta||_{L^2(H_u)}+\sum_{i_1+i_2\leq 2}||\nabla^{i_1}\psi\nabla^{i_2}\psi||_{L^2(H_u)}).
\end{split}
\end{equation*}
Since, by Proposition \ref{Ricci},
\begin{equation}\label{quadratic}
\begin{split}
&\sum_{i_1+i_2\leq 2}||\nabla^{i_1}\psi\nabla^{i_2}\psi||_{L^2(H_u)}\\
\leq &C\epsilon^{\frac 12}\sup_{\ub}(||\nabla\psi||_{L^4(S_{u,\ub})}||\nabla\psi||_{L^4(S_{u,\ub})}+||\psi||_{L^\infty(S_{u,\ub})}\sum_{i\leq 2}||\nabla^i\psi||_{L^2(S_{u,\ub})})\leq C(\mathcal O_0)\epsilon^{\frac 12},
\end{split}
\end{equation}
we have
\begin{equation}\label{ellipticchih}
||\nabla^3\chih||_{L^2(H_u)} \leq C\epsilon^{\frac 12}||\nabla^2\Theta||_{L^2(S_{u,\ub})}+C\mathcal R+C(\mathcal O_0)\epsilon^{\frac 12}.
\end{equation}
Similarly, using the Codazzi equation for $\chibh$, we get
\begin{equation}\label{ellipticchibh}
||\nabla^3\chibh||_{L^2(\Hb_{\ub})} \leq C\epsilon^{\frac 12}||\nabla^2\Theta||_{L^2(S_{u,\ub})}+C\mathcal R+C(\mathcal O_0)\epsilon^{\frac 12}.
\end{equation}

Define
$$\mu=-\div\eta-\rhoc,$$
$$\mub=-\div\etab-\rhoc,$$
where as before $\rhoc=\rho-\frac 12 \chibh\cdot\chih$.
Notice that $\div\eta$ and $\div\etab$ satisfy the following equations:
$$\nabla_4\div\eta=-\div\beta+\psi\nabla\psi+\psi\psi\psi+\psi\Psi_H,$$
$$\nabla_3\div\etab=\div\betab+\psi\nabla\psi+\psi\psi\psi+\psi\Psi_{\Hb},$$
Recall the Bianchi equation:
\begin{equation*}
\begin{split}
&\nabla_4\rho+\frac 32\trch\rho=\div\beta-\frac 12\chibh\cdot\alpha+\zeta\cdot\beta+2\etab\cdot\beta,\\
&\nab_3\rho+\frac 32\trchb\rho=-\div\betab- \frac 12\chih\cdot\alphab+\zeta\cdot\betab-2\eta\cdot\betab,\\
\end{split}
\end{equation*}
In the $\nabla_4$ equation for $\rho$, the term $\alpha$ appears. This is undesirable because we do not have good estimates for $\alpha$. Instead we renormalize the equations using
$$\rhoc:=\rho-\frac{1}{2}\chibh\cdot\chih.$$
Now, we write down the equations for $\rhoc$ instead of $\sigma$ and $\rho$. Since $\nabla_4\chih=-\alpha+\psi\psi$, the $\nab_4$ equations now do not contain $\alpha$. Schematically, we have
\begin{equation*}
\begin{split}
&\nabla_4\rhoc=\div\beta+\psi\nabla\psi+\psi\Psi_H+\psi\psi\psi,\\
&\nabla_3\rhoc=-\div\betab+\psi\nabla\psi+\psi\Psi_{\Hb}+\psi\psi\psi,
\end{split}
\end{equation*}
where all the curvature components are now acceptable. Thus $\mu$ and $\mub$ satisfy the following equations:
$$\nabla_4 \mu=\psi\Psi_H+\psi\nabla\psi_H+\psi\psi\psi,$$
$$\nabla_3 \mub=\psi\Psi_{\Hb}+\psi\nabla\psi_{\Hb}+\psi\psi\psi,$$
Thus $\mu$ and $\mub$ also satisfy the schematic equation for $\Theta$. We now show that we can get the estimates $\nabla^3\eta$ and $\nabla^3\eta$ from that of $\nabla^2\mu$ and $\nabla^3\mub$. To do this, we consider the following div-curl systems
$$\div\eta=-\mu-\rho+\psi\psi,$$
$$\curl\eta=\sigma+\psi\psi,$$
$$\div\etab=-\mub-\rho+\psi\psi,$$
$$\curl\etab=-\sigma+\psi\psi.$$
Notice that $\rho$ and $\sigma$ can be estimated by $\mathcal R$ on both $H_u$ and $\Hb_{\ub}$. Then, using (\ref{quadratic}) and Proposition \ref{ellipticthm}, we have
\begin{equation}\label{ellipticetau}
||\nabla^3(\eta,\etab)||_{L^2(H_u)} \leq C\epsilon^{\frac 12}||\nabla^2\Theta||_{L^2(S_{u,\ub})}+C\mathcal R+C(\mathcal O_0)\epsilon^{\frac 12},
\end{equation}
and
\begin{equation}\label{ellipticetaub}
||\nabla^3(\eta,\etab)||_{L^2(\Hb_{\ub})} \leq C\epsilon^{\frac 12}||\nabla^2\Theta||_{L^2(S_{u,\ub})}+C\mathcal R+C(\mathcal O_0)\epsilon^{\frac 12}.
\end{equation}
The Proposition thus follows from (\ref{elliptictrch}), (\ref{elliptictrchb}), (\ref{ellipticchih}), (\ref{ellipticchibh}), (\ref{ellipticetau}), (\ref{ellipticetaub}) and Proposition \ref{Theta}.
\end{proof}

The above construction of $\Theta$ provides the desired estimates for the third derivatives of all the Ricci coefficients with the exception of $\omega$ and $\omegab$. Therefore, to estimate the $\tilde{\mathcal O}_{3,2}$ norm, it remains to prove estimates for $\nab^3\omega$ and $\nab^3\omegab$. To this end we follow the strategy in \cite{KlRo}. Define two auxiliary quantities $\omega^\dagger$ and $\omegab^\dagger$ to be solutions to 
$$\nabla_3\omega^\dagger=\frac{1}{2}\sigma,$$
$$\nabla_4\omegab^\dagger=\frac{1}{2}\sigma,$$
with zero boundary conditions along $H_0$ and $\Hb_0$ respectively. For the remainder of the Section, $\psi$ will also be used to denote $\omega^\dagger$ and $\omegab^\dagger$; $\psi_H$ will also be used to denote $\omega^\dagger$; and $\psi_{\Hb}$ will also be used to denote $\omegab^\dagger$. It is easy to see that
\begin{proposition}\label{daggerest}
Suppose
$$\mathcal R<\infty.$$
Then there exists $\epsilon_0=\epsilon_0(\mathcal O_0,\mathcal R)$ such that for every $\epsilon\leq \epsilon_0$,
$$\sum_{i\leq 2}||\nab^i(\omega^\dagger,\omegab^\dagger)||_{L^2(S_{u,\ub})}\leq C.$$
\end{proposition}
\begin{proof}
Since $\omega^\dagger$ and $\omegab^\dagger$ obey the equations
$$\nabla_3\omega^\dagger=\frac{1}{2}\sigma,$$
$$\nabla_4\omegab^\dagger=\frac{1}{2}\sigma,$$
with zero data, it follows from Proposition \ref{transport} that 
$$\sum_{i\leq 2}||\nab^i(\omega^\dagger,\omegab^\dagger)||_{L^2(S_{u,\ub})}\leq C(\mathcal O_0)\epsilon^{\frac 12}\mathcal R.$$
The Proposition follows from choosing $\epsilon$ appropriately.
\end{proof}
We are now ready to use elliptic estimates to control $\nab^3\omega$ and $\nab^3\omegab$. As before, our strategy is to use a combination $\Theta=\nab\psi+\Psi$ that can be estimated using a transport equation. However, $\Theta$ in this case does not satisfy an equation of the type in Proposition \ref{Theta}.
\begin{proposition}\label{omegaelliptic}
Suppose 
$$\mathcal R<\infty,\quad\tilde{\mathcal O}_{3,2}<\infty.$$
Then there exists $\epsilon_0=\epsilon_0(\mathcal O_0,\mathcal R,\tilde{\mathcal O}_{3,2})$ such that whenever $\epsilon\leq \epsilon_0$, we have
$$||\nab^3(\omega,\omega^\dagger)||_{L^2(H_u)}\leq C(\mathcal O_0)(\epsilon^{\frac 12}+\mathcal R),$$
$$||\nab^3(\omegab,\omegab^\dagger)||_{L^2(\Hb_{\ub})}\leq C(\mathcal O_0)(\epsilon^{\frac 12}+\mathcal R),$$
\end{proposition}
\begin{proof}
Let
$$\kappa:=\nabla\omega+^*\nabla\omega^{\dagger}-\frac 12 \beta,$$
$$\kappab:=-\nabla\omegab+^*\nabla\omegab^{\dagger}-\frac 12 \betab.$$
We have 
$$\nabla_3\nabla\omega=\frac 12 \nabla\rho+\psi\nabla\psi+\psi\Psi+\psi\psi\psi,$$
$$\nabla_3{ }^*\nabla\omega^\dagger=\frac {1}{2} {^*\nabla}\sigma+\psi\Psi+\psi\nabla\omega^\dagger,$$
and the Bianchi equation
$$\nabla\rho+^*\nabla\sigma=\nabla_3\beta+\psi\Psi.$$
Therefore
\begin{equation}\label{kappaeqn}
\nabla_3\kappa=\psi\Psi+\psi\nabla\psi+\psi\psi\psi.
\end{equation}
Similar, using
$$\nabla_4\nabla\omegab=\frac 12 \nabla\rho+\psi\nabla\psi+\psi\Psi+\psi\psi\psi,$$
$$\nabla_4{ }^*\nabla\omegab^\dagger=\frac 12 {^*\nabla}\sigma+\psi\Psi+\psi\nabla\omega^\dagger,$$
and the Bianchi equation
$$-\nabla\rho+^*\nabla\sigma=\nabla_4\betab+\psi\Psi,$$
we have
\begin{equation}\label{kappabeqn}
\nabla_4\kappab=\psi\Psi+\psi\nabla\psi+\psi\psi\psi.
\end{equation}
The equations for $\kappa$ and $\kappab$ do not have the favorable structure of the equations of the other $\Theta$ variables. In particular, for the $\nab_3\kappa$ equation, $\beta$ and $\nab\omega$ appear as source terms. The loss of the $L^2$ energy estimate for $\alpha$ is automatically accompanied by the loss of estimate for $\beta$ in $L^2(\Hb)$. Therefore, $\nab^2\beta$ cannot be controlled in $L^2(\Hb_{\ub})$. This means that we cannot estimate $\nab^2\kappa$ in $L^2(\Hb_{\ub})$. Nevertheless, as we now show, we can estimate $\nab^2\kappa$ in $L^2(H_u)$. This estimate will be sufficient in the sequel.

Commuting (\ref{kappaeqn}) with angular derivatives, we get
$$\nabla_3\nab^2\kappa=\sum_{i_1+i_2+i_3\leq 2}\psi^{i_1}\nab^{i_2}\psi\nab^{i_3}\Psi+\sum_{i_1+i_2+i_3\leq 3}\psi^{i_1}\nab^{i_2}\psi\nabla^{i_3}\psi.$$
Notice that we now allow $\psi$ to represent the terms $\omega^\dagger$ and $\omegab^\dagger$. Since by Proposition \ref{daggerest} $\nab^i\omega^\dagger$ and $\nab^i\omegab^\dagger$ satisfy the same estimates as any other $\psi$ for $i\leq 2$, those terms can be estimated in the same way as in Proposition \ref{Theta}. The new terms that have not been estimated in Proposition \ref{Theta} are $\psi\nab^3\omegab^\dagger$, $\displaystyle\sum_{i_1+i_2+i_3\leq 2}\psi^{i_1}\nab^{i_2}\psi\nab^{i_3}\beta$ and $\psi\nab^3(\chih,\omega,\omega^\dagger)$. We estimate them by
$$||\psi\nab^3\omegab^\dagger||_{L^1_uL^2(S)}\leq C(\mathcal O_0)\epsilon^{\frac 12}||\nab^3\omegab^\dagger||_{L^2_uL^2(S)},$$
and
$$||\sum_{i_1+i_2+i_3\leq 2}\psi^{i_1}\nab^{i_2}\psi\nab^{i_3}\beta||_{L^1_uL^2(S_{u,\ub})}\leq C(\mathcal O_0)\epsilon^{\frac 12}\sum_{i\leq 2}||\nab^i\beta||_{L^2_uL^2(S_{u,\ub})},$$
and
$$||\psi\nab^3(\chih,\omega,\omega^\dagger)||_{L^1_uL^2(S)}\leq C(\mathcal O_0)\epsilon^{\frac 12}||\nab^3(\chih,\omega,\omega^\dagger)||_{L^2_uL^2(S)}.$$
Then, using Proposition \ref{transport} and the estimates in the proof of Proposition \ref{Theta}, we have
\begin{equation*}
\begin{split}
&||\nab^2\kappa||_{L^2(S_{u,\ub})}\\
\leq& C(\mathcal O_0)(1+\epsilon^{\frac 12}\tilde{\mathcal O}_{3,2}+\epsilon^{\frac 12}\mathcal R+\epsilon^{\frac 12}||\nab^3\omegab^\dagger||_{L^2_uL^2(S_{u,\ub})}\\
&\qquad+\epsilon^{\frac 12}\sum_{i\leq 2}||\nab^i\beta||_{L^2_uL^2(S_{u,\ub})}+\epsilon^{\frac 12}||\nab^3(\chih,\omega,\omega^\dagger)||_{L^2_uL^2(S_{u,\ub})}).
\end{split}
\end{equation*}
This holds for any fixed $u,\ub$. Integrating this in $L^2_{\ub}$, we get
\begin{equation*}
\begin{split}
&||\nab^2\kappa||_{L^2(H_u)}\\
\leq &C(\mathcal O_0)\epsilon^{\frac 12}(1+\epsilon^{\frac 12}(\tilde{\mathcal O}_{3,2}+\mathcal R+||\nab^3\omegab^\dagger||_{L^\infty_{\ub}L^2_uL^2(S)})\\
&\qquad+\sum_{i\leq 2}||\nab^i\beta||_{L^2_{\ub}L^2_uL^2(S)}+||\nab^3(\chih,\omega,\omega^\dagger)||_{L^2_{\ub}L^2_uL^2(S)}).
\end{split}
\end{equation*}
Using Cauchy-Schwarz in the $u$ variable, we have
$$\sum_{i\leq 2}||\nab^i\beta||_{L^2_{\ub}L^2_uL^2(S_{u,\ub})}\leq \epsilon^{\frac 12}\mathcal R,$$
and
$$||\nab^3(\chih,\omega,\omega^\dagger)||_{L^2_{\ub}L^2_uL^2(S_{u,\ub})})\leq \epsilon^{\frac 12}\tilde{\mathcal O}_{3,2}+\epsilon^{\frac 12}||\nab^3\omega^\dagger||_{L^\infty_{u}L^2_{\ub}L^2(S_{u,\ub})}.$$
Therefore,
\begin{equation}\label{kappa}
||\nab^2\kappa||_{L^2(H_u)}\leq C(\mathcal O_0)\epsilon^{\frac 12}(1+\epsilon^{\frac 12}\tilde{\mathcal O}_{3,2}+\epsilon^{\frac 12}\mathcal R+\epsilon^{\frac 12}||\nab^3\omegab^\dagger||_{L^\infty_{\ub}L^2_uL^2(S_{u,\ub})}+\epsilon^{\frac 12}||\nab^3\omega^\dagger||_{L^\infty_{u}L^2_{\ub}L^2(S_{u,\ub})}).
\end{equation}
Similarly, we can use (\ref{kappabeqn}) to conclude that 
\begin{equation}\label{kappab}
||\nab^2\kappab||_{L^2(\Hb_{\ub})}\leq C(\mathcal O_0)\epsilon^{\frac 12}(1+\epsilon^{\frac 12}\tilde{\mathcal O}_{3,2}+\epsilon^{\frac 12}\mathcal R+\epsilon^{\frac 12}||\nab^3\omegab^\dagger||_{L^\infty_{\ub}L^2_uL^2(S_{u,\ub})}+\epsilon^{\frac 12}||\nab^3\omega^\dagger||_{L^\infty_{u}L^2_{\ub}L^2(S_{u,\ub})}).
\end{equation}
We show that we can recover $\nabla^3\omega$, $\nabla^3\omegab$, $\nabla^3\omega^\dagger$ and $\nabla^3\omegab^\dagger$ by considering the div-curl system for $\nabla\omega$ and $\nabla\omega^\dagger$. We have
$$\div\nabla\omega=\div\kappa+\frac 12\div\beta,$$
$$\curl\nabla\omega=0,$$
$$\div\nabla\omega^\dagger=\curl\kappa+\frac 12\curl\beta,$$
$$\curl\nabla\omega^\dagger=0,$$
$$\div\nabla\omegab=-\div\kappab-\frac 12\div\betab,$$
$$\curl\nabla\omegab=0,$$
$$\div\nabla\omegab^\dagger=-\curl\kappab-\frac 12\curl\betab,$$
$$\curl\nabla\omegab^\dagger=0.$$
Thus, by Proposition \ref{ellipticthm} and \ref{Kest}, we have
\begin{equation*}\label{ellipticomega}
||\nabla^3(\omega,\omega^\dagger)||_{L^2(H_u)}\leq C(||\nabla^2\kappa||_{L^2(H_u)}+\sum_{i\leq 2}||\nabla^i\beta||_{L^2(H_u)}),
\end{equation*}
and
\begin{equation*}\label{ellipticomegab}
||\nabla^3(\omegab,\omegab^\dagger)||_{L^2(\Hb_{\ub})}\leq C(||\nabla^2\kappab||_{L^2(\Hb_{\ub})}+\sum_{i\leq 2}||\nabla^i\betab||_{L^2(\Hb_{\ub})}).
\end{equation*}
Using (\ref{kappa}) and (\ref{kappab}), we can thus choose $\epsilon$ sufficiently small to conclude the Proposition.
\end{proof}

Putting these all together gives
\begin{proposition}\label{Riccielliptic}
Assume
$$\mathcal R<\infty.$$
Then there exists $\epsilon_0=(\mathcal O_0,\mathcal R)$ such that 
\[
\tilde{\mathcal O}_{3,2}\leq C(\mathcal O_0)(\epsilon^{\frac 12}+\mathcal R).
\]
\end{proposition}
\begin{proof}
Impose the bootstrap assumption
\begin{equation}\label{BA2}
\tilde{\mathcal O}_{3,2}\leq \Delta_1.
\end{equation}
Under this assumption, we can choose $\epsilon$ depending on $\mathcal O_0$, $\mathcal R$ and $\Delta_1$ such that the conclusions in Propositions \ref{Theta}, \ref{ellipticTheta} and \ref{omegaelliptic} hold. 
Thus, 
$$\tilde{\mathcal O}_{3,2}\leq C(\mathcal O_0)(\epsilon^{\frac 12}+\mathcal R).$$
This improves the bootstrap assumption (\ref{BA2}) for $\Delta_1$ sufficiently large. Moreover, $\Delta_1$ can be chosen to depend only on $\mathcal O_0$ and $\mathcal R$. Thus, $\epsilon_0$ depends only on $\mathcal O_0$ and $\mathcal R$.
\end{proof}

\subsection{Energy Estimates for Curvature}\label{energyestimatessec}

In this Subsection, we carry out step 3 in the proof of the a priori estimates as outlined in the beginning of the Section. We prove energy estimates for the curvature components and their first two derivatives, thus providing bounds for $\mathcal R$ and closing all the estimates.  These energy estimates will be derived after renormalization directly from the Bianchi equations. This is analogous to \cite{Holzegel}, where a similar procedure without curvature renormalization has been carried out to derive energy estimates for all components of curvature. We remind the reader of our challenge to avoid any information on the non-$L^2$-integrable curvature component $\alpha$ in our context. 

We will need the following two Propositions, which can be proved by a direct computation.
\begin{proposition}\label{intbyparts34}
Suppose $\phi_1$ and $\phi_2$ are $r$ tensorfields, then
$$\int_{D_{u,\ub}} \phi_1 \nabla_4\phi_2+\int_{D_{u,\ub}}\phi_2\nabla_4\phi_1= \int_{\Hb_{\ub}(0,u)} \phi_1\phi_2-\int_{\Hb_0(0,u)} \phi_1\phi_2+\int_{D_{u,\ub}}(2\omega-\trch)\phi_1\phi_2,$$
$$\int_{D_{u,\ub}} \phi_1 \nabla_3\phi_2+\int_{D_{u,\ub}}\phi_2\nabla_3\phi_1= \int_{H_{u}(0,\ub)} \phi_1\phi_2-\int_{H_0(0,\ub)} \phi_1\phi_2+\int_{D_{u,\ub}}(2\omegab-\trchb)\phi_1\phi_2.$$
\end{proposition}
\begin{proposition}\label{intbypartssph}
Suppose we have an $r$ tensorfield $^{(1)}\phi$ and an $r-1$ tensorfield $^{(2)}\phi$.
$$\int_{D_{u,\ub}}{ }^{(1)}\phi^{A_1A_2...A_r}\nabla_{A_r}{ }^{(2)}\phi_{A_1...A_{r-1}}+\int_{D_{u,\ub}}\nabla^{A_r}{ }^{(1)}\phi_{A_1A_2...A_r}{ }^{(2)}\phi^{A_1...A_{r-1}}= -\int_{D_{u,\ub}}(\eta+\etab){ }^{(1)}\phi{ }^{(2)}\phi.$$
\end{proposition}
By using the Bianchi equations and the integration by parts formula in Proposition \ref{intbyparts34} and \ref{intbypartssph}, we can prove the following energy estimates. The main new ingredient is the observation that $\alpha$ can be renormalized away in the energy estimates. We can therefore control $\beta,\rho,\sigma,\betab,\alphab$ without any knowledge of $\alpha$
\begin{proposition}\label{ee}
The following $L^2$ estimates for the curvature components hold for $i\leq 2$:
\begin{equation*}
\begin{split}
&\sum_{\Psi\in\{\beta,\rho,\sigma,\betab\}}\int_{H_u} |\nabla^i\Psi|^2_\gamma+\sum_{\Psi\in\{\rho,\sigma,\betab,\alphab\}}\int_{\Hb_{\ub}} |\nabla^i\Psi|^2_\gamma  \\
\leq &\sum_{\Psi\in\{\beta,\rho,\sigma,\betab\}}\int_{H_{u'}} |\nabla^i\Psi|^2_\gamma+\sum_{\Psi\in\{\rho,\sigma,\betab,\alphab\}}\int_{\Hb_{\ub'}} |\nabla^i\Psi|^2_\gamma+\epsilon^{\frac{1}{2}}C(\mathcal O_0).\\
&+\int_{D_{u,\ub}}\nabla^i\Psi\sum_{i_1+i_2+i_3+i_4=i}\nabla^{i_1}\psi^{i_2}\nabla^{i_3}\psi\nabla^{i_4}\Psi\\
&+\int_{D_{u,\ub}}\nabla^i\Psi\sum_{i_1+i_2+i_3+i_4=i-1}\nabla^{i_1}\psi^{i_2}\nabla^{i_3}K\nabla^{i_4}\Psi\\
&+\int_{D_{u,\ub}}\nabla^i\Psi\sum_{i_1+i_2+i_3+i_4=i+1}\nabla^{i_1}\psi^{i_2}\nabla^{i_3}\psi\nabla^{i_4}\psi.
\end{split}
\end{equation*}
\end{proposition}
\begin{proof}
We first consider the following sets of Bianchi equations:
\begin{equation*}
\begin{split}
&\nabla_3\beta+tr\chib\beta=\nabla\rho + 2\omegab \beta +^*\nabla\sigma +2\chih\cdot\betab+3(\eta\rho+^*\eta\sigma),\\
&\nabla_4\sigma+\frac 32tr\chi\sigma=-\div^*\beta+\frac 12\chibh\cdot ^*\alpha-\zeta\cdot^*\beta-2\etab\cdot
^*\beta,\\
&\nabla_4\rho+\frac 32tr\chi\rho=\div\beta-\frac 12\chibh\cdot\alpha+\zeta\cdot\beta+2\etab\cdot\beta,\\
\end{split}
\end{equation*}
In the $\nabla_4$ equations for $\sigma$ and $\rho$, the term $\alpha$ appears. This is undesirable because we do not have good estimates for $\alpha$. Instead we renormalize both $\sigma$ and $\rho$ as in the proof of Proposition \ref{ellipticTheta}. Define
$$\sigmac:=\sigma+\frac{1}{2}\chibh\wedge\chih,$$
$$\rhoc:=\rho-\frac{1}{2}\chibh\cdot\chih.$$
Now, we write down the equations for $\sigmac$ and $\rhoc$ instead of $\sigma$ and $\rho$. Since $\nabla_4\chih=-\alpha+\psi\psi$, the equations now do not contain $\alpha$. Schematically, we have
\begin{equation}\label{rcsc}
\begin{split}
&\nabla_4\sigmac=-\div^*\beta+\psi\nabla\psi+\psi\Psi+\psi\psi\psi,\\
&\nabla_4\rhoc=\div\beta+\psi\nabla\psi+\psi\Psi+\psi\psi\psi,
\end{split}
\end{equation}
where $\Psi\in\{\beta,\rho,\sigma,\betab,\alphab\}$ is now one of the acceptable curvature components.
Rewrite also the first equation as
\begin{equation}\label{bc}
\nabla_3\beta=\nabla\rhoc+^*\nabla\sigmac+\psi\Psi+\psi\nabla\psi.
\end{equation}
Applying Proposition \ref{intbypartssph} yields the following: 
\begin{equation*}
\begin{split}
\int_{D_{u,\ub}} <\beta,\nabla_3\beta>_\gamma =&\int <\beta,\nabla\rhoc+^*\nabla\sigmac>_\gamma+<\beta,\psi\Psi>_\gamma \\
=&\int_{D_{u,\ub}} -<\div\beta,\rhoc>_\gamma+<\div ^*\beta,\sigmac>_\gamma +<\beta,\psi\Psi+\psi\nabla\psi>_\gamma\\
=&\int_{D_{u,\ub}} -<\nabla_4\rhoc,\rhoc>_\gamma-<\nabla_4\sigmac,\sigmac>_\gamma +<\Psi,\psi\Psi+\psi\nabla\psi+\psi\psi\psi>_\gamma.
\end{split}
\end{equation*}
Applying Proposition \ref{intbyparts34} then yields the energy estimates for the $0$-th derivative of the curvature:
\begin{equation*}
\begin{split}
&\int_{H_u} |\beta|^2_\gamma+\int_{\Hb_{\ub}} \rhoc^2+\sigmac^2  \\
\leq &\int_{H_{u'}} |\beta|^2_\gamma+\int_{\Hb_{\ub'}} \rhoc^2+\sigmac^2
+\int_{D_{u,\ub}}\Psi(\psi\Psi+\psi\nabla\psi+\psi\psi\psi).
\end{split}
\end{equation*}
We now commute equations (\ref{rcsc}) and (\ref{bc}) with $\nabla^i$ to get
\begin{equation*}
\begin{split}
&\nabla_3\nabla^i\beta- \nabla \nabla^i\rhoc-\nabla\nabla^i\sigmac \\
\sim&\sum_{i_1+i_2+i_3+i_4=i+1}\nabla^{i_1}\psi^{i_2}\nabla^{i_3}\psi\nabla^{i_4}\psi \\
&+\sum_{i_1+i_2+i_3+i_4+i_5=i-1}\nabla^{i_1}\psi^{i_2}\nabla^{i_3}K^{i_4}\nabla^{i_5}(\rhoc,\sigmac)+\sum_{i_1+i_2+i_3=i}\nabla^{i_1}\psi^{i_2}\nabla^{i_3}(\psi\Psi),\\
\end{split}
\end{equation*}
and
\begin{equation*}
\begin{split}
&\nabla_4\nabla^i\rhoc- \div\nabla^i\beta\\
\sim&\sum_{i_1+i_2+i_3+i_4=i+1}\nabla^{i_1}\psi^{i_2}\nabla^{i_3}\psi\nabla^{i_4}\psi \\
&+\sum_{i_1+i_2+i_3+i_4+i_5=i-1}\nabla^{i_1}\psi^{i_2}\nabla^{i_3}K^{i_4}\nabla^{i_5}\beta+\sum_{i_1+i_2+i_3=i}\nabla^{i_1}\psi^{i_2}\nabla^{i_3}(\psi\Psi),\\
\end{split}
\end{equation*}
and 
\begin{equation*}
\begin{split}
&\nabla_4\nabla^i\sigmac- \div^*\nabla^i\beta\\
\sim&\sum_{i_1+i_2+i_3+i_4=i+1}\nabla^{i_1}\psi^{i_2}\nabla^{i_3}\psi\nabla^{i_4}\psi \\
&+\sum_{i_1+i_2+i_3+i_4+i_5=i-1}\nabla^{i_1}\psi^{i_2}\nabla^{i_3}K^{i_4}\nabla^{i_5}\beta+\sum_{i_1+i_2+i_3=i}\nabla^{i_1}\psi^{i_2}\nabla^{i_3}(\psi\Psi),\\
\end{split}
\end{equation*}
where 
$$(\div\nabla^i\beta)_{A_1...A_i}=(\nabla^{i+1}\beta)^B{ }_{A_1...A_iB}.$$
We integrate by parts using Proposition \ref{intbyparts34} and \ref{intbypartssph} as in the $0$-th derivative case to get
\begin{equation*}
\begin{split}
&\int_{H_u} |\nabla^i\beta|^2_\gamma+\int_{\Hb_{\ub}} |\nabla^i(\rhoc,\sigmac)|^2_\gamma  \\
\leq &\int_{H_{u'}} |\nabla^i\beta|^2_\gamma+\int_{\Hb_{\ub'}} |\nabla^i(\rhoc,\sigmac)|^2_\gamma+\int_{D_{u,\ub}}\nabla^i\Psi\sum_{i_1+i_2+i_3+i_4=i}\nabla^{i_1}\psi^{i_2}\nabla^{i_3}\psi\nabla^{i_4}\Psi\\
&+\int_{D_{u,\ub}}\nabla^i\Psi\sum_{i_1+i_2+i_3+i_4=i-1}\nabla^{i_1}\psi^{i_2}\nabla^{i_3}K\nabla^{i_4}\Psi\\
&+\int_{D_{u,\ub}}\nabla^i\Psi\sum_{i_1+i_2+i_3+i_4=i+1}\nabla^{i_1}\psi^{i_2}\nabla^{i_3}\psi\nabla^{i_4}\psi.
\end{split}
\end{equation*}
We then consider the following set of Bianchi equations. Notice that $\alpha$ does not appear in any equations in this set.
\begin{equation*}
\begin{split}
&\nabla_3\sigma+\frac 32tr\chib\sigma=-\div ^*\betab+\frac 12\chih\cdot ^*\alphab-\zeta\cdot ^*\betab-2\eta\cdot 
^*\betab,\\
&\nabla_3\rho+\frac 32tr\chib\rho=-\div\betab- \frac 12\chih\cdot\alphab+\zeta\cdot\betab-2\eta\cdot\betab,\\
&\nabla_4\betab+tr\chi\betab=-\nabla\rho +^*\nabla\sigma+ 2\omega\betab +2\chibh\cdot\beta-3(\etab\rho-^*\etab\sigma),\\
\end{split}
\end{equation*}
From these we can derive the following estimates:
\begin{equation*}
\begin{split}
&\int_{H_u} |\nabla^i(\rho,\sigma)|^2_\gamma+\int_{\Hb_{\ub}} |\nabla^i\betab|^2_\gamma  \\
\leq &\int_{H_{u'}} |\nabla^i(\rho,\sigma)|^2_\gamma+\int_{\Hb_{\ub'}} |\nabla^i\betab|^2_\gamma+\int_{D_{u,\ub}}\nabla^i\Psi\sum_{i_1+i_2+i_3+i_4=i}\nabla^{i_1}\psi^{i_2}\nabla^{i_3}\psi\nabla^{i_4}\Psi\\
&+\int_{D_{u,\ub}}\nabla^i\Psi\sum_{i_1+i_2+i_3+i_4=i-1}\nabla^{i_1}\psi^{i_2}\nabla^{i_3}K\nabla^{i_4}\Psi.
\end{split}
\end{equation*}
Finally, we look at the following set of Bianchi equations:
\begin{equation*}
\begin{split}
&\nabla_3\betab+2tr\chib\betab=-\div\alphab-2\omegab\betab+\etab \cdot\alphab,\\
&\nabla_4\alphab+\frac 12 tr\chi\alphab=-\nabla\hot \betab+ 4\omega\alphab-3(\chibh\rho-^*\chibh\sigma)+
(\zeta-4\etab)\hot \betab
\end{split}
\end{equation*}
Again, we notice the absence of the $\alpha$ terms. Hence we have the estimate
\begin{equation*}
\begin{split}
&\int_{H_u} |\nabla^i\betab|^2_\gamma+\int_{\Hb_{\ub}} |\nabla^i\alphab|^2_\gamma  \\
\leq &\int_{H_{u'}} |\nabla^i\betab|^2_\gamma+\int_{\Hb_{\ub'}} |\nabla^i\alphab|^2_\gamma+\int_{D_{u,\ub}}\nabla^i\Psi\sum_{i_1+i_2+i_3+i_4=i}\nabla^{i_1}\psi^{i_2}\nabla^{i_3}\psi\nabla^{i_4}\Psi\\
&+\int_{D_{u,\ub}}\nabla^i\Psi\sum_{i_1+i_2+i_3+i_4=i-1}\nabla^{i_1}\psi^{i_2}\nabla^{i_3}K\nabla^{i_4}\Psi.
\end{split}
\end{equation*}
We have thus established that 
$$\sum_{\Psi\in\{\beta,\rho,\sigma,\betab\}}\int_{H_u} |\nabla^i\Psi|^2_\gamma+\sum_{\Psic\in\{\rhoc,\sigmac,\betab,\alphab\}}\int_{\Hb_{\ub}} |\nabla^i\Psic|^2_\gamma$$
can be bounded by the right hand side in the statement of the Proposition. To conclude, we use the fact that
$$\Psic=\Psi+\psi\psi,$$
and thus
\begin{equation*}
\begin{split}
&\sum_{\Psi\in\{\rho,\sigma,\betab,\alphab\}}\int_{\Hb_{\ub}} |\nabla^i\Psi|^2_\gamma\\
\leq &\sum_{\Psic\in\{\rhoc,\sigmac,\betab,\alphab\}}\int_{\Hb_{\ub}} |\nabla^i\Psic|^2_\gamma+\sum_{i_1+i_2\leq 2}||\nabla^{i_1}\psi\nabla^{i_2}\psi||_{L^2(\Hb_{\ub})}\\
\leq& \sum_{\Psic\in\{\rhoc,\sigmac,\betab,\alphab\}}\int_{\Hb_{\ub}} |\nabla^i\Psic|^2_\gamma+\epsilon^{\frac{1}{2}}C(\mathcal O_0).
\end{split}
\end{equation*}
\end{proof}
Finally, we prove the boundedness of the norm $\mathcal R$:
\begin{proposition}\label{curvature}
There exists $\epsilon_0=\epsilon_0(\mathcal O_0,\mathcal R_0)$ such that for every $\epsilon\leq\epsilon_0$,
\[
 \mathcal R\leq C(\mathcal O_0, \mathcal R_0).
\]
\end{proposition}
\begin{proof}
Assume as a bootstrap assumption:
\begin{equation}\label{BA3}
\mathcal R\leq \Delta_2.
\end{equation}
Thus we can apply Propositions \ref{Ricci}, \ref{Kest} and \ref{Riccielliptic}.
By Proposition \ref{ee}, we have
\begin{equation}\label{eenonlinear}
\begin{split}
&\sum_{i=0}^{2}(\sum_{\Psi\in\{\beta,\rho,\sigma,\betab\}}\int_{H_u} (\nabla^i\Psi)^2+\sum_{\Psi\in\{\rho,\sigma,\betab,\alphab\}}\int_{\underline{H}_{\underline{u}}} (\nabla^i\Psi)^2)\\
\leq& C\mathcal R_0^2+\epsilon^{\frac{1}{2}}C(\mathcal O_0)+\int_{\mathcal D_{u,\underline{u}}}\sum_{i\leq 2}\nabla^i\Psi\sum_{i_1+i_2+i_3+i_4\leq 2}\nabla^{i_1}\psi^{i_2}\nabla^{i_3}\psi\nabla^{i_4}\Psi \\
&+\int_{\mathcal D_{u,\underline{u}}}\sum_{i\leq 2}\nabla^i\Psi\sum_{i_1+i_2+i_3+i_4\leq 1}\nabla^{i_1}\psi^{i_2}\nabla^{i_3}K\nabla^{i_4}\Psi \\
&+\int_{\mathcal D_{u,\underline{u}}}\sum_{i\leq 2}\nabla^i\Psi\sum_{i_1+i_2+i_3+i_4\leq 3}\nabla^{i_1}\psi^{i_2}\nabla^{i_3}\psi\nabla^{i_4}\psi \\
\leq &C\mathcal R_0^2+C\sum_{i\leq 2}\int_{D_{u,\ub}} (\nab^i\Psi)^2 +\sum_{i_1+i_2+i_3+i_4\leq 2}\int_{\mathcal D_{u,\underline{u}}}\left(\nabla^{i_1}\psi^{i_2}\nabla^{i_3}\psi\nabla^{i_4}\Psi\right)^2\\
&+\sum_{i_1+i_2+i_3+i_4\leq 1}\int_{\mathcal D_{u,\underline{u}}}\left(\nabla^{i_1}\psi^{i_2}\nabla^{i_3} K\nabla^{i_4}\Psi\right)^2 +\sum_{i_1+i_2+i_3+i_4\leq 3}\int_{\mathcal D_{u,\underline{u}}}\left(\nabla^{i_1}\psi^{i_2}\nabla^{i_3} \psi\nabla^{i_4}\psi\right)^2.
\end{split}
\end{equation}
It is helpful to first note that 
$$\sum_{i\leq 2}||\nabla^i\Psi||_{L^2(D_{u,\ub})}\leq C\epsilon^{\frac 12}\mathcal R,$$
since $\nab^i\Psi$ can be estimated either in $L^2(H_u)$ or $L^2(\Hb_{\ub})$.
By Sobolev Embedding in Propositions \ref{L4} and \ref{Linfty}, we thus have
$$\sum_{i\leq 1}||\nabla^i\Psi||_{L^2_uL^2_{\ub}L^4(S)}+||\Psi||_{L^2_uL^2_{\ub}L^\infty(S)}\leq C\epsilon^{\frac 12}\mathcal R,$$
A similar argument shows that
$$\sum_{i\leq 3}||\nabla^i\psi||_{L^2(D_{u,\ub})}+\sum_{i\leq 2}||\nabla^i\psi||_{L^2_uL^2_{\ub}L^4(S)}+\sum_{i\leq 1}||\nab^i\psi||_{L^2_uL^2_{\ub}L^\infty(S)}\leq C\epsilon^{\frac 12}\tilde{\mathcal O}_{3,2},$$
Thus by Proposition \ref{Riccielliptic}, we have
$$\sum_{i\leq 3}||\nabla^i\psi||_{L^2(D_{u,\ub})}+\sum_{i\leq 2}||\nabla^i\psi||_{L^2_uL^2_{\ub}L^4(S)}+\sum_{i\leq 1}||\nab^i\psi||_{L^2_uL^2_{\ub}L^\infty(S)}\leq C(\mathcal O_0,\mathcal R_0)\epsilon^{\frac 12}\mathcal R.$$
Now we estimate the first nonlinear term in (\ref{eenonlinear}):
\begin{equation*}
\begin{split}
\sum_{i_1+i_2+i_3+i_4\leq 2}\int_{\mathcal D_{u,\underline{u}}}\left(\nabla^{i_1}\psi^{i_2}\nabla^{i_3}\psi\nabla^{i_4}\Psi\right)^2. \\
\end{split}
\end{equation*}
We can estimate, by Proposition \ref{Ricci},
\begin{equation*}
\begin{split}
&\sum_{i_1,i_2\leq 1,i_3\leq 1}||\psi^{i_1}\nabla^{i_2}\psi\nabla^{i_3}\Psi||_{L^2(D_{u,\ub})}\\
\leq &(\sum_{i_1\leq 2}\sup_{u,\ub}||\psi||_{L^\infty(S_{u,\ub})}^{i_1})(\sum_{i_2\leq 2}\sup_{u,\ub}||\nabla^{i_2}\psi||_{L^2(S_{u,\ub})})(\sum_{i_3\leq 2}||\nabla^{i_3}\Psi||_{L^2(D_{u,\ub})})\\
\leq &\epsilon^{\frac 12} C(\mathcal O_0)R.
\end{split}
\end{equation*}
We then estimate the term
$$\sum_{i_1+i_2+i_3\leq 1}\int_{\mathcal D_{u,\underline{u}}}\left(\psi^{i_1}\nabla^{i_2} K\nabla^{i_3}\Psi\right)^2.$$
We have, by Propositions \ref{Ricci} and \ref{Kest},
\begin{equation*}
\begin{split}
&\sum_{i_1+i_2+i_3\leq 1}||\psi^{i_1}\nabla^{i_2} K\nabla^{i_3}\Psi||_{L^2(D_{u,\ub})}\\
\leq &(\sum_{i_1\leq 2}\sup_{u,\ub}||\psi||_{L^\infty(S_{u,\ub})}^{i_1})(\sum_{i_2\leq 1}\sup_{u,\ub}||\nabla^{i_2}K||_{L^2(S_{u,\ub})})(\sum_{i_3\leq 2}||\nabla^{i_3}\Psi||_{L^2(D_{u,\ub})})\\
\leq &\epsilon^{\frac 12} C(\mathcal O_0,\mathcal R_0)R.
\end{split}
\end{equation*}
Finally, we estimate
$$\sum_{i_1+i_2+i_3+i_4\leq 3}\int_{\mathcal D_{u,\underline{u}}}\left(\nabla^{i_1}\psi^{i_2}\nabla^{i_3} \psi\nabla^{i_4}\psi\right)^2.$$
By Propositions \ref{Ricci} and \ref{Riccielliptic},
\begin{equation*}
\begin{split}
&\sum_{i_1+i_2+i_3+i_4\leq 3}||\nabla^{i_1}\psi^{i_2}\nabla^{i_3} \psi\nabla^{i_4}\psi||_{L^2(D_{u,\ub})}\\
\leq &(\sum_{i_1\leq 3}\sup_{u,\ub}||\psi||_{L^\infty(S_{u,\ub})}^{i_1})(\sum_{i_2\leq 2}\sup_{u,\ub}||\nabla^{i_2}\psi||_{L^2(S_{u,\ub})})(\sum_{i_3\leq 3}||\nabla^{i_3}\psi||_{L^2(D_{u,\ub})})\\
\leq &\epsilon^{\frac 12} C(\mathcal O_0,\mathcal R_0)R.
\end{split}
\end{equation*}
Thus
$$\sum_{i=0}^{2}(\sum_{\Psi\in\{\beta,\rho,\sigma,\betab\}}\int_{H_u} (\nabla^i\Psi)^2+\sum_{\Psi\in\{\rho,\sigma,\betab,\alphab\}}\int_{\underline{H}_{\underline{u}}} (\nabla^i\Psi)^2)
\leq C\mathcal R_0^2+\epsilon^{\frac{1}{2}}C(\mathcal O_0)+\epsilon^{\frac 12} C(\mathcal O_0,\mathcal R_0)R.$$
By choosing $\epsilon$ sufficiently small we have
$$\mathcal R\leq C(\mathcal O_0,\mathcal R_0).$$
This improves (\ref{BA3}) for $\Delta_2$ sufficiently large. Moreover, $\Delta_2$ can be chosen to depend only on $\mathcal O_0$ and $\mathcal R_0$. Thus the choice of $\epsilon$ depends only on $\mathcal O_0$ and $\mathcal R_0$.
\end{proof}
\subsection{Propagation of Regularity}\label{Propregsec}
Up to this point, we have been considering spacetimes with smooth characteristic initial data. In the context of an impulsive gravitational wave, our estimates in this Section apply to spacetimes arising from the smooth data approximating the given singular data. For these spacetimes, we can also prove that estimates for the higher order derivatives of the curvature. In view of the bounds that we have already obtained, this follows from standard arguments. In particular, energy estimates for the higher order derivatives can be derived as in Section \ref{energyestimatessec}. In this case, since the equations are linear in the highest order derivatives, one can use the bounds from the previous Sections to obtain the desired estimates. More precisely, we have
\begin{proposition}\label{propagationregularity}
Suppose, in addition to the assumptions of Theorem \ref{timeofexistence}, the bound 
$$\sum_{j\leq J}\sum_{k\leq K}\sum_{i\leq I}(\sum_{\Psi\in\{\rho,\sigma,\betab,\alphab\}}||\nab_3^j\nab_4^k \nab^i\Psi||_{L^2(\Hb_0)}+\sum_{\Psi\in\{\beta,\rho,\sigma,\betab\}}||\nab_3^i \nab_4^k\nab^j\Psi||_{L^2(H_0)})\leq D_{I,J,K}. $$
holds initially for some $I,J,K$. Then, the following bounds hold for $0\leq u\leq u_*$, $0\leq \ub\leq \ub_*$,
$$\sum_{j\leq J}\sum_{k\leq K}\sum_{i\leq I}(\sum_{\Psi\in\{\rho,\sigma,\betab,\alphab\}}||\nab_3^j\nab_4^k \nab^i\Psi||_{L^2(\Hb_{\ub})}+\sum_{\Psi\in\{\beta,\rho,\sigma,\betab\}}||\nab_3^j \nab_4^k\nab^i\Psi||_{L^2(H_u)})\leq D'_{I,J,K}. $$
for $u_*,\ub_*\leq \epsilon$, where $\epsilon$ can be chosen as in Theorem \ref{timeofexistence} and the constant $D'_{I,J,K}$ depends only on the size of the initial data norm in the assumption of Theorem \ref{timeofexistence} and $D_{I,J,K}$.
\end{proposition}
\begin{proof}
Since the arguments are similar to those in Section \ref{energyestimatessec}, we will only provide a sketch. We will proceed with an induction argument in $i$, $j $ and $k$ and we will only consider the highest order terms, assuming appropriate control on the lower order terms. From \cite{KN}, we have the commutation formula
$$[\nab_3,\nab_4]f=-2\omega \nab_3 f+2\omegab\nab_4 f+4\zeta\cdot\nab f$$
for all scalar functions $f$. Therefore, by commuting $\nab_3^j\nab_4^k\nab^i$ with the equations \eqref{rcsc} and \eqref{bc}, we have
\begin{equation*}
\begin{split}
&\nabla_3\nab_3^j\nab_4^k\nabla^i\beta- \nabla \nab_3^j\nab_4^k\nabla^i\rhoc-\nabla\nab_3^j\nab_4^k\nabla^i\sigmac \\
\sim&\psi\nab_3^j\nab_4^k\nab^i(\beta,\rhoc,\sigmac,\betab)+\psi\psi\nab_3^j\nab_4^k\nab^{i+1}\psi+\mbox{terms with fewer derivatives},\\
\end{split}
\end{equation*}
and
\begin{equation*}
\begin{split}
&\nabla_4\nab_3^j\nab_4^k\nabla^i\rhoc- \div\nab_3^j\nab_4^k\nabla^i\beta\\
\sim&\psi\nab_3^j\nab_4^k\nab^i(\beta,\rhoc,\sigmac,\betab)+\psi\psi\nab_3^j\nab_4^k\nab^{i+1}\psi+\mbox{terms with fewer derivatives},\\
\end{split}
\end{equation*}
and 
\begin{equation*}
\begin{split}
&\nabla_4\nab_3^j\nab_4^k\nabla^i\sigmac- \div^*\nab_3^j\nab_4^k\nabla^i\beta\\
\sim&\psi\nab_3^j\nab_4^k\nab^i(\beta,\rhoc,\sigmac,\betab)+\psi\psi\nab_3^j\nab_4^k\nab^{i+1}\psi+\mbox{terms with fewer derivatives}.
\end{split}
\end{equation*}
Notice that since $\alpha$ does not appear in \eqref{rcsc} and \eqref{bc}, we do not have the term $\nab_3^j\nab_4^k\nab^i\alpha$ in these commuted equations. Therefore, integrating by parts as in Section \ref{energyestimatessec}, we have
\begin{equation*}
\begin{split}
&||\nab_3^j\nab_4^k \nab^i(\rhoc,\sigmac)||_{L^2(\Hb_{\ub})}^2+||\nab_3^j \nab_4^k\nab^i\beta||_{L^2(H_u)}^2\\
\leq &\mbox{Data}+\|\nab_3^j\nab_4^k\nab^i\Psi\psi\nab_3^j\nab_4^k\nab^i(\beta,\rhoc,\sigmac,\betab)\|_{L^1(\mathcal D_{u,\ub})}+\|\nab_3^j\nab_4^k\nab^i\Psi\psi\psi\nab_3^j\nab_4^k\nab^{i+1}\psi\|_{L^1(\mathcal D_{u,\ub})}\\
&+\|\nab_3^j\nab_4^k\nab^i\Psi\times(\mbox{terms with fewer derivative})\|_{L^1(\mathcal D_{u,\ub})}.
\end{split}
\end{equation*}
Similarly, using the other null Bianchi equations in \eqref{eq:null.Bianchi}, we also have
\begin{equation}\label{hr.1}
\begin{split}
&\sum_{\Psic\in\{\rhoc,\sigmac,\betab,\alphab\}}||\nab_3^j\nab_4^k \nab^i\Psic||_{L^2(\Hb_{\ub})}^2+\sum_{\Psi\in\{\beta,\rho,\sigma,\betab\}}||\nab_3^j \nab_4^k\nab^i\Psi||_{L^2(H_u)}^2\\
\leq &\mbox{Data}+\|\nab_3^j\nab_4^k\nab^i\Psi\psi\nab_3^j\nab_4^k\nab^i(\beta,\rhoc,\sigmac,\betab,\alphab)\|_{L^1(\mathcal D_{u,\ub})}+\|\nab_3^j\nab_4^k\nab^i\Psi\psi\psi\nab_3^j\nab_4^k\nab^{i+1}\psi\|_{L^1(\mathcal D_{u,\ub})}\\
&+\|\nab_3^j\nab_4^k\nab^i\Psi\times(\mbox{terms with fewer derivative})\|_{L^1(\mathcal D_{u,\ub})}.
\end{split}
\end{equation}
As in Section \ref{Ricciellipticsec}, the highest order derivatives of the Ricci coefficients $\nab_3^j\nab_4^k\nab^{i+1}\psi$ can be controlled by the highest order derivatives of the curvature components using elliptic estimates:
\begin{equation}\label{hr.2}
\begin{split}
&\sum_{\psi\in\{\trch,\chih,\omega,\trchb,\eta,\etab\}}\|\nab_3^j\nab_4^k\nab^{i+1}\psi\|_{L^2(H_u)}+\sum_{\psi\in\{\trch,\chih,\omegab,\trchb,\eta,\etab\}}\|\nab_3^j\nab_4^k\nab^{i+1}\psi\|_{L^2(\Hb_{\ub})}\\
\leq &\sum_{\Psic\in\{\rhoc,\sigmac,\betab,\alphab\}}||\nab_3^j\nab_4^k \nab^i\Psic||_{L^2(\Hb_{\ub})}+\sum_{\Psi\in\{\beta,\rho,\sigma,\betab\}}||\nab_3^j \nab_4^k\nab^i\Psi||_{L^2(H_u)}\\
&+\mbox{terms under control}.
\end{split}
\end{equation}
Using the boundedness of the $L^\infty$ norm of $\psi$ and the control of the terms with fewer derivatives that we have assumed as the induction hypothesis, \eqref{hr.1} and \eqref{hr.2} imply
\begin{equation}\label{hr.3}
\begin{split}
&\sum_{\Psic\in\{\rhoc,\sigmac,\betab,\alphab\}}||\nab_3^j\nab_4^k \nab^i\Psic||_{L^2(\Hb_{\ub})}^2+\sum_{\Psi\in\{\beta,\rho,\sigma,\betab\}}||\nab_3^j \nab_4^k\nab^i\Psi||_{L^2(H_u)}^2\\
\leq &\mbox{Data term}+\mbox{Bounded terms}+\int_0^{\ub} \sum_{\Psic\in\{\rhoc,\sigmac,\betab,\alphab\}}||\nab_3^j\nab_4^k \nab^i\Psic||_{L^2(\Hb_{\ub_*})}^2 d\ub_*\\
&+\int_0^u\sum_{\Psi\in\{\beta,\rho,\sigma,\betab\}}||\nab_3^j \nab_4^k\nab^i\Psi||_{L^2(H_{u_*})}^2 du_*\\
&+\int_0^{\ub}\int_0^u\sum_{\Psi\in\{\beta,\rho,\sigma,\betab\}}||\nab_3^j \nab_4^k\nab^i\Psi||_{L^2(H_{u_*})}\sum_{\Psic\in\{\rhoc,\sigmac,\betab,\alphab\}}||\nab_3^j\nab_4^k \nab^i\Psic||_{L^2(\Hb_{\ub_*})} du_*d\ub_*.
\end{split}
\end{equation}
To obtain the desired conclusion, introduce the bootstrap assumption
$$\sum_{\Psic\in\{\rhoc,\sigmac,\betab,\alphab\}}||\nab_3^j\nab_4^k \nab^i\Psic||_{L^2(\Hb_{\ub}(0,u))}^2+\sum_{\Psi\in\{\beta,\rho,\sigma,\betab\}}||\nab_3^j \nab_4^k\nab^i\Psi||_{L^2(H_u(0,\ub))}^2\leq A e^{2A(u+\ub)}$$
for $A$ sufficiently large. Then \eqref{hr.3} implies that we have
$$\sum_{\Psic\in\{\rhoc,\sigmac,\betab,\alphab\}}||\nab_3^j\nab_4^k \nab^i\Psic||_{L^2(\Hb_{\ub}(0,u))}^2+\sum_{\Psi\in\{\beta,\rho,\sigma,\betab\}}||\nab_3^j \nab_4^k\nab^i\Psi||_{L^2(H_u(0,\ub))}^2\leq C e^{2A(u+\ub)}$$
for some constant $C$ depending on the data term and the bounded terms. Therefore, this improves the bootstrap assumption for $A$ sufficiently large. Finally, notice that since we can control the derivatives of the Ricci coefficients by the derivatives of the renormalized curvature components, we in fact have the bound for the curvature components themselves. This concludes the proof of the proposition.
\end{proof}
\begin{remark}
Notice that for approximating solutions for an impulsive gravitational wave, the constants in the initial norms $D_{I,J,K}$ depend on $n$. However, the initial data satisfy uniform estimates for each $I\geq 2,J\geq 0$ with $K=0$. Therefore, the corresponding uniform estimates will hold in the region $0\leq u\leq u_*$ and $0\leq\ub\leq \ub_*$.
\end{remark}

\subsection{End of Proof of Theorem \ref{timeofexistence}}\label{EndofProof}
Once we have the estimates for all the higher derivatives, the proof of Theorem \ref{timeofexistence} is standard. We refer, for example, to Section 6 of \cite{L} for details.

\subsection{Additional Estimates for an Impulsive Gravitational Wave}\label{AddEst}
In this Subsection, we focus our attention on the sequence of spacetimes $(\mathcal M_n, g_n)$ with characteristic initial data converging to that of an impulsive gravitational wave as constructed in Section \ref{initialcondition}. The conclusion of Theorem \ref{timeofexistence} can be applied to $(\mathcal M_n, g_n)$. The initial data of the sequence of spacetimes, however, possess an additional property: $\alpha_n$ and its angular derivatives are bounded uniformly in $L^1_{\ub}$. In this Subsection, we prove that in the constructed spacetimes $(\mathcal M_n, g_n)$, the same estimates hold for all $u\leq u_*$ with $u_*\leq \epsilon$. Moreover, we show that $\alpha_n$ concentrates around $\ub_s$ and if $\alpha_n$ is initially more regular uniformly in $n$ away from $\ub_s$, it is also uniformly more regular away from $\ub_s$ for $0\leq u\leq u_*\leq \epsilon$.

We begin by proving energy estimates for $\alpha_n$ in $L^2(H_u)$ and $\beta_n$ in $L^2(\Hb_{\ub})$ using the pair of Bianchi equations for $(\nab_3\alpha_n,\nab_4\beta_n)$ in (\ref{eq:null.Bianchi}).
\begin{proposition}\label{alphaenergy}
The following estimate holds:
\begin{equation*}
\begin{split}
&\int_{H_u} |\nabla^i\alpha_n|^2_\gamma+\int_{\Hb_{\ub}} |\nabla^i\beta_n|^2_\gamma  \\
\leq &\int_{H_{u'}} |\nabla^i\alpha_n|^2_\gamma+\int_{\Hb_{\ub'}} |\nabla^i\beta|^2_\gamma+\int_{D_{u,\ub}}\nabla^i(\alpha_n,\beta_n)\sum_{i_1+i_2+i_3+i_4=i}\nabla^{i_1}\psi^{i_2}\nabla^{i_3}\psi\nabla^{i_4}\Psi\\
&+\int_{D_{u,\ub}}\nabla^i(\alpha_n,\beta_n)\sum_{i_1+i_2+i_3+i_4=i}\nabla^{i_1}\psi^{i_2}\nabla^{i_3}K\nabla^{i_4}\Psi,
\end{split}
\end{equation*}
where here, unlike in other places, we use $\Psi$ to denote all possible curvature components, including $\alpha_n$.
\end{proposition}
\begin{proof}
Consider the following Bianchi equations:
\begin{equation*}
\begin{split}
&\nabla_3\alpha+\frac 12 tr\chib \alpha=\nabla\hot \beta+ 4\omegab\alpha-3(\chih\rho+^*\chih\sigma)+
(\zeta+4\eta)\hot\beta,\\
&\nabla_4\beta+2tr\chi\beta = \div\alpha - 2\omega\beta +  \eta \alpha,\\
\end{split}
\end{equation*}
Using Proposition \ref{intbypartssph},
\begin{equation*}
\begin{split}
\int <\alpha_n,\nabla_3\alpha_n>_\gamma =&\int <\alpha_n,\nabla\hot\beta_n>_\gamma+<\alpha_n,\psi\Psi>_\gamma \\
=&\int -<\div\alpha_n,\beta_n>_\gamma+<\alpha_n,\psi\Psi>_\gamma \\
=&\int -<\nabla_4\beta_n,\beta_n>_\gamma+<\alpha_n,\psi\Psi>_\gamma +<\beta_n,\psi\Psi>_\gamma\\
\end{split}
\end{equation*}
Integrate by parts using Proposition \ref{intbyparts34} to get that for $u\geq u'$, $\ub\geq \ub'$,
\begin{equation*}
\begin{split}
\int_{H_u} |\alpha_n|^2_\gamma+\int_{\Hb_{\ub}} |\beta_n|^2_\gamma \leq \int_{H_{u'}} |\alpha_n|^2_\gamma+\int_{\Hb_{\ub'}} |\beta_n|^2_\gamma+\int_{D_{u,\ub}}<(\alpha_n,\beta_n),\psi\Psi>_\gamma\\
\end{split}
\end{equation*}
We use the commutation formula, and note that the special structure is preserved in the highest order:
\begin{equation*}
\begin{split}
&\nabla_3\nabla^i\alpha- \nabla\hot \nabla^i\beta \\
\sim&\sum_{i_1+i_2+i_3+i_4=i}\nabla^{i_1}\psi^{i_2}\nabla^{i_3}\psi\nabla^{i_4}\alpha \\
&+\sum_{i_1+i_2+i_3+i_4+i_5=i-1}\nabla^{i_1}\psi^{i_2}\nabla^{i_3}K^{i_4}\nabla^{i_5}\beta+\sum_{i_1+i_2+i_3=i}\nabla^{i_1}\psi^{i_2}\nabla^{i_3}(\psi\Psi),\\
&\nabla_4\nabla^i\beta- \div\nabla^i\alpha\\
\sim&\sum_{i_1+i_2+i_3+i_4=i}\nabla^{i_1}\psi^{i_2}\nabla^{i_3}\psi\nabla^{i_4}\beta\\
&+\sum_{i_1+i_2+i_3+i_4+i_5=i-1}\nabla^{i_1}\psi^{i_2}\nabla^{i_3}K^{i_4}\nabla^{i_5}\alpha+\sum_{i_1+i_2+i_3=i}\nabla^{i_1}\psi^{i_2}\nabla^{i_3}(\psi\Psi),\\
\end{split}
\end{equation*}
Performing the integration by parts as before using Propositions \ref{intbyparts34} and \ref{intbypartssph}, we have
\begin{equation*}
\begin{split}
&\int_{H_u} |\nabla^i\alpha_n|^2_\gamma+\int_{\Hb_{\ub}} |\nabla^i\beta_n|^2_\gamma  \\
\leq &\int_{H_{u'}} |\nabla^i\alpha_n|^2_\gamma+\int_{\Hb_{\ub'}} |\nabla^i\beta_n|^2_\gamma+\int_{D_{u,\ub}}\nabla^i(\alpha_n,\beta_n)\sum_{i_1+i_2+i_3+i_4=i}\nabla^{i_1}\psi^{i_2}\nabla^{i_3}\psi\nabla^{i_4}\Psi\\
&+\int_{D_{u,\ub}}\nabla^i(\alpha_n,\beta_n)\sum_{i_1+i_2+i_3+i_4=i}\nabla^{i_1}\psi^{i_2}\nabla^{i_3}K\nabla^{i_4}\Psi.
\end{split}
\end{equation*}
\end{proof}

\subsubsection{Before the Impulse}

We first apply Proposition \ref{alphaenergy} for the region $0\leq \ub\leq \ub_s$. Notice that the initial norm
$$\sum_{i\leq 2}\int_0^{\ub_s} ||\nabla^i\alpha_n||_{L^2(S_{0,\ub})}^2 d\ub\leq C$$
is bounded independent of $n$. Therefore,
\begin{proposition}\label{beforeshock}
$$\sum_{i\leq 2}(\sup_{u\leq \epsilon}||\nabla^i\alpha_n||_{L^2(H_u(0,\ub_s))}^2+\sup_{\ub\leq \ub_s}||\nabla^i\beta_n||^2_{L^2(\Hb_{\ub})}) \leq C(\mathcal O_0,\mathcal R_0),$$
independent of $n$.
\end{proposition}
\begin{proof}
By Proposition \ref{alphaenergy} and all the estimates in the previous section, we have
\begin{equation*}
\begin{split}
&\sum_{i\leq 2}(||\nabla^i\alpha_n||_{L^2(H_u(0,\ub_s))}^2+\sup_{\ub\leq \ub_s}||\nabla^i\beta_n||^2_{L^2(\Hb_{\ub})})\\
 \leq &C(\mathcal O_0,\mathcal R_0)(1+||\nabla^i\alpha_n||_{L^2(H)}+\int_0^u ||\nabla^i\alpha_n||^2_{L^2(H_{u'})} du').
\end{split}
\end{equation*}
By Gronwall's inequality, we have
$$\sum_{i\leq 2}(\sup_{u\leq \epsilon}||\nabla^i\alpha_n||_{L^2(H_u(0,\ub_s))}^2+\sup_{\ub\leq \ub_s}||\nabla^i\beta_n||^2_{L^2(\Hb_{\ub})}) \leq C(\mathcal O_0,\mathcal R_0),$$
as claimed.
\end{proof} 

\subsubsection{The Impulse Region}

For an impulsive gravitational wave, the initial data represented by the $\alpha$ component of curvature and its derivative are not in $L^2$. By the construction of Section \ref{initialcondition}, the approximating sequence $\alpha_n$ satisfies initially 
$$\sum_{i\leq 2}\int_{\ub_s}^{\ub_s+2^{-n}} ||\nabla^i\alpha_n||_{L^2(S_{0,\ub})}^2 d\ub\leq C2^{n}.$$
In the following Proposition, we show that this bound is propagated.
\begin{proposition}\label{duringshock}
$$\sum_{i\leq 2}(\sup_{u\leq \epsilon}||\nabla^i\alpha_n||_{L^2(H_u(\ub_s,\ub_s+2^{-n}))}^2+\sup_{\ub_s\leq \ub\leq \ub_s+2^{-n}}||\nabla^i\beta_n||^2_{L^2(\Hb_{\ub})}) \leq C(\mathcal O_0,\mathcal R_0)2^{n},$$
where $C(\mathcal O_0,\mathcal R_0)$ is independent of $n$.
\end{proposition}
\begin{proof}
The proof is the same as Proposition \ref{beforeshock}, except that the initial data for $\nabla^i\alpha$ can only be bounded by $2^{\frac n2}$ in $L^2(H_0(\ub_s,\ub_s+2^{-n}))$.
\end{proof}

\subsubsection{After the Impulse}

For $\ub_s+2^{-n}\leq \ub\leq \ub_*$, since the initial data satisfy $$\sum_{i\leq 2}\int_{\ub_s+2^{-n}}^{\ub_*} ||\nabla^i\alpha_n||_{L^2(S_{0,\ub})}^2 d\ub\leq C,$$
independent of $n$, we again have a better estimate. An extra challenge arises from the fact that the estimates derived in Proposition \ref{duringshock} depends on $n$. This can be overcome if we allow a loss in derivative:
\begin{proposition}\label{aftershock}
$$\sum_{i\leq 1}(\sup_{u\leq u_*}||\nabla^i\alpha_n||_{L^2(H_u(\ub_s+2^{-n},\ub_*))}^2+\sup_{\ub_s+2^{-n}\leq\ub\leq \ub_*}||\nabla^i\beta_n||^2_{L^2(\Hb_{\ub})}) \leq C(\mathcal O_0,\mathcal R_0),$$
independent of $n$.
\end{proposition}
\begin{proof}
To bound the left hand side, we can apply the argument as in Proposition \ref{beforeshock} to get
$$\sum_{i\leq 1}(\sup_{u\leq u_*}||\nabla^i\alpha_n||_{L^2(H_u(\ub_s+2^{-n},\ub_*))}^2+\sup_{\ub_s+2^{-n}\leq\ub\leq \ub_*}||\nabla^i\beta_n||^2_{L^2(\Hb_{\ub})}) \leq C(\mathcal O_0,\mathcal R_0)2^{\frac n2}.$$
The problem in proving Proposition \ref{aftershock} with a bound independent of $n$ is that according to Proposition \ref{duringshock}, $||\nab^i\beta_n||_{L^2(\Hb_{\ub_s})}$ may grow in $n$. However, we can still obtain the desired bounds for $\alpha_n$ and $\nab\alpha_n$ (note that the argument below will not apply to $\nab^2\alpha_n$ unless extra assumptions on the initial data are used) by using Proposition \ref{RS}, which implies that
$$\sum_{i\leq 1}\sup_{u,\ub}||\nab^i\beta_n||_{L^2(S_{u,\ub})}\leq C(\mathcal O_0,\mathcal R_0).$$
Taking $\ub=\ub_s+2^{-n}$ and integrating, we thus have
$$\sum_{i\leq 1}||\nabla^i\beta_n||^2_{L^2(\Hb_{\ub_s+2^{-n}})}\leq C(\mathcal O_0,\mathcal R_0).$$
\end{proof}

We can also prove estimates for higher derivatives away from $\ub_s$. The influence of a propagating curvature impulse results in a loss of derivatives in such estimates compared to the assumed regularity of the initial data. By the remark after Proposition \ref{propagationregularity}, we have estimates, uniform in $n$, for the $\nab_3$ and $\nab$ derivatives for $\psi$ and $\Psi$. We now show that we can also have uniform in $n$ estimates for the $\nab_4$ derivatives for $\psi$ and $\Psi$ for $\ub\geq\ub_s+2^{-n}$. By the Bianchi equations (\ref{eq:null.Bianchi}), this in turn can be reduced to showing estimates for $\nab_4$ derivatives of $\alpha$. This is proved in the following Proposition:
\begin{proposition}\label{alphaapriori}
For $\ub_s+2^{-n}\leq\ub\leq\epsilon$, if the initial data set satisfies
$$\sum_{k\leq K}\sum_{i\leq I}\sum_{\Psi\in\{\rho,\sigma,\betab,\alphab\}}||\nab_4^k \nab^i\Psi_n||_{L^2(\Hb_0)}\leq C_{I,K}. $$
$$\sum_{k\leq K}\sum_{i\leq I}\sup_{\ub}\sum_{\Psi\in\{\beta,\rho,\sigma,\betab\}}|| \nab_4^k\nab^i\Psi_n||_{L^2(S_{0,\ub})}\leq C_{I,K}.$$
then
$$\sum_{i\leq I-1}\sum_{k\leq \min\{K,\lfloor\frac{I-i-1}{2}\rfloor\}}(\sup_{\ub}||\nab_4^k\nabla^i\alpha_n||_{L^2(H_u(\ub_s+2^{-n},\epsilon))}+\sup_{\ub_s+2^{-n}\leq \ub\leq \epsilon}||\nab_4^k\nabla^i\beta_n||_{L^2(\Hb_{\ub})})\leq C_{I,K}',$$
where $C_{I,K}$ is independent of $n$.
\end{proposition}
\begin{proof}
The main step is to prove that on $\Hb_{\ub_s+2^{-n}}$, we have the estimate
$$||\nab_4^k\nab^i\beta_n||_{L^2(\Hb_{\ub_s+2^{-n}})}\leq C$$
independent of $n$. For $k=0$, we can commute with angular derivatives and integrate along the $u$ direction using the Bianchi equation
$$\nab_3\beta+\trchb\beta=\nab\rho+2\omegab\beta+^*\nab\sigma+2\chih\cdot\betab+3(\eta\rho+^*\eta\sigma).$$
Since the right hand side has one more angular derivative, this integration loses a derivative, allowing us only to prove
$$\sum_{i\leq I-1}||\nab^i\beta_n||_{L^2(\Hb_{\ub_s+2^{-n}})}\leq C$$

Now we consider also the $\nab_4$ derivatives. Differentiate the equation by $\nab_4^k\nab_i$ and commute $[\nab_3,\nab_4^k\nab_i]$ on the left hand side. Moreover, except for the component $\alpha$, whenever we see $\nab_4^{k_1}\Psi$, we substitute an appropriate Bianchi equation from (\ref{eq:null.Bianchi}). Then we get an equation
$$\nab_3\nab_4^k\nab^i\beta=\nab_4^{k-1}\nab^{i+2}\beta+...$$
where $...$ denote terms that are lower order in terms of derivatives. We then integrate this along the $u$ direction. Thus
$$||\nab_4^k\nab^i\beta_n||_{L^2(\Hb_{\ub_s+2^{-n}})}\leq C||\nab_4^{k-1}\nab^{i+2}\beta_n||_{L^2(\Hb_{\ub_s+2^{-n}})}+...$$
Inducting in $k$, we get that
$$||\nab_4^k\nab^i\beta_n||_{L^2(\Hb_{\ub_s+2^{-n}})}\leq C||\nab^{i+2k}\beta_n||_{L^2(\Hb_{\ub_s+2^{-n}})}+... $$
Now we have reduced to the case where there are only angular derivatives falling on $\beta_n$. By the above, we need $i+2k+1\leq I$. Thus we can prove 
\begin{equation}\label{smoothnessdata}
\sum_{i\leq I-1}\sum_{k\leq \min\{K,\lfloor\frac{I-i-1}{2}\rfloor\}}||\nab_4^k\nab^i\beta_n||_{L^2(\Hb_{\ub_s+2^{-n}})}\leq C.
\end{equation}
The conclusion thus follows from a standard energy estimate type argument as in Proposition \ref{alphaenergy} and using (\ref{smoothnessdata}) as the initial data on $\Hb_{\ub_s+2^{-n}}$
\end{proof}

Combining the estimates in Propositions \ref{beforeshock}, \ref{duringshock}, \ref{alphaapriori} and Cauchy-Schwarz, we have the following uniform $L^1_{\ub}$ estimate.
\begin{proposition}\label{totalvariation}
$$\sup_{u\leq u_*}\sum_{i\leq I}\int_0^{\ub_*} ||\nabla^i\alpha_n||_{L^2(S_{u,\ub})}d\ub\leq C_I,$$
where $C_I$ is independent of $n$.
\end{proposition}
Proposition \ref{totalvariation} is crucial in showing that the limiting spacetime will have $\alpha$ defined as a finite measure with a singular atom at $\ub_s$.

\section{Convergence}\label{convergence}
In this Section, we show that a sequence of initial data satisfying uniform estimates of Theorem \ref{rdthmv2} with converging initial data gives rise to a sequence of converging spacetimes. The convergence will be understood as follows: the spacetime will be identified in the system of double null coordinates $(u,\ub,\th^1,\th^2)$ and convergence will be established for the sequence of the corresponding spacetime metrics. In view of the quasilinear nature of the Einstein equations, we can only hope to prove convergence of our approximating spacetimes in a norm with one derivative fewer than the a priori estimates that we established for them. 

To get estimates of metric, we use the fact that the metric components satisfy inhomogeneous transport equations with right hand side expressed as Ricci coefficients. The estimates for the difference of Ricci coefficients and curvature components are derived by considering the system of difference equations obtained from the original system of transport, elliptic and Bianchi equations. As was the case for the a priori estimates, the challenge for the difference system is a lack of any a priori control of the difference of the singular curvature components $\alpha$. For the approximating sequence of initial data, $\alpha_n$ is not a Cauchy sequence in $L^2(H_0)$. In the proof of the a priori estimates, we handled the lack of information of the $\alpha$ component of curvature via a renormalization procedure and the observation that the system satisfied by the Ricci coefficients and the renormalized curvature components can be closed without any reference to the $\alpha$ component of curvature. Potentially, this property may fail when we consider the difference system. It is however a remarkable fact as we will show below that the difference equations still possess the same structure.

The following is the main Theorem in which we estimate the difference of the metrics, Ricci coefficients and curvature components of two spacetimes:

\begin{theorem}\label{convergencethm}
Suppose we have two sets of initial data $(1)$ and $(2)$ satisfying the conditions in Theorem \ref{timeofexistence} with the same constants $C$ and $c$. By Theorem \ref{timeofexistence}, we can solve for vacuum spacetimes $(\mathcal M^{(1)}, g^{(1)})$ and $(\mathcal M^{(2)}, g^{(2)})$ corresponding to the initial data sets $(1)$ and $(2)$ in the region $0\leq u\leq u_*$ and $0\leq \ub\leq\ub_*$ for $u_*,\ub_*\leq \epsilon$ . Let $(u,\ub,\th^1,\th^2)$ be the coordinate system introduced in Section \ref{coordinates} such that the metrics take the form
$$g^{(i)}=-2(\Omega^{(i)})^2(du\otimes d\ub+d\ub\otimes du)+(\gamma^{(i)})_{AB}(d\th^A-(b^{(i)})^Adu)\otimes (d\th^B-(b^{(i)})^Bdu),$$ 
where $\Omega=1$ and $b^{A}=0$ on $H_0$ and $\Hb_0$. We can now identify the two spacetimes by identifying points with the same value of coordinate functions. Define $g'=g^{(1)}-g^{(2)}$, $\psi'=\psi^{(1)}-\psi^{(2)}$ and $\Psi'=\Psi^{(1)}-\Psi^{(2)}$ to be the difference of the metric, the difference of the Ricci coefficients and the difference of the curvature components respectively. If the data satisfy 
$$\sup_u|(\frac{\partial}{\partial\th})^i\gamma_{AB}'(\ub=0)|+\sup_{\ub}|(\frac{\partial}{\partial\th})^i\gamma_{AB}'(u=0)|\leq a,$$
$$\sum_{\psi\neq\chih,\omega}(\sum_{i\leq 1} \sup_{u}||\nabla^i\psi'||_{L^2(S_{u,0})}+\sum_{i\leq 1} \sup_{\ub}||\nabla^i\psi'||_{L^2(S_{0,\ub})})\leq a,$$
$$\sum_{i\leq 1} ||\nabla^i(\chih',\omega')||_{L^{p_0}_{\ub}L^2(S_{0,\ub})}\leq a\quad\mbox{for some fixed }2\leq p_0<\infty,$$
$$||\nabla^2(\chih',\omega',\eta',\etab')||_{L^2(H_0)}+||\nabla^2(\chibh',\omegab',\eta',\etab')||_{L^2(\Hb_{0})}\leq a,$$
$$\sup_{u} ||\nabla^2(\trch',\trchb')||_{L^2(S_{u,0})}+\sup_{\ub} ||\nabla^2(\trch',\trchb')||_{L^2(S_{0,\ub})}\leq a,$$
$$\sum_{i\leq 1}\left(\sum_{\Psi\in\{\beta,\rho,\sigma,\betab\}} ||\nabla^i\Psi'||_{L^{2}_{\ub}L^2(S_{0,\ub})}+\sum_{\Psi\in\{\rho,\sigma,\betab,\alphab\}} ||\nabla^i\Psi'||_{L^{2}_{u}L^2(S_{u,0})}\right)\leq a,$$
where the angular covariant derivative and all the norms are defined with respect to the spacetime $(1)$. Then the following estimates hold in $\{0\leq u\leq u_*\}\cap\{0\leq\ub\leq\ub_*\}$:
$$\sup_{u,\ub}|(\frac{\partial}{\partial\th})^i\gamma_{AB}'|\leq C'a,$$
$$\sum_{\psi\neq\chih,\omega}\sum_{i\leq 1} \sup_{u,\ub}||\nabla^i\psi'||_{L^2(S_{u,\ub})}\leq C'a,$$
$$\sum_{i\leq 1} \sup_u ||\nabla^i(\chih',\omega')||_{L^{p_0}_{\ub}L^2(S_{u,\ub})}\leq C'a,$$
$$\sup_u ||\nabla^2(\chih',\omega',\eta',\etab')||_{L^2(H_u)}+\sup_{\ub} ||\nabla^2(\chibh',\omegab',\eta',\etab')||_{L^2(\Hb_{\ub})}+\sup_{u,\ub} ||\nabla^2(\trch',\trchb')||_{L^2(S_{u,\ub})}\leq C'a,$$
$$\sum_{i\leq 1}\left(\sum_{\Psi\in\{\beta,\rho,\sigma,\betab\}} \sup_u||\nabla^i\Psi'||_{L^2_{\ub}L^2(S_{u,\ub})}+\sum_{\Psi\in\{\rho,\sigma,\betab,\alphab\}} \sup_{\ub}||\nabla^i\Psi'||_{L^2_{u}L^2(S_{u,\ub})}\right)\leq C'a,$$
for some constant $C'$ depending only on $C$ and $c$ and independent of $a$.
\end{theorem}
By iterating Theorem \ref{convergencethm}, we can reduce its proof to the region $0\leq u\leq \delta$, $0\leq \ub\leq \delta$ for some sufficiently small $\delta$ independent of $a$. Most of this Section will be devoted to a proof of Theorem \ref{convergencethm} for $0\leq u\leq\delta$ and $0\leq \ub\leq\delta$. The proof of Theorem \ref{convergencethm} will be carried out in Sections \ref{convsec1}, \ref{convsec2}, \ref{convsec3}, \ref{convsec4} and \ref{convsec5}.

Theorem \ref{convergencethm} implies the following convergence result:
\begin{theorem}\label{convergencethm2}
Suppose we have a sequence of initial data that coincide on $\Hb_0$ and satisfy the assumptions of Theorem \ref{timeofexistence} with uniform constants $C$ and $c$ on $H_0$ and $\Hb_0$. By Theorem \ref{timeofexistence}, for every initial data in the sequence, a unique smooth solution to the vacuum Einstein equations $(\mathcal M_n, g_n)$ exists in $0\leq u\leq u_*$, $0\leq \ub\leq \ub_*$ for $u_*, \ub_*\leq \epsilon$ and the metric takes the following form in the coordinate system $(u,\ub,\th^1,\th^2)$:
$$g_n=-2(\Omega_n)^2(du\otimes d\ub+d\ub\otimes du)+(\gamma_n)_{AB}(d\th^A-(b_n)^Adu)\otimes (d\th^B-(b_n)^Bdu),$$ 
where $\Omega=1$ and $b^A=0$ on $H_0$ and $\Hb_0$. Denote the Ricci coefficients and curvature components by $\psi_n$ and $\Psi_n$ respectively. Identify the spacetimes in the sequence by the value of the coordinate functions $(u,\ub,\th^1,\th^2)$. Define also $g_n'=g_n-g_{n-1}$, $\psi'_n=\psi_n-\psi_{n-1}$ and $\Psi'_n=\Psi_n-\Psi_{n-1}$. If
$$\sup_u|(\frac{\partial}{\partial\th})^i(\gamma_{AB})'_n(\ub=0)|\leq a_n,$$
$$\sum_{\psi\neq\chih,\omega}\sum_{i\leq 1} \sup_{\ub}||\nabla^i\psi'_n||_{L^2(S_{0,\ub})}\leq a_n,$$
$$\sum_{i\leq 1} ||\nabla^i(\chih_n',\omega_n')||_{L^{p_0}_{\ub}L^2(S_{0,\ub})}\leq a_n\quad\mbox{for some fixed }2\leq p_0<\infty,$$
$$||\nabla^2(\chih_n',\omega_n',\eta_n',\etab_n')||_{L^2(H_0)}+\sup_{\ub} ||\nabla^2(\trch_n',\trchb_n')||_{L^2(S_{0,\ub})}\leq a_n,$$
$$\sum_{i\leq 1}\left(\sum_{\Psi\in\{\beta,\rho,\sigma,\betab\}} ||\nabla^i\Psi_n'||_{L^{2}_{\ub}L^2(S_{0,\ub})}+\sum_{\Psi\in\{\rho,\sigma,\betab,\alphab\}} ||\nabla^i\Psi_n'||_{L^{2}_{u}L^2(S_{u,0})}\right)\leq a_n,$$
for some $a_n$ such that $\displaystyle\sum a_n<\infty$,
then the spacetime metrics converge uniformly to a continuous limiting spacetime metric $g_\infty$
$$g_\infty=-2(\Omega_\infty)^2(du\otimes d\ub+d\ub\otimes du)+ (\gamma_{\infty})_{AB}(d\th^A-(b_\infty)^Adu)\otimes(d\th^B-(b_\infty)^Bdu)$$
in the region $0\leq u\leq u_*$, $0\leq \ub\leq \ub_*$.
Moreover, 
$$(\frac{\partial}{\partial \th}g_n,\frac{\partial}{\partial u}g_n)\mbox{ converge to }(\frac{\partial}{\partial \th}g_\infty,\frac{\partial}{\partial u}g_\infty)\mbox{ in }L^\infty_u L^\infty_{\ub} L^4(S),$$
$$(\frac{\partial^2}{\partial \th^2}g_n,\frac{\partial^2}{\partial u\partial\th}g_n,\frac{\partial^2}{\partial u^2}g_n)\mbox{ converge to }(\frac{\partial^2}{\partial \th^2}g_\infty,\frac{\partial^2}{\partial u\partial\th}g_\infty,\frac{\partial^2}{\partial u^2}g_{\infty})\mbox{ in }L^\infty_u L^\infty_{\ub} L^2(S),$$
$$(\frac{\partial}{\partial \ub}g_n, \frac{\partial}{\partial\ub}((\gamma_n^{-1})^{AB}\frac{\partial}{\partial\ub}(\gamma_n)_{AB}))\mbox{ converge to }(\frac{\partial}{\partial \ub}g_{\infty}, \frac{\partial}{\partial\ub}((\gamma_{\infty}^{-1})^{AB}\frac{\partial}{\partial\ub}(\gamma_{\infty})_{AB}))\mbox{ in }L^\infty_u L^{p_0}_{\ub} L^\infty(S),$$
$$(\frac{\partial^2}{\partial \th \partial \ub}g_n,\frac{\partial^2}{\partial u\partial\ub}g_n,\frac{\partial^2}{\partial\ub^2}(b^A)_n)\mbox{ converge to }(\frac{\partial^2}{\partial \th \partial \ub}g_\infty,\frac{\partial^2}{\partial u\partial\ub}g_\infty,\frac{\partial^2}{\partial\ub^2}(b^A)_\infty)\mbox{ in }L^\infty_u L^{p_0}_{\ub} L^4(S).$$
As a consequence, in the limiting spacetime,
$$\frac{\partial}{\partial \th}g_\infty,\frac{\partial}{\partial u}g_\infty\in C^0_u C^0_{\ub} L^4(S),$$
$$\frac{\partial^2}{\partial \th^2}g_\infty,\frac{\partial^2}{\partial u\partial\th}g_\infty,\frac{\partial^2}{\partial u^2}g_\infty\in C^0_u C^0_{\ub} L^2(S),$$
$$\frac{\partial}{\partial \ub}g_\infty, \frac{\partial}{\partial\ub}((\gamma_\infty^{-1})^{AB}\frac{\partial}{\partial\ub}(\gamma_\infty)_{AB}) \in L^\infty_u L^\infty_{\ub} L^\infty(S),$$
$$\frac{\partial^2}{\partial \th \partial \ub}g_\infty,\frac{\partial^2}{\partial u\partial\ub}g_\infty,\frac{\partial^2}{\partial \ub^2}(b^A)_\infty\in L^\infty_u L^\infty_{\ub} L^4(S).$$
\end{theorem}

\begin{remark}
Notice that in the case of initial data of the impulsive gravitational wave, we have $\nabla^i\chih_n$ converging in $L^p_{\ub} L^2(S)$ for any $2\leq p <\infty$, but not for $p=\infty$. We thus take a sequence of initial data that converges in a topology which is consistent with $\nab^i\chih_n$ converging in $L^p_{\ub}L^2(S)$.
\end{remark}

The convergence of the Ricci coefficients and the acceptable curvature components follows from Theorem \ref{convergencethm}. In Section \ref{limit}, we will show that the convergence of the Ricci coefficients and the acceptable curvature components imply the asserted convergence of the metric and the regularity property of the limiting spacetime. The above convergence theorem is strong enough to show that the limiting spacetime is a solution to the Einstein equations:
\begin{theorem}\label{Einstein}
Suppose all the assumptions of Theorem \ref{convergencethm2} hold. Then the limiting spacetime metric satisfies the Einstein equations in $L^\infty_u L^\infty_{\ub}L^2(S)$.
\end{theorem}
This will also be proved in Section \ref{limit}. Moreover, this limiting spacetime solution is the unique solution to the vacuum Einstein equations.
\begin{theorem}\label{uniquenessthm}
The solution to the Einstein equations given by Theorems \ref{convergencethm2} and \ref{Einstein} is unique among spacetimes that arise as $C^0$ limit of smooth solutions to the vacuum Einstein equations.
\end{theorem}
A more precise version of the uniqueness theorem is formulated as Proposition \ref{uniquenessprop} and will be proved in Section \ref{uniquenesssec}. Moreover, if the initial data is assumed to be more regular, the limit spacetime metric is more regular:
\begin{theorem}\label{regularitythm}
Suppose, in addition to the assumptions of Theorem \ref{convergencethm2}, the bounds
$$\sum_{j\leq J}\sum_{i\leq I}(\sum_{\Psi\in\{\rho,\sigma,\betab,\alphab\}}||\nab_3^j\nab^i\Psi_n||_{L^2(\Hb_0)}+\sum_{\Psi\in\{\beta,\rho,\sigma,\betab\}}||\nab_3^i \nab^j\Psi_n||_{L^2(H_0)})\leq C $$
hold uniformly independent of $n$.
Then 
$$\sum_{j\leq J+2}\sum_{i\leq \min\{I,j-2\}}(\frac{\partial}{\partial u})^j(\frac{\partial}{\partial \th})^i g_\infty\in L^\infty_u L^\infty_{\ub} L^2(S),$$
$$\sum_{j\leq J+2}\sum_{i\leq \min\{I,j-2\}}(\frac{\partial}{\partial u})^j(\frac{\partial}{\partial \th})^i\frac{\partial}{ \partial \ub}g_\infty\in L^\infty_u L^{p_0}_{\ub} L^2(S).$$
\end{theorem}
This will be proved in Section \ref{regularityp}. 

For general initial data satisfying the assumptions of Theorem \ref{rdthmv2}, using the construction in Section \ref{initialcondition}, there exists an approximating sequence of smooth initial data satisfying the assumption of Theorem \ref{convergencethm2}. Thus, the combination of Theorems \ref{timeofexistence}, \ref{convergencethm2}, \ref{Einstein}, \ref{regularitythm} and \ref{uniquenessthm} together imply Theorem \ref{rdthmv2}.

The remainder of this Section will be organized as follows: Theorem \ref{convergencethm} is proved in Sections \ref{convsec1}-\ref{convsec5}. After the definition of the norms in Section \ref{convsec1}, the proof is carried out in three steps:\\

\noindent{\bf STEP 1} (Section \ref{convsec2}): The difference of the metric components are estimated assuming the bounds for the difference of the Ricci coefficients.\\

\noindent{\bf STEP 2} (Section \ref{convsec4}): The difference of the Ricci coefficients are controlled assuming the estimates of the difference of the curvature components. This relies on the transport equations for the difference quantities derived in Section \ref{convsec3}.\\

\noindent{\bf STEP 3} (Section \ref{convsec5}): Finally, the bounds for the difference of curvature components are obtained, closing all the estimates for Theorem \ref{convergencethm}.\\

In Section \ref{AddEstMetric}, additional estimates are derived for the metric components. These additional estimates will be used together with Theorem \ref{convergencethm} to construct a limiting spacetime and to obtain Theorem \ref{convergencethm2} in Section \ref{limit}. In this Section, Theorem \ref{Einstein} is also proved, showing that the limiting spacetime satisfies the Einstein equations. In Section \ref{uniquenesssec}, we formulate and prove a precise version of Theorem \ref{uniquenessthm}, establishing uniqueness of the limiting spacetime. In Section \ref{regularityp}, Theorem \ref{regularitythm} is proved, showing additional regularity in the spacetime with more regular initial data. Finally, in Section \ref{limitgiw}, we return to the case of an impulsive gravitational wave and obtain extra regularity properties for these spacetimes.

\subsection{Norms}\label{convsec1}
We begin the proof of Theorem \ref{convergencethm}.
Define the following $L^2$ norms for the difference of the null curvature components or their renormalized versions:
$$\mathcal R'=\sum_{i\leq 1}\sup_u||\nabla^i(\beta',\rho',\sigma',\betab')||_{L^2(H_u)}+\sum_{i\leq 1}\sup_{\ub}||\nabla^i(\rhoc',\sigmac',\betab',\alphab')||_{L^2(\Hb_{\ub})}$$
Define the following norm for difference of the the Ricci coefficients:
$$\mathcal O'=\sum_{i\leq 1}\sup_{u}||\nabla^i(\hat{\chi}',\omega')||_{L^{p_0}_{\ub}L^2(S_{u,\ub})}+\sum_{i\leq 1}\sup_{u,\ub}||(\nabla^i(tr\chi',\eta',\underline{\eta}',\underline{\omega}',\underline{\hat{\chi}}',tr\underline{\chi}')||_{L^2(S_{u,\ub})}.$$
These norms will be used to estimate the difference of the null Ricci coefficients except for those involving the highest derivatives, for which the estimates are weaker. We therefore also introduce the norms:
\begin{equation*}
\begin{split}
\tilde{\mathcal O}'=\sup_u ||\nabla^2(\chih',\omega',\eta',\etab')||_{L^2(H_u)}+\sup_{\ub} ||\nabla^2(\chibh',\omegab',\eta',\etab')||_{L^2(\Hb_{\ub})}+\sup_{u,\ub} ||\nabla^2(\trch',\trchb')||_{L^2(S_{u,\ub})},
\end{split}
\end{equation*}
Notice that the norms $\mathcal R'$, $\mathcal O'$ and $\tilde{\mathcal O}'$ are difference counterparts of the norms $\mathcal R$, $\mathcal O$ and $\tilde{\mathcal O}$ for the difference quantities. Note, however, that the former provides control of one fewer derivatives than the latter. This is due to the fact that convergence will be proved in a norm which is one derivative weaker than the corresponding norms for the a priori estimates. Define also the following norms:
\begin{equation*}
\begin{split}
\mathcal O''=\sup_u(||(\chih,\omega,\eta,\underline{\eta},\trch,\trchb)'||_{L^2_{\ub}L^\infty(S_{u,\ub})}+||\nabla(\chih,\omega,\eta,\underline{\eta},\trch,\trchb)'||_{L^2_{\ub}L^4(S_{u,\ub})}).
\end{split}
\end{equation*}
It follows by Sobolev Embedding in Propositions \ref{L4} and \ref{Linfty} that
\begin{proposition}\label{OpSobolev}
$$\mathcal O''\leq C(\tilde{\mathcal O}'+\mathcal O').$$
\end{proposition}

\subsection{Estimates for the Difference of the Metrics}\label{convsec2}
In this section, we show that the difference of the metrics and their coordinate angular derivatives can be controlled by the $\mathcal O'$ and $\tilde{\mathcal O}'$ norms. Recall that the metrics take the form
$$g^{(i)}=-2(\Omega^{(i)})^2(du\otimes d\ub+d\ub\otimes du)+(\gamma^{(i)})_{AB}(d\th^A-(b^{(i)})^Adu)\otimes (d\th^B-(b^{(i)})^Bdu),$$ 
for $i=1,2$.
\begin{proposition}\label{Omegap}
\[
\sup_{u,\ub}|(\Omega',(\Omega^{-1})')(u,\ub)|\leq C\delta^{\frac 12}\mathcal O''.
\]
\end{proposition}
\begin{proof}
Recall that
\[
 \omega=-\frac{1}{2}\nabla_4\log\Omega=\frac{1}{2}\Omega\nabla_4\Omega^{-1}=\frac{1}{2}\frac{\partial}{\partial \ub}\Omega^{-1}.
\]
Hence 
\[
 \frac{1}{2}\frac{\partial}{\partial\ub}(\Omega^{-1})'=\omega'.
\]
By integrating along the $\ub$ direction, noticing that $\Omega'=0$ on $\Hb_0$, and using Cauchy-Schwarz we get
$$|(\Omega^{-1})'|\leq C\delta^{\frac 12}\mathcal O''.$$
In order to get the estimate for $\Omega'$, we note that
\[
(\Omega'+\Omega^{(2)})^{-1}=(\Omega^{(1)})^{-1}=(\Omega^{-1})^{(2)}+(\Omega^{-1})'.
\]
Therefore,
\[
 \Omega'=((\Omega^{(2)})^{-1}+(\Omega)^{-1})')^{-1}-\Omega^{(2)}=\frac{(\Omega^{-1})'}{1+\frac{(\Omega^{-1})'}{(\Omega^{-1})^{(2)}}}.
\]
In view of the upper and lower bounds of $\Omega$ in Proposition \ref{Omega},
$$|\Omega'|\leq C\delta^{\frac 12}\mathcal O''.$$

\end{proof}
Using the estimates for $\Omega'$, we also have estimates for $\gamma'$.
\begin{proposition}\label{gammap}
$\gamma'$ satisfies the following pointwise bounds:
$$\sup_{u,\ub}|(\gamma_{AB}',((\gamma^{-1})^{AB})')(u,\ub)|\leq a+C\delta^{\frac 12}\mathcal O''_u.$$
\end{proposition}
\begin{proof}
The components of $\gamma$ solve the following ODE:
$$\frac{\partial}{\partial \ub}\gamma_{AB}=2\Omega\chi_{AB}.$$
This implies
$$\frac{\partial}{\partial \ub}\log(\det\gamma)=\Omega\trch.$$
From this we can derive an equation for $(\det\gamma)'$:
$$\frac{\partial}{\partial \ub}(\det\gamma)'=\frac{1}{(\det\gamma)^{(2)}}\left(-(\det\gamma)'\frac{\partial}{\partial \ub}(\det\gamma)^{(1)}+(\det\gamma)^{(1)}(\det\gamma)^{(2)}(\Omega\trch)'\right). $$
By Proposition \ref{gamma}, we have uniform upper and lower bounds on $\det\gamma$ and uniform estimates for $\Omega$ and $\frac{\partial}{\partial \ub}(\det\gamma)$. The previous Proposition gives $|\Omega'|\leq C\delta^{\frac 12}\mathcal O''$. Moreover, $\int |\trch'| d\ub'\leq \delta\mathcal O''$ by definition. Thus,
\begin{equation}\label{volumeformd}
|(\det\gamma)'|\leq C\delta^{\frac 12}\mathcal O''.
\end{equation}
We can also derive an equation for $(\gamma_{AB})'$:
$$\frac{\partial}{\partial \ub}(\gamma_{AB})'=2(\Omega\chi_{AB})'.$$
Thus
$$|(\gamma_{AB})'(u,\ub)|\leq |(\gamma_{AB})'(u,0)|+C\int_0^{\ub}|(\chi_{AB})'|d\ub'+C\delta^{\frac 12}\mathcal O'' |\chi_{AB}|.$$
By assumption, $|(\gamma_{AB})'(u,0)|\leq a$. By Proposition \ref{gamma}, matrices $\gamma^{(1)}$ and $\gamma^{(2)}$ are uniformly non-degenerate matrices. Hence, regardless of whether we define the $L^\infty$ norm with respect to $\gamma^{(1)}$ or $\gamma^{(2)}$, we have
$$\sum_{A,B=1,2}|(\chih_{AB})'|\leq C\sup_{u,\ub}||\chi'||_{L^\infty(S_{u,\ub})}.$$
Therefore, 
\begin{equation}\label{metricd}
|(\gamma_{AB})'|\leq a+C\delta^{\frac 12}\mathcal O''.
\end{equation}
Now, (\ref{volumeformd}) and (\ref{metricd}) together also imply the pointwise bound in coordinates for 
$$|((\gamma^{-1})^{AB})'|\leq a+C\delta^{\frac 12}\mathcal O''.$$
\end{proof}
This estimate allows us to conclude that the $L^p$ norms defined with respect to either metric $(1)$ or $(2)$ differ only by $a+C\delta^{\frac 12}\mathcal O''$.
\begin{proposition}\label{normscomparable}
Given any tensor, we can define its $L^p(S)$ norms $||\phi||^{(1)}_{L^p(S_{u,\ub})}$ and $||\phi||^{(2)}_{L^p(S_{u,\ub})}$ with respect to the first and second metric. Suppose
$$||\phi||^{(1)}_{L^p(S_{u,\ub})}<\infty.$$
Then
$$|||\phi||^{(1)}_{L^p(S_{u,\ub})}-||\phi||^{(2)}_{L^p(S_{u,\ub})}|\leq C(a+\delta^{\frac 12}\mathcal O'')||\phi||^{(1)}_{L^p(S_{u,\ub})}.$$
In particular, $\phi$ is also in $L^p$ with respect to the second metric.
\end{proposition}
\begin{proof}
This follows from the pointwise control of $(\gamma_{AB})'$, $((\gamma^{-1})^{AB})'$ and $(\det\gamma)'$ in coordinates.
\end{proof}
In a similar manner, we can control the $(\frac{\partial}{\partial\th^A})$ derivatives of $\gamma'$.
\begin{proposition}\label{dgammap}
$$\sup_{u,\ub}||(\frac{\partial}{\partial\th^C})\gamma_{AB}'||_{L^4(S_{u,\ub})}\leq Ca+C\delta^{\frac 12}\mathcal O'',$$
where $L^4(S)$ is understood as the $L^4$ norm for a scalar function and by Proposition \ref{normscomparable} can be defined with respect either to the metric $\gamma^{(1)}$ or $\gamma^{(2)}$.
\end{proposition}
\begin{proof}
By the equation 
$$\frac{\partial}{\partial \ub}\gamma_{AB}=2\Omega\chi_{AB}$$
and the assumption, we have
\begin{equation}\label{gammadd}
\begin{split}
||(\frac{\partial}{\partial\th^C})^i\gamma_{AB}'(u,\ub)||_{L^4(S_{u,\ub'})}\leq& Ca+C\int_0^{\ub}\left(\|(\frac{\partial}{\partial\th^C})^i\Omega'||_{L^4(S_{u,\ub'})}+\|(\frac{\partial}{\partial\th^C})^i\chi'||_{L^4(S_{u,\ub'})}\right) d\ub'.
\end{split}
\end{equation}
Since
$$\nabla(\log \Omega)=\frac 12 (\eta+\etab),$$
we have
$$|(\frac{\partial}{\partial\th^C})\Omega'|\leq C|(\eta,\etab)'|+C|\Gamma'||\Omega|+C|\Gamma||\Omega'| \leq C|(\eta,\etab)'|+C|\frac{\partial}{\partial\th^D}\gamma_{AB}'|+C|\Omega'|,$$
where $\Gamma$ is the connection coefficients on the spheres with respect to $\gamma$.
Moreover, by Cauchy-Schwarz,
$$\int ||(\frac{\partial}{\partial\th^C})\chi'||_{L^4(S_{u,\ub'})} d\ub'\leq C\delta^{\frac 12}\mathcal O''.$$
Thus by (\ref{gammadd}), we have
\begin{equation*}
\begin{split}
&||(\frac{\partial}{\partial\th^C})\gamma_{AB}'(u,\ub)||_{L^4(S_{u,\ub})}\\
\leq&Ca+C\delta^{\frac 12}\mathcal O''+C\int_0^{\ub} ||\eta',\etab',\chih',\frac{\partial}{\partial\th^D}\gamma_{AB}',\Omega'||_{L^4(S_{u,\ub'})}d\ub'.
\end{split}
\end{equation*}
Using already established estimates, Gronwall's inequality and Cauchy-Schwarz, we thus have
$$||(\frac{\partial}{\partial\th^C})\gamma_{AB}'(u,\ub)||_{L^4(S_{u,\ub})}\leq Ca+C\delta^{\frac 12}\mathcal O''.$$
\end{proof}
A consequence of the above Proposition is the estimates on the difference of the connection coefficients:
\begin{proposition}\label{Gammap}
$$\sup_{u,\ub}||\Gamma'||_{L^4(S_{u,\ub})}\leq Ca+C\delta^{\frac 12}\mathcal O''.$$
\end{proposition}
\begin{proof}
This follows from the fact that $\Gamma'$ can be expressed as a linear combination $\gamma'$ and its first angular derivative.
\end{proof}
This implies that $\nabla^{(1)}$ and $\nabla^{(2)}$ are comparable in the following sense:
\begin{proposition}\label{connectionpL4}
Let $p\leq 4$. Suppose $\phi$ is a tensor. 
Then
$$||(\nabla^{(1)})\phi-(\nabla^{(2)})\phi||_{L^p(S_{u,\ub})} \leq C(a+\delta^{\frac 12}\mathcal O'')||\phi||_{L^\infty(S_{u,\ub})}.$$
\end{proposition}
\begin{proof}
Since we have estimates for $\gamma_{AB}$ and $(\gamma^{-1})^{AB}$. It suffice to estimate 
$$(\nabla^{(1)})_{\frac{\partial}{\partial\th^A}}\phi-(\nabla^{(2)})_{\frac{\partial}{\partial\th^A}}\phi.$$
Using Proposition \ref{Gammap}, we have
\begin{equation*}
\begin{split}
&||(\nabla^{(1)})_{\frac{\partial}{\partial\th^A}}\phi-(\nabla^{(2)})_{\frac{\partial}{\partial\th^A}}\phi||_{L^p(S_{u,\ub})}\\
\leq &||\Gamma^{(1)}\phi-\Gamma^{(2)}\phi||_{L^p(S_{u,\ub})}\\
\leq &C||\Gamma'||_{L^4(S_{u,\ub})}||\phi||_{L^\infty(S_{u,\ub})}\\
\leq &C(a+\delta^{\frac 12}\mathcal O'')||\phi||_{L^\infty(S_{u,\ub})}.
\end{split}
\end{equation*}
\end{proof}
We now show that the difference between the angular covariant derivatives of two tensors is comparable to the angular covariant derivative of the difference:
\begin{proposition}\label{angularpL4}
Let $p\leq 4$. Suppose $\phi^{(1)}$ and $\phi^{(2)}$ are tensors defined on spacetimes $(1)$ and $(2)$ respectively.
Then
$$||\nabla^{(1)}(\phi')-(\nabla\phi)'||_{L^p(S_{u,\ub})}\leq C(a+\delta^{\frac 12}\mathcal O'')||\phi||_{L^\infty(S_{u,\ub})}.$$
\end{proposition}
\begin{proof}
It is easy to see that
$$\nabla^{(1)}(\phi')-(\nabla\phi)'=-\nabla^{(1)}\phi^{(2)}+\nabla^{(2)}\phi^{(2)}.$$
The Proposition thus follows from Proposition \ref{connectionpL4}.
\end{proof}
We now estimate the coordinate angular derivative of $\Omega'$ in $L^4(S)$:
\begin{proposition}\label{dOmegap}
We have
\[
\sup_{u,\ub}||\frac{\partial}{\partial \th^A}\Omega'||_{L^4(S_{u,\ub})}\leq Ca+C\delta^{\frac 12}\mathcal O''.
\]
\end{proposition}
\begin{proof}
Recall that
\[
 \omega=-\frac{1}{2}\nabla_4\log\Omega=\frac{1}{2}\Omega\nabla_4\Omega^{-1}=\frac{1}{2}\frac{\partial}{\partial \ub}\Omega^{-1}.
\]
Hence 
\[
 \frac{1}{2}\frac{\partial}{\partial\ub}(\Omega^{-1})'=\omega'.
\]
Thus, in order to estimate $\frac{\partial}{\partial \th^A}\Omega'$, we need to estimate $\frac{\partial}{\partial \th^A}\omega'$ in $L^1_{\ub}L^4(S)$. Since by the a priori estimates in the previous section we have $||\nabla\omega||_{L^\infty(S)}\leq C$, it suffices by Proposition \ref{angularpL4} to estimate 
$$\sum_{i\leq 1}\int_0^{\ub}||\nabla^i\omega'||_{L^4(S_{u,\ub'})}d\ub'$$
By the definition of the norm and Cauchy-Schwarz
$$\sum_{i\leq 1}\int_0^{\ub}||\nabla^i\omega'||_{L^4(S_{u,\ub'})}d\ub'\leq C\delta^{\frac 12}\mathcal O''.$$
The conclusion thus follows.
\end{proof}
The above Proposition allows us to estimate the second coordinate angular derivatives of $\gamma'$ in $L^2(S)$.
\begin{proposition}\label{ddgammap}
$$\sup_{u,\ub}||(\frac{\partial^2}{\partial\th^C \partial\th^D})\gamma_{AB}'||_{L^2(S_{u,\ub})}\leq Ca+C\delta^{\frac 12}(\mathcal O''+\tilde{\mathcal O}'),$$
where $L^2(S)$ is understood as the $L^2$ norms for scalar functions.
\end{proposition}
\begin{proof}
By the equation 
$$\frac{\partial}{\partial \ub}\gamma_{AB}=2\Omega\chi_{AB},$$
and the a priori estimates in the previous section, we have
\begin{equation}\label{gamma2dd}
\begin{split}
&||(\frac{\partial^2}{\partial\th^C\partial\th^D})\gamma_{AB}'(u,\ub)||_{L^2(S_{u,\ub})}\\
\leq& Ca+C\int_0^{\ub}\left(\|(\frac{\partial^2}{\partial\th^C\partial\th^D})\Omega'||_{L^2(S_{u,\ub'})}+\|(\frac{\partial^2}{\partial\th^C\partial\th^D})\chi'||_{L^2(S_{u,\ub'})}\right) d\ub'\\
&+ C\int_0^{\ub}\left(\|(\frac{\partial}{\partial\th^C})\Omega'||_{L^2(S_{u,\ub'})}+\|(\frac{\partial}{\partial\th^C})\chi'||_{L^2(S_{u,\ub'})}\right) d\ub'.
\end{split}
\end{equation}
The last term has already been estimated in the proof of Proposition \ref{dgammap}:
\begin{equation*}
\begin{split}
&\int_0^{\ub}\left(\|(\frac{\partial}{\partial\th^C})\Omega'||_{L^2(S_{u,\ub'})}+\|(\frac{\partial}{\partial\th^C})\chi'||_{L^2(S_{u,\ub'})}\right) d\ub'\\
\leq& Ca+C\delta^{\frac 12}\mathcal O''.
\end{split}
\end{equation*}
Since
$$\nabla(\log \Omega)=\frac 12 (\eta+\etab),$$
we have
\begin{equation*}
\begin{split}
&||(\frac{\partial^2}{\partial\th^C\partial\th^D})\Omega'||_{L^2(S_{u,\ub'})}\\
\leq& C||\frac{\partial}{\partial\th}(\eta,\etab)'||_{L^2(S_{u,\ub'})}+C||\frac{\partial}{\partial\th}\Gamma'||_{L^2(S_{u,\ub'})}||\Omega||_{L^\infty(S_{u,\ub'})}+C||\frac{\partial}{\partial\th}\Gamma||_{L^2(S_{u,\ub'})}||\Omega'|||_{L^\infty(S_{u,\ub'})}\\
& +C||\Gamma'||_{L^4(S_{u,\ub'})}||\frac{\partial}{\partial\th}\Omega||_{L^4(S_{u,\ub'})}+C||\Gamma||_{L^\infty(S_{u,\ub'})}||\frac{\partial}{\partial\th}\Omega'||_{L^2(S_{u,\ub'})}\\
\leq &C\sum_{i\leq 1}||\nabla^i(\eta,\etab)'||_{L^2(S_{u,\ub'})}+C\sum_{i\leq 1}||(\frac{\partial}{\partial\th})^i\Gamma'||_{L^2(S_{u,\ub'})}+C\sum_{i\leq 1}||(\frac{\partial}{\partial\th})^i\Omega'||_{L^2(S_{u,\ub'})}.
\end{split}
\end{equation*}
Using Propositions \ref{dOmegap} and the definition of the norms $\mathcal O''$ and $\tilde{\mathcal O}'$, we have
\begin{equation*}
\begin{split}
&\sum_{i\leq 1}\int_0^{\ub}(||\nabla^i(\eta,\etab)'||_{L^2(S_{u,\ub'})}+||(\frac{\partial}{\partial\th})^i\Gamma'||_{L^2(S_{u,\ub'})}+||(\frac{\partial}{\partial\th})^i\Omega'||_{L^2(S_{u,\ub'})})d\ub'\\
\leq &C\delta^{\frac 12}(\mathcal O''+\tilde{\mathcal O}'_u)+C\int_0^{\ub}||(\frac{\partial^2}{\partial\th^E\partial\th^F})\gamma_{AB}'(u,\ub)||_{L^2(S_{u,\ub'})}d\ub').
\end{split}
\end{equation*}
Moreover,
$$\int ||(\frac{\partial^2}{\partial\th^C\partial\th^D})\chi'||_{L^2(S_{u,\ub'})} d\ub'\leq C(\mathcal O''+\tilde{\mathcal O}').$$
Thus by (\ref{gamma2dd}), we have
\begin{equation*}
\begin{split}
&||(\frac{\partial^2}{\partial\th^C\partial\th^D})\gamma_{AB}'(u,\ub)||_{L^2(S_{u,\ub})}\\
\leq&Ca+C\delta^{\frac 12}(\mathcal O''+\tilde{\mathcal O}')+C\int_0^{\ub}||(\frac{\partial^2}{\partial\th^E\partial\th^F})\gamma_{AB}'(u,\ub)||_{L^2(S_{u,\ub'})}d\ub').
\end{split}
\end{equation*}
By Gronwall's inequality, we thus have
$$||(\frac{\partial^2}{\partial\th^C \partial\th^D})\gamma_{AB}'||_{L^2(S_{u,\ub})}\leq Ca+C\delta^{\frac 12}(\mathcal O''+\tilde{\mathcal O}').$$
\end{proof}
The proof of the previous Proposition also shows that
\begin{proposition}\label{ddOmegap}
$$\sup_{u,\ub}||(\frac{\partial^2}{\partial\th^C \partial\th^D})\Omega'||_{L^2(S_{u,\ub})}\leq Ca+C\delta^{\frac 12}(\mathcal O''+\tilde{\mathcal O}'),$$
\end{proposition}
Proposition \ref{ddgammap} implies the following estimates of $K'$, the difference of the Gauss curvatures:
\begin{proposition}\label{Kp}
$K'$ satisfies the following bounds:
$$\sup_{u,\ub}||K'||_{L^2(S_{u,\ub})}\leq Ca+C\delta^{\frac 12}(\mathcal O''+\tilde{\mathcal O}').$$
\end{proposition}
\begin{proof}
This follows from the estimates of the derivatives of $\gamma'$ in Propositions \ref{dgammap} and \ref{ddgammap}.
\end{proof}
Moreover, Proposition \ref{ddgammap} allows us to conclude that the covariant second derivative with respect to connection $(1)$ and connection $(2)$ are comparable:
\begin{proposition}\label{connectionpL2}
Let $p\leq 2$. Suppose $\phi$ is a tensor. 
Then
$$||(\nabla^{(1)})^2\phi-(\nabla^{(2)})^2\phi||_{L^p(S_{u,\ub})} \leq C(a+\delta^{\frac 12}\mathcal O'')\sum_{i\leq 1}||\nabla^i\phi||_{L^4(S_{u,\ub})}.$$
\end{proposition}
\begin{proof}
By Propositions \ref{Gammap}, \ref{ddgammap} and \ref{Kp}.
\end{proof}
This implies that taking the covariant second derivatives of the difference is comparable to taking the difference of the covariant second derivatives:
\begin{proposition}\label{angularpL2}
Let $p\leq 2$. Suppose $\phi^{(1)}$ and $\phi^{(2)}$ are tensors defined on spacetimes $(1)$ and $(2)$ respectively.
Then
$$||(\nabla^{(1)})^2(\phi')-(\nabla^2\phi)'||_{L^p(S_{u,\ub})}\leq C(a+\delta^{\frac 12}\mathcal O'')\sum_{i\leq 1}||\nabla^i\phi||_{L^4(S_{u,\ub})}.$$
\end{proposition}
\begin{proof}
The proof follows as that of Proposition \ref{angularpL4}, using Proposition \ref{connectionpL2} instead of Proposition \ref{connectionpL4}.
\end{proof}
We now prove the estimates for $b'$:
\begin{proposition}\label{bp}
\[
\sup_{u,\ub}||(b^A)'||_{L^\infty(S_{u,\ub})}, \sup_{u,\ub}||\frac{\partial}{\partial\th^B}(b^A)'||_{L^4(S_{u,\ub})},\sup_{u,\ub}||\frac{\partial^2}{\partial \th^C\partial\th^B}(b^A)'||_{L^2(S_{u,\ub})}\leq C(\mathcal O''+\tilde{\mathcal O}').
\]
\end{proposition}
\begin{proof}
Recall that $b$ satisfies the transport equation 
$$\frac{\partial b^A}{\partial \ub}=-4\Omega^2\zeta^A.$$
The Proposition thus follows from differentiating the equation by $\frac{\partial}{\partial\th}$, integrating in $\ub$ and using the the bounds in Propositions \ref{Omegap}, \ref{dOmegap} and \ref{ddOmegap} and the definitions of the norms.
\end{proof}

\subsection{Transport Equations for the Difference Quantities}\label{convsec3}
In this Subsection, we show how to derive estimates for the difference $\phi'$ of the quantities $\phi^{(1)}$ and $\phi^{(2)}$ satisfying transport equations in the $\nab_4^{(1)}$, $\nab_4^{(2)}$ or $\nab_3^{(1)}$, $\nab_3^{(2)}$ directions. This should be compared with the previous Subsection where we considered difference quantities satisfying transport equations, but unlike in this Subsection, the corresponding transport equations held with respect to the $\frac{\partial}{\partial \ub}$, $\frac{\partial}{\partial u}$ derivatives rather than $\nab_4$, $\nab_3$, which themselves depend on the spacetimes under consideration.

In order to estimate the quantities $\phi'$ associated to the difference between that of different spacetimes, we need to derive a transport equation for $\phi'$ from the transport equation for $\phi$. 
\begin{proposition}\label{peqn}
Consider spacetimes that satisfy the hypotheses of Theorem \ref{timeofexistence} and \ref{convergencethm}. Let $\phi$ be a $(0,r)$ S-tensor. Suppose $\nabla_4\phi=F$. Then
\begin{equation*}
\begin{split}
\nabla_4^{(1)}\phi'\sim & F'+\frac{(\Omega^{-1})'}{\Omega^{-1}}\nabla_4\phi+\frac{(\Omega^{-1})'}{\Omega^{-1}}\gamma^{-1}\chi\phi+(\gamma^{-1}\chi)'\phi.
\end{split}
\end{equation*}
Similarly, suppose $\nabla_3\phi=G$. Then
\begin{equation*}
\begin{split}
\nabla_3^{(1)}\phi'\sim & G'+\frac{(\Omega^{-1})'}{\Omega^{-1}}\nabla_3\phi+\Omega^{-1}(b^A)'\nabla_{\frac{\partial}{\partial \th^A}}\phi+\frac{(\Omega^{-1})'}{\Omega^{-1}}\gamma^{-1}\chib\phi+(\gamma^{-1}\chib)'\phi+\Omega^{-1}\frac{\partial b'}{\partial\th}\phi.
\end{split}
\end{equation*}
\end{proposition}
\begin{proof}
We write in coordinates
$$\nabla_4\phi_{A_1...A_r}=\Omega^{-1}\frac{\partial}{\partial \ub}\phi_{A_1...A_r}-(\gamma^{-1})^{CD}\chi_{A_iD}\phi_{A_1...\hat{A_i}C...A_r},$$
where $\hat{A_i}$ denotes that the original $A_i$ in the $i$-th slot of the tensor is removed. This equation holds in both spacetimes $(1)$ and $(2)$.
Then, we derive the equations for $\phi'$. We will write schematically without the exact constant depending only on $r$:
\begin{equation*}
\begin{split}
&(\nabla_4)^{(1)}\phi' \\
=&\left((\Omega^{-1})^{(1)}\frac{\partial}{\partial \ub}-(\gamma^{-1})^{(1)}\chi^{(1)}\right)\left(\phi^{(1)}-\phi^{(2)}\right) \\
=&\left((\Omega^{-1})^{(1)}\frac{\partial}{\partial \ub}-(\gamma^{-1})^{(1)}\chi^{(1)}\right)\phi^{(1)}-\left((\Omega^{-1})^{(2)}\frac{\partial}{\partial \ub}-(\gamma^{-1})^{(2)}\chi^{(2)}\right)\phi^{(2)}\\
& -(\Omega^{-1})'\frac{\partial}{\partial \ub}\phi^{(2)}+(\gamma^{-1}\chi)'\phi^{(2)}\\
=&F'-(\Omega^{-1})'\frac{\partial}{\partial \ub}\phi^{(2)}+(\gamma^{-1}\chi)'\phi^{(2)}\\
=&F'-\frac{(\Omega^{-1})'}{(\Omega^{-1})^{(1)}}(\nabla_4)^{(1)}\phi^{(2)}+\frac{(\Omega^{-1})'}{(\Omega^{-1})^{(1)}}(\gamma^{-1}\chi)^{(1)}\phi^{(2)}+(\gamma^{-1}\chi)'\phi^{(2)}.
\end{split}
\end{equation*}
For the $\nabla_3$ equations, we write in coordinates
\begin{equation*}
\begin{split}
\nabla_3\phi_{A_1...A_r}=&\Omega^{-1}\frac{\partial}{\partial u}\phi_{A_1...A_r}+\Omega^{-1}b^B\frac{\partial}{\partial \th^B}\phi_{A_1...A_r} \\
&-(\gamma^{-1})^{CD}\chib_{A_i D}\phi_{A_1...\hat{A_i}C...A_r}+\Omega^{-1}\frac{\partial b^B}{\partial\th^{A_i}}\phi_{A_1...\hat{A_i}C...A_r}.
\end{split}
\end{equation*}
We derive the difference equations as before. However, $e_3$ in the respective spacetimes are not parallel to each other and we therefore have extra terms.
\begin{equation*}
\begin{split}
&(\nabla_3)^{(1)}\phi' \\
=&\left((\Omega^{-1})^{(1)}\frac{\partial}{\partial u}+(\Omega^{-1})^{(1)}b^{(1)}\frac{\partial}{\partial \th} -(\gamma^{-1})^{(1)}\chib^{(1)}+(\Omega^{-1})^{(1)}\frac{\partial b^{(1)}}{\partial \th}\right)\left(\phi^{(1)}-\phi^{(2)}\right) \\
=&\left((\Omega^{-1})^{(1)}\frac{\partial}{\partial u}+(\Omega^{-1})^{(1)}b^{(1)}\frac{\partial}{\partial \th} -(\gamma^{-1})^{(1)}\chib^{(1)}+(\Omega^{-1})^{(1)}\frac{\partial b^{(1)}}{\partial \th}\right)\phi^{(1)} \\
&-\left((\Omega^{-1})^{(2)}\frac{\partial}{\partial u}+(\Omega^{-1})^{(2)}b^{(2)}\frac{\partial}{\partial \th} -(\gamma^{-1})^{(2)}\chib^{(2)}+(\Omega^{-1})^{(2)}\frac{\partial b^{(2)}}{\partial \th}\right)\phi^{(2)}\\
& -(\Omega^{-1})'\frac{\partial}{\partial u}\phi^{(2)}+(\gamma^{-1}\chib)'\phi^{(2)}-(\Omega^{-1}b)'\frac{\partial}{\partial \th}\phi^{(2)}-(\Omega^{-1}\frac{\partial b}{\partial \th})'\phi^{(2)}\\
=&G'-\frac{(\Omega^{-1})'}{(\Omega^{-1})^{(1)}}\nabla_3^{(1)}\phi^{(1)}-(\Omega^{-1})^{(2)}b'\frac{\partial}{\partial\th}\phi^{(2)}-(\Omega^{-1})^{(2)}\frac{\partial b'}{\partial \th}\phi^{(2)}\\
&+\frac{(\Omega^{-1})'}{(\Omega^{-1})^{(1)}}(\gamma^{-1}\chib)^{(1)}\phi^{(1)}+(\gamma^{-1}\chib)'\phi^{(2)}.
\end{split}
\end{equation*}
\end{proof}
This, together with the commutation estimates in Proposition \ref{commuteeqn}, gives the following estimates for $(\nabla^i\phi)'$.
\begin{proposition}\label{pcommuteeqn}
Suppose $\nabla_4\phi=F$. Then
\begin{equation*}
\begin{split}
\nabla_4(\nabla^i\phi)'\sim &\sum_{i_1+i_2+i_3=i}(\nabla^{i_1}(\eta,\etab)^{i_2}\nabla^{i_3} F)'+\sum_{i_1+i_2+i_3+i_4=i}(\nabla^{i_1}(\eta,\etab)^{i_2}\nabla^{i_3}\chi\nabla^{i_4} \phi)' \\
&+(\Omega^{-1})'\left(\sum_{i_1+i_2+i_3=i}\nabla^{i_1}(\eta,\etab)^{i_2}\nabla^{i_3} F+\sum_{i_1+i_2+i_3+i_4=i}\nabla^{i_1}(\eta,\etab)^{i_2}\nabla^{i_3}\chi\nabla^{i_4} \phi\right)\\
&+(\gamma^{-1}\chi)'\nabla^i\phi.
\end{split}
\end{equation*}
Similarly, suppose $\nabla_3\phi=G$. Then
\begin{equation*}
\begin{split}
\nabla_3(\nabla^i\phi)'\sim &\sum_{i_1+i_2+i_3=i}(\nabla^{i_1}(\eta,\etab
)^{i_2}\nabla^{i_3} G)'+\sum_{i_1+i_2+i_3+i_4=i}(\nabla^{i_1}(\eta,\etab)^{i_2}\nabla^{i_3}\chib\nabla^{i_4} \phi)' \\
&+(\Omega^{-1})'\left(\sum_{i_1+i_2+i_3=i}\nabla^{i_1}(\eta,\etab)^{i_2}\nabla^{i_3} G+\sum_{i_1+i_2+i_3+i_4=i}\nabla^{i_1}(\eta,\etab)^{i_2}\nabla^{i_3}\chib\nabla^{i_4} \phi\right)\\
&+b'\nabla^{i+1}\phi+((\gamma^{-1}\chi)'+\frac{\partial b'}{\partial\th})\nabla^i\phi.
\end{split}
\end{equation*}
\end{proposition}
\begin{proof}
According to the already established estimates, we have $\frac{1}{2}\leq \Omega\leq 2$. The result now follows directly from Propositions \ref{commuteeqn} and \ref{peqn}.
\end{proof}
We now use Propositions \ref{transport} and \ref{pcommuteeqn} to obtain estimates from a transport equation.
\begin{proposition}\label{transportp}
Suppose $\nabla_4\phi=F$. Then
\begin{equation*}
\begin{split}
&\sum_{i\leq 1}||\nabla^{i}\phi'||_{L^2(S_{u,\ub})} \\
\leq &\sum_{i\leq 1}||\nabla^{i}\phi'||_{L^2(S_{u,0})}+C\int_0^{\ub} \sum_{i\leq 1}||(\nabla^iF)'||_{L^2(S_{u,\ub'})} d\ub'\\
&+C\delta^{\frac 12}(\mathcal O''+\tilde{\mathcal O}'+\mathcal O' )(\sum_{i\leq 1}\sup_{\ub'\leq \ub}(||\nabla^iF||_{L^2(S_{u,\ub'})}+||\nabla^i\phi||_{L^2(S_{u,\ub'})})).
\end{split}
\end{equation*}

Similarly, suppose $\nabla_3\phi=G$. Then
\begin{equation*}
\begin{split}
&\sum_{i\leq 1}||\nabla^{i}\phi'||_{L^2(S_{u,\ub})} \\
\leq &\sum_{i\leq 1}||\nabla^{i}\phi'||_{L^2(S_{0,\ub})}+C\int_0^{\ub} \sum_{i\leq 1}||(\nabla^iG)'||_{L^2(S_{u,\ub'})} d\ub'\\
&+C\delta^{\frac 12}(\mathcal O''+\tilde{\mathcal O}'+\mathcal O' )(\sum_{i\leq 1}\sup_{\ub'\leq \ub}(||\nabla^iG||_{L^2(S_{u,\ub'})}+||\nabla^i\phi||_{L^2(S_{u,\ub'})})).
\end{split}
\end{equation*}
\end{proposition}
\begin{proof}
We use the pointwise estimates for $\gamma'$, $\Omega'$ and $b'$ proved in the last section.

By Propositions \ref{angularpL4}, it suffices to estimate 
$$\sum_{i\leq 1}||(\nabla^i\phi)'||_{L^2(S_{u,\ub})}$$
since the difference can be estimated by
$$\sum_{i\leq 1}||(\nabla^i\phi)'-\nabla(\phi')||_{L^2(S_{u,\ub})}\leq C(\mathcal O''+\tilde{\mathcal O}')(||\nabla\phi||_{L^4(S_{u,\ub})}+||\phi||_{L^\infty(S_{u,\ub})}).$$
We can estimate $||(\nabla^i\phi)'||_{L^2(S_{u,\ub})}$ using Proposition \ref{transport} and the equation in Proposition \ref{pcommuteeqn}. Recall the formula in Proposition \ref{pcommuteeqn}
\begin{equation*}
\begin{split}
\nabla_4(\nabla^i\phi)'\sim &\sum_{i_1+i_2+i_3=i}(\nabla^{i_1}(\eta,\etab)^{i_2}\nabla^{i_3} F)'+\sum_{i_1+i_2+i_3+i_4=i}(\nabla^{i_1}(\eta,\etab)^{i_2}\nabla^{i_3}\chi\nabla^{i_4} \phi)' \\
&+(\Omega^{-1})'\left(\sum_{i_1+i_2+i_3=i}\nabla^{i_1}(\eta,\etab)^{i_2}\nabla^{i_3} F+\sum_{i_1+i_2+i_3+i_4=i}\nabla^{i_1}(\eta,\etab)^{i_2}\nabla^{i_3}\chi\nabla^{i_4} \phi\right)\\
&+(\gamma^{-1}\chi)'\nabla^i\phi.
\end{split}
\end{equation*}
We estimate term by term:
\begin{equation*}
\begin{split}
&\int_0^{\ub}\sum_{i_1+i_2+i_3\leq 1}||(\nabla^{i_1}(\eta,\etab)^{i_2}\nabla^{i_3} F)'||_{L^2(S_{u,\ub'})}d\ub'\\
\leq &C\int_0^{\ub}(\sum_{i_1\leq 1}||(\eta,\etab)||^{i_1}_{L^\infty(S_{u,\ub'})})\sum_{i_2\leq 1}||(\nabla^{i_2} F)'||_{L^2(S_{u,\ub'})}d\ub'\\
\leq &C\sum_{i\leq 1}||(\nabla^{i} F)'||_{L^1_{\ub}L^2(S_{u,\ub})}.
\end{split}
\end{equation*}
The second term can be estimated by
\begin{equation*}
\begin{split}
&\int_0^{\ub}\sum_{i_1+i_2+i_3\leq 1}||(\nabla^{i_1}(\eta,\etab)^{i_2}\nabla^{i_3}\chi\nabla^{i_4} \phi)'||_{L^2(S_{u,\ub'})}d\ub'\\
\leq &C\int_0^{\ub}(\sum_{i_1\leq 1}\sum_{i_2\leq 2}||\nabla^{i_1}(\eta,\etab,\trch,\chih)||^{i_2}_{L^\infty(S_{u,\ub'})})\sum_{i_3\leq 1}||(\nabla^{i_3} \phi)'||_{L^2(S_{u,\ub'})}d\ub'\\
&+C\int_0^{\ub}(\sum_{i_1\leq 1}||\nab^{i_1}(\eta',\etab',\trch',\chih')||_{L^2(S_{u,\ub'})})\\
&\quad\quad\quad\quad\times(\sum_{i_2\leq 1}\sum_{i_3\leq 1}||\nabla^{i_1}(\eta,\etab,\trch,\chih)||^{i_3}_{L^\infty(S_{u,\ub'})})(\sum_{i_4\leq 1}||\nabla^{i_4} \phi||_{L^2(S_{u,\ub'})})d\ub'\\
\leq &C\sum_{i\leq 1}||(\nabla^{i} \phi)'||_{L^1_{\ub}L^2(S_{u,\ub})}+C\delta^{\frac 12}(\mathcal O''+\tilde{\mathcal O}'+\mathcal O')(||\nabla\phi||_{L^2(S_{u,\ub})}+||\phi||_{L^\infty(S_{u,\ub})}).\\
\end{split}
\end{equation*}
The third term is easier to estimate since by Proposition \ref{Omegap}, $(\Omega^{-1})'$ can be estimated in $L^\infty$:
\begin{equation*}
\begin{split}
&\int_0^{\ub}\sum_{i_1+i_2+i_3\leq 1}||(\Omega^{-1})'\left(\nabla^{i_1}(\eta,\etab)^{i_2}\nabla^{i_3} F+\nabla^{i_1}(\eta,\etab)^{i_2}\nabla^{i_3}\chi\nabla^{i_4} \phi\right)||_{L^2(S_{u,\ub'})}d\ub'\\
\leq &C\delta^{\frac 12}(\mathcal O''+\tilde{\mathcal O}')(\sum_{i\leq 1}\sup_{\ub'\leq \ub}||\nabla^iF||_{L^2(S_{u,\ub'})}+\sum_{i\leq 1}\sup_{\ub'\leq \ub}||\nabla^i\phi||_{L^2(S_{u,\ub'})}).
\end{split}
\end{equation*}
The fourth term can be controlled in the same way as the third term, since $(\gamma^{-1}\chi)'$ can be estimated in $L^2_{\ub}L^\infty$:
\begin{equation*}
\begin{split}
&\int_0^{\ub}\sum_{i\leq 1}||(\gamma^{-1}\chi)'\nabla^i\phi||_{L^2(S_{u,\ub'})}d\ub'\\
\leq &C\delta^{\frac 12}(\mathcal O''+\tilde{\mathcal O}')(\sum_{i\leq 1}\sup_{\ub'\leq \ub}||\nabla^i\phi||_{L^2(S_{u,\ub'})}).
\end{split}
\end{equation*}
Putting all these estimates together using Proposition \ref{transport}, we have
\begin{equation*}
\begin{split}
&\sum_{i\leq 1}||\nabla^{i}\phi'||_{L^2(S_{u,\ub})} \\
\leq &\sum_{i\leq 1}||\nabla^{i}\phi'||_{L^2(S_{0,\ub})}+C\sum_{i\leq 1}||(\nabla^iF)'||_{L^1_{\ub}L^2(S_{u,\ub})}+C\delta||\nabla^{i}\phi'||_{L^2_{\ub}L^2(S_{u,\ub})}\\
&+C\delta^{\frac 12}(\mathcal O''+\tilde{\mathcal O}'+\mathcal O' )(\sum_{i\leq 2}\sup_{\ub'\leq \ub}(||\nabla^iF||_{L^2(S_{u,\ub'})}+||\nabla^i\phi||_{L^2(S_{u,\ub'})})).
\end{split}
\end{equation*}
Take $\delta$ to be sufficiently small depending only on $C$ but not $a$, we can absorb the third term to the left hand side to get
\begin{equation*}
\begin{split}
&\sum_{i\leq 1}||\nabla^{i}\phi'||_{L^2(S_{u,\ub})} \\
\leq &\sum_{i\leq 1}||\nabla^{i}\phi'||_{L^2(S_{0,\ub})}+C\sum_{i\leq 1}||(\nabla^iF)'||_{L^1_{\ub}L^2(S_{u,\ub})}\\
&+C\delta^{\frac 12}(\mathcal O''+\tilde{\mathcal O}'+\mathcal O' )(\sum_{i\leq 2}\sup_{\ub'\leq \ub}(||\nabla^iF||_{L^2(S_{u,\ub'})}+||\nabla^i\phi||_{L^2(S_{u,\ub'})})).
\end{split}
\end{equation*}
The proof for the $\nabla_3$ equation is analogous.
\end{proof}
\subsection{Estimates for the Ricci Coefficient Difference}\label{convsec4}
The results of the previous Subsection can now be used to estimate the norm ${\mathcal O}'$ in terms of $\mathcal R'$. 

In this Subsection and the subsequent Section \ref{convsec5}, in order to simplify notation, we will omit the superindices $(1)$ and $(2)$ in all the background quantities, for example, $\nab_3, \nab, \Omega, \beta...$

\begin{proposition}\label{Riccip}
There exists $\delta$ and $C$ depending only on the a priori estimates in Theorem \ref{timeofexistence} (but independent of $a$) such that
$$\mathcal O'\leq Ca+C\delta^{\frac 12}\mathcal R'+C\delta^{\frac 12}\tilde{\mathcal O}'.$$
\end{proposition}
\begin{proof}
We control $\mathcal O'$ using the null structure equations.
First, we consider the Ricci coefficients $\trchb, \chibh, \etab$. They all satisfy $\nabla_3$ equations of the form
$$\nabla_3\psi=\psi\psi+\Psi$$
such that $\chih, \omega$ do not appear in the $\psi$ terms on the right hand side; and that $\beta$ does not appear as $\Psi$. This is important because $\chih'$ and $\omega'$ cannot be controlled by $C\mathcal O'$ in the $L^2(S)$ norms and $\beta'$ cannot be controlled by $C\mathcal R'$ in $L^2(\Hb)$. In order to estimate $\psi$ using this equation, we need to estimate the curvature term:
\begin{equation*}
\begin{split}
\sum_{i\leq 1}||(\nabla^i\Psi)'||_{L^1_uL^2(S)}\leq C\delta^{\frac 12}(\mathcal O''+\tilde{\mathcal O}')+C\delta^{\frac 12}\mathcal R'.
\end{split}
\end{equation*}
and the nonlinear terms:
\begin{equation*}
\begin{split}
&\sum_{i\leq 1}||(\nabla^i(\psi\psi))'||_{L^1_u L^2(S_{u,\ub})}\\
\leq &C\delta^{\frac 12}(\mathcal O''+\tilde{\mathcal O}')+C \delta^{\frac 12}\sum_{i_1+i_2\leq 1}\sup_{u'\leq u}||\gamma'\nabla^{i_1}\psi\nabla^{i_2}\psi||_{ L^2(S_{u',\ub})}\\
&+C\delta^{\frac 12}\sum_{i_1+i_2\leq 1}\sup_{u'\leq u}||\nabla^{i_1}\psi\nabla^{i_2}\psi'||_{L^2(S_{u',\ub})}\\
\leq &C\delta^{\frac 12}(\mathcal O''+\tilde{\mathcal O}'+\mathcal O').
\end{split}
\end{equation*}
By Proposition \ref{transportp}, we thus have
\begin{equation}\label{Ricci31}
\begin{split}
&\sum_{i\leq 1}\sup_{u}||\nabla^i(\trchb',\chibh',\etab')||_{L^2(S_{u,\ub})}\\
\leq &\sum_{i\leq 1}||\nabla^i(\trchb',\chibh',\etab')||_{L^2(S_{0,\ub})}+ C\delta^{\frac 12}\sup_{u}(\mathcal O''_{u}+\tilde{\mathcal O}'_{u}+\mathcal O')+C\delta^{\frac 12}\mathcal R'\\
\leq & a+ C\delta^{\frac 12}(\mathcal O''+\tilde{\mathcal O}'+\mathcal O')+C\delta^{\frac 12}\mathcal R'.
\end{split}
\end{equation}
We now consider the equation for $\omega$, which schematically looks like
$$\nabla_3\omega=\omega\psi+\chih\psi+\psi\psi+\rhoc,$$
where $\psi\neq \chih,\omega$ are the good components as above. The term $\psi'\omega$ or $\psi'\chih$ can be estimated since $\psi'$ can be put in an appropriate $L^p(S)$ norm and be controlled by $C\mathcal O'$. However, the most difficult terms are $\omega'\psi$ and $\chih'\psi$ since $\omega'$ and $\chih'$ cannot be bounded in $L^1_uL^p(S)$ by $C\mathcal O'$. Thus, by Proposition \ref{transportp}, we have
\begin{equation*}
\begin{split}
&\sum_{i\leq 1}||\nabla^i\omega'||_{L^2(S_{u,\ub})}\\
\leq &C\sum_{i\leq 1}||\nabla^i\omega'||_{L^2(S_{0,\ub})}+ C\sum_{i_1+i_2\leq 1}||\nabla^{i_1}\psi\nabla^{i_2}(\chih',\omega')||_{L^1_u L^p(S_{u,\ub})}+C\delta^{\frac 12}\mathcal R'+C\delta^{\frac 12}(\mathcal O''+\tilde{\mathcal O}'+\mathcal O').
\end{split}
\end{equation*}
Since this holds for every $\ub$, we can integrate in $L^2$ in $\ub$ to get
\begin{equation}\label{Ricci32}
\begin{split}
&\sum_{i\leq 1}||\nabla^i\omega'||_{L^2_{\ub}L^2(S_{u,\ub})}\\
\leq &C\sum_{i\leq 1}||\nabla^i\omega'||_{L^2_{\ub}L^2(S_{0,\ub})}+  \int_0^u||\nabla^{i_1}\psi\nabla^{i_2}(\chih',\omega')||_{L^2_{\ub} L^2(S_{u',\ub})} du'+C\mathcal R'+C\delta^{\frac 12}(\mathcal O''+\tilde{\mathcal O}'+\mathcal O')\\
\leq & a+C\mathcal R'+C\delta^{\frac 12}(\mathcal O''+\tilde{\mathcal O}'+\mathcal O'),
\end{split}
\end{equation}
since we can control $\nabla^i\omega'$ by $C\mathcal O'$ after integrating along the $\ub$ direction.
We can estimate $\chih$ in a similar manner as $\omega$. $\chih$ satisfies
$$\nabla_3\chih=\psi\psi+\psi\chih+\nabla\eta,$$
where $\psi\neq \chih,\omega$. Putting $\nabla\eta$ in the $\tilde{\mathcal O}'$ norm and using already obtained estimates, we have
\begin{equation*}
\begin{split}
&\sum_{i\leq 1}||\nabla^i\chih'||_{L^2(S_{u,\ub})}\\
\leq &C\sum_{i\leq 1}||\nabla^i\chih'||_{L^2(S_{0,\ub})}+ C\sum_{i_1+i_2\leq 1}||\nabla^{i_1}\psi \nabla^{i_2}\chih'||_{L^1_u L^2(S_{u,\ub})}+C\delta^{\frac 12}(\mathcal O''+\tilde{\mathcal O}'+\mathcal O').
\end{split}
\end{equation*}
As for $\omega'$, we now integrate in $L^2$ in $\ub$ to get
\begin{equation}\label{Ricci33}
\begin{split}
&\sum_{i\leq 1}||\nabla^i\chih'||_{L^2_{\ub}L^2(S_{u,\ub})}\\
\leq &\sum_{i\leq 1}||\nabla^i\chih'||_{L^2_{\ub}L^2(S_{0,\ub})}+ C \int_0^u||\nabla^{i_1}\psi \nabla^{i_2}\chih'||_{L^2_{\ub} L^2(S_{u',\ub})}du'+C\delta^{\frac 12}(\mathcal O''+\tilde{\mathcal O}'+\mathcal O')\\
\leq &a+C\delta^{\frac 12}(\mathcal O''+\tilde{\mathcal O}'+\mathcal O').
\end{split}
\end{equation}
It remains to consider $\trch,\eta,\omegab$. They satisfy $\nabla_4$ equations of the form:
$$\nabla_4\psi=\psi\psi+\Psi,$$
where $\psi$ can be any Ricci coefficients and $\Psi$ can be any curvature components $\neq\alpha,\alphab$, i.e., all the $\Psi$ terms can be controlled in $L^2(H)$ by $\mathcal R'$. Moreover, since we are now integrating in $\ub$, all terms $\psi'$ can be estimated by $C\mathcal O'$. By Proposition \ref{transportp}, we have
$$\sum_{i\leq 1}||\nabla^i(\trch',\eta',\omegab')||_{L^2(S_{u,\ub})}\leq \sum_{i\leq 1}||\nabla^i(\trch',\eta',\omegab')||_{L^2(S_{0,\ub})}+ C\mathcal R'+C\int_0^{\ub}\mathcal O' d\ub'.$$
Therefore,
$$\sum_{i\leq 1}||\nabla^i(\trch',\eta',\omegab')||_{L^2(S_{u,\ub})}\leq a+C\delta^{\frac 12}\mathcal R'+C\delta^{\frac 12}(\sup_u\mathcal O''+\tilde{\mathcal O}'+\mathcal O').$$
Using $\mathcal O''\leq C(\tilde{\mathcal O}'+\mathcal O')$ and all the above estimates, we have
$$\mathcal O'\leq Ca+C\delta^{\frac 12}\mathcal R'+C\delta^{\frac 12}(\tilde{\mathcal O}'+\mathcal O').$$
By choosing $\delta$ sufficiently small so that $C\delta^{\frac 12}\leq \frac 12$, we have
$$\mathcal O'\leq Ca+C\delta^{\frac 12}\mathcal R'+C\delta^{\frac 12}\tilde{\mathcal O}'.$$
Notice that in particular $\delta$ and $C$ depend only on the a priori estimates from Theorem \ref{timeofexistence} and are independent of $a$.
\end{proof}
We now move on to estimate $\tilde{\mathcal O}'$ by $\mathcal R'$. Recall from Propositions \ref{Theta} and \ref{ellipticTheta} that $\tilde{\mathcal O}$ was controlled using a combination of transport equations for $\Theta$ and Hodge systems. The norm $\tilde{\mathcal O}'$ can be dealt with in a similar fashion. In order to perform this scheme, we need show that as was with the case of the norm $\tilde{\mathcal O}$ where we controlled $\nab\psi$ from $\Theta$, the difference $\psi'$ satisfies elliptic equations with $\Theta'$ as a source. This is given by the following Proposition:
\begin{proposition}\label{commuteelliptic}
Let $\phi$ be a $(0,r)$-tensorfield. Suppose 
$$\div\phi=f,\quad\curl\phi=g,\quad\tr\phi=h$$
on each of the spacetimes. Then
$$\div\phi'\sim f'+ (\gamma^{-1})'\nabla\phi+\Gamma'\phi,$$
$$\curl\phi'\sim g'+ (\eps)'\nabla\phi,$$
$$\tr\phi' \sim h' + (\gamma^{-1})'\phi.$$
\end{proposition}
\begin{proof}
This is a straightforward coordinate computation.
\end{proof}
Based on the estimates we have on $\gamma'$, $(\gamma^{-1})'$ and $\Gamma'$, and using the estimates in Proposition \ref{ellipticthm}, we have
\begin{proposition}\label{ellipticthmp}
Let $\phi^{(1)}$ and $\phi^{(2)}$ be totally symmetric $r+1$ covariant tensorfields on 2-spheres $(\mathbb S^2,\gamma^{(1)})$, $(\mathbb S^2,\gamma^{(2)})$ respectively satisfying
$$\div\phi^{(i)}=f^{(i)},\quad \curl\phi^{(i)}=g^{(i)},\quad \mbox{tr}\phi^{(i)}=h^{(i)},$$
for $i=1,2$. 
Then
\begin{equation*}
\begin{split}
||\nabla^{2}\phi'||_{L^2(S)}\leq &Ca\sum_{i\leq 1}(||\nabla^{i}f||_{L^2(S)}+||\nabla^{i}g||_{L^2(S)}+||\nabla^{i}h||_{L^2(S)}+||\phi||_{L^2(S)})\\
&+C\sum_{i\leq 1}(||\nabla^{i}f'||_{L^2(S)}+||\nabla^{i}g'||_{L^2(S)}+||\nabla^{i}h'||_{L^2(S)}+||\phi'||_{L^2(S)})
\end{split}
\end{equation*}

\end{proposition}
\begin{proof}
This is a direct consequence of Propositions \ref{commuteelliptic} and \ref{ellipticthm}.
\end{proof}

We can now begin to estimate $\tilde{\mathcal O}'$. First, we define $\Theta'$ and derive estimates for $\Theta'$ using transport equations. Let $\Theta'$ denote $(\nabla\trch)',(\nabla\trchb)',\mu',\mub',\kappa',\kappab'$. We have the following estimates:
\begin{proposition}\label{Thetap}
$$\sum_{i\leq 1}||\nabla^i(\nabla\trch',\nabla\trchb',\mu',\mub')||_{L^2(S_{u,\ub})}\leq Ca+C\delta^{\frac 12}(\tilde{\mathcal O}'+\mathcal R').$$
\end{proposition}
\begin{proof}
From the proof of Proposition \ref{ellipticTheta}, we know that each of $\nabla\trch,\nabla\trchb,\mu,\mub$ satisfies either
$$\nab_3\Theta=\psi\nab\psi_{\Hb}+\psi\psi\psi+\psi\Psi_{\Hb},$$
or
$$\nab_4\Theta=\psi\nab\psi_{H}+\psi\psi\psi+\psi\Psi_H,$$
where as in Section \ref{Ricciellipticsec}, we use the notation $\psi_H\in\{\trch,\chih,\eta,\etab,\omega,\trchb\}$, $\psi_{\Hb}\in\{\trch,\eta,\etab,\omegab,\trchb,\chibh\}$, $\Psi_H\in\{\beta,\rho,\sigma,\betab\}$ and $\Psi_{\Hb}\in\{\rho,\sigma,\betab,\alphab\}$.

We focus on those $\Theta$'s satisfying the $\nab_4$ equation. The other case can be treated analogously. By Proposition \ref{transportp}, we need to control
$$\sum_{i\leq 1}||(\nabla^i(\psi\nab\psi_{H}+\psi\psi\psi+\psi\Psi_{H}))'||_{L^1_{\ub}L^2(S_{u,\ub})},$$
$$\sum_{i\leq 1}||\nabla^i(\psi\nab\psi_{H}+\psi\psi\psi+\psi\Psi_{H})||_{L^2(S_{u,\ub})},$$
and 
$$\sum_{i\leq 1}||\nabla^i\Theta||_{L^2(S_{u,\ub})}.$$
We first estimate the $'$ terms. Among those we first look at the term with two derivatives on the Ricci coefficient. We have, using Proposition \ref{Riccip}, 
$$||(\psi\nabla^2\psi_H)'||_{L^1_{\ub}L^2(S_{u,\ub})}\leq C\tilde{\mathcal O}'+ C||\psi'||_{L^1_{\ub}L^\infty(S_{u,\ub})}+C\delta^{\frac 12}||\nabla^2\psi_H'||_{L^2_{\ub}L^2(S_{u,\ub})}\leq C(a+\tilde{\mathcal O}'+\delta^{\frac 12}\mathcal R').$$
The other term with Ricci coefficients can be estimated by $\mathcal O'$ and thus by Proposition \ref{Riccip}:
$$\sum_{i_1+i_2\leq 1} ||(\psi^{i_1}\nabla^{i_2}\psi\nabla\psi)'||_{L^1_{\ub}L^2(S_{u,\ub})}\leq C(a+\delta^{\frac 12}\mathcal R'+\delta^{\frac 12}\tilde{\mathcal O}').$$
We then move to the term with the curvature components:
\begin{equation*}
\begin{split}
&\sum_{i_1+i_2+i_3\leq 1} ||(\psi^{i_1}\nabla^{i_2}\psi\nabla^{i_3}\Psi_{H})'||_{L^1_{\ub}L^2(S_{u,\ub})}\\
\leq &C(a+\delta^{\frac 12}\mathcal R'+\delta^{\frac 12}\tilde{\mathcal O}')+C\delta^{\frac 12}\sum_{i\leq 1}||\nab^i\Psi_{H}'||_{L^2_{\ub}L^2(S_{u,\ub})}\\
\leq &C(a+\delta^{\frac 12}\mathcal R'+\delta^{\frac 12}\tilde{\mathcal O}').
\end{split}
\end{equation*}
The terms without $'$ can be estimated using the a priori estimates in Theorem \ref{timeofexistence}:
$$\sum_{i\leq 1}(||\nabla^i(\psi\nab\psi_{H}+\psi\psi\psi+\psi\Psi_{H})||_{L^2(S_{u,\ub})}+
||\nabla^i\Theta||_{L^2(S_{u,\ub})})\leq C.$$
Thus, by Proposition \ref{transportp}, we have
$$\sum_{i\leq 1}||\nabla^i(\nabla\trch',\nabla\trchb',\mu',\mub')||_{L^2(S_{u,\ub})}\leq Ca+C\delta^{\frac 12}(\tilde{\mathcal O}'+\mathcal R').$$
\end{proof}
We now consider the difference quantities $\kappa'$ and $\kappab'$. As in Proposition \ref{omegaelliptic}, they will only satisfy $L^2(H_u)$ and $L^2(\Hb_{\ub})$ estimates. As in the proof of Proposition \ref{omegaelliptic}, in what follows, we allow $\psi$ also to be $\omega^\dagger$ and $\omegab^\dagger$.
\begin{proposition}\label{ellipticomegap}
$$\sum_{i\leq 1}||\nabla^i\kappa'||_{L^2(H_u)}\leq C\delta^{\frac 12}(a+\tilde{\mathcal O}'+\mathcal R'),$$
$$\sum_{i\leq 1}||\nabla^i\kappab'||_{L^2(\Hb_{\ub})}\leq C\delta^{\frac 12}(a+\tilde{\mathcal O}'+\mathcal R').$$
\end{proposition}
\begin{proof}
Recall equations (\ref{kappaeqn}) and (\ref{kappabeqn}):
$$\nab_3\kappa=\psi\nab\psi+\psi\psi\psi+\psi\Psi,$$
or
$$\nab_4\kappab=\psi\nab\psi+\psi\psi\psi+\psi\Psi.$$
We focus on $\kappa$. As $\kappab$ is easier. The only terms that are new compared to the proof of Proposition \ref{Thetap} 
$$\sum_{i\leq 1}||\nab^i(\psi\beta)'||_{L^1_{u}L^2(S_{u,\ub})}\mbox{ and }\sum_{i\leq 2}||\nab^i(\chih',\omega')||_{L^1_{u}L^2(S_{u,\ub})}.$$
Thus using the estimates in the proof of Proposition \ref{Thetap}, we have
$$\sum_{i\leq 1}||\nabla^i\kappa'||_{L^2(S_{u,\ub})}\leq Ca+C\delta^{\frac 12}(\tilde{\mathcal O}'+\mathcal R)+C\delta^{\frac 12}(\sum_{i\leq 1}||\nab^i\beta'||_{L^2_uL^2(S_{u,\ub})}+\sum_{i\leq 2}||\nab^i(\chih',\omega')||_{L^2_uL^2(S_{u,\ub})}).$$
Integrating over $\ub$ in $L^2$, we get
\begin{equation*}
\begin{split}
&\sum_{i\leq 1}||\nabla^i\kappa'||_{L^2(H_u)}\\
\leq &C\delta^{\frac 12}a+C\delta(\tilde{\mathcal O}'+\mathcal R')+C\delta^{\frac 12}(\sum_{i\leq 1}||\nab^i\beta'||_{L^2_{\ub}L^2_uL^2(S_{u,\ub})}+\sum_{i\leq 2}||\nab^i(\chih',\omega')||_{L^2_{\ub}L^2_uL^2(S_{u,\ub})}).
\end{split}
\end{equation*}
Thus, using the definition of the norms $\tilde{\mathcal O}'$ and $\mathcal R'$,
$$\sum_{i\leq 1}||\nabla^i\kappa'||_{L^2(H_u)}\leq C\delta^{\frac 12}a+C\delta(\tilde{\mathcal O}'+\mathcal R').$$
Similarly
$$\sum_{i\leq 1}||\nabla^i\kappab'||_{L^2(H_u)}\leq C\delta^{\frac 12}a+C\delta(\tilde{\mathcal O}'+\mathcal R').$$
\end{proof}
We use Propositions \ref{Thetap} and \ref{ellipticomegap} and elliptic estimates to derive all the estimates for $\nabla^{2}\psi'$.
\begin{proposition}\label{Ricciellipticp}
$$\tilde{\mathcal O}'\leq Ca+C\mathcal R'.$$
\end{proposition}
\begin{proof}
We will first prove
$$\sup_u||\nabla^{2}(\chih',\omega',\eta',\etab')||_{L^2(H_u)}+\sup_{\ub}||\nabla^2(\chibh',\omegab',\eta',\etab')||_{L^2(\Hb_{\ub})}\leq Ca+C\delta^{\frac 12}\tilde{\mathcal O}'+C\mathcal R'.$$
We have the following div-curl systems:
$$\div(\chih,\eta,\etab)=(\nab\trch,\mu)+\psi\psi+(\beta,\rho,\sigma),$$
$$\curl(\chih,\eta,\etab)=(\nab\trch,\mu)+\psi\psi+(\beta,\rho,\sigma),$$
and
$$\div(\chibh,\eta,\etab)=(\nab\trchb,\mub)+\sum_{\psi\neq \chih,\omega}\psi\psi+(\rhoc,\sigmac,\betab),$$
$$\curl(\chibh,\eta,\etab)=(\nab\trchb,\mub)+\sum_{\psi\neq \chih,\omega}\psi\psi+(\rhoc,\sigmac,\betab);$$
as well as the div-curl system
$$\div(\nab\omega,\nab\omega^\dagger)=\nab\kappa+\psi\psi+(\beta,\rho,\sigma),$$
$$\curl(\nab\omega,\nab\omega^\dagger)=\nab\kappa+\psi\psi+(\beta,\rho,\sigma),$$
and
$$\div(\nab\omegab,\nab\omegab^\dagger)=\nab\kappab+\sum_{\psi\neq \chih,\omega}\psi\psi+(\rhoc,\sigmac,\betab),$$
$$\curl(\nab\omegab,\nab\omegab^\dagger)=\nab\kappab+\sum_{\psi\neq \chih,\omega}\psi\psi+(\rhoc,\sigmac,\betab).$$
By Proposition \ref{ellipticthmp}, we need to control the terms involving $\Theta$
$$\sup_{u}\sum_{i\leq 1}||\nabla^i(\nab\trch',\mu',\kappa')||_{L^2(H_u)}+\sup_{\ub}\sum_{i\leq 1}||\nabla^i(\nab\trchb',\mub',\kappab')||_{L^2(\Hb_{\ub})},$$
the curvature terms
$$\sup_u\sum_{i\leq 1}||\nabla^i(\beta',\rho',\sigma')||_{L^2(H_u)}+\sup_{\ub}||\nabla^i(\rhoc',\sigmac',\betab')||_{L^2(\Hb_{\ub})}$$
and the lower order terms involving the Ricci coefficients
$$\sum_{i\leq 1}||\nabla^i(\psi'\psi)||_{L^2(H_u)}+\sum_{i\leq 1}\sum_{\psi\neq \chih,\omega}||\nabla^i(\psi'\psi)||_{L^2(\Hb_{\ub})}.$$
The terms involving $\Theta$ are controlled by the estimates in Proposition \ref{Thetap} by
$$C\delta^{\frac 12}(a+\tilde{\mathcal O}'+\mathcal R').$$
The curvature terms can be estimated using the definition of $\mathcal R'$ by
$$C\mathcal R'.$$
The remaining terms can be bounded by
\begin{equation*}
\begin{split}
&\sum_{i\leq 1}||\nabla^i(\psi'\psi)||_{L^2(H_u)}\\
\leq &C(||\psi'||_{L^\infty(S_{u,\ub})}\sum_{i\leq 1}||\nabla^{i}\psi||_{L^2(S_{u,\ub})}+\sum_{i\leq 1}||\nabla^{i}\psi'||_{L^2(S_{u,\ub})}||\nabla^{i_2}\psi||_{L^\infty(S_{u,\ub})})\\
\leq &Ca+C\delta^{\frac 12}\tilde{\mathcal O}'+C\delta^{\frac 12}\mathcal R'.
\end{split}
\end{equation*}
Thus we have
$$\sup_u||\nabla^{2}(\chih',\omega',\eta',\etab')||_{L^2(H_u)}+\sup_{\ub}||\nabla^2(\chibh',\omegab',\eta',\etab')||_{L^2(\Hb_{ub})}\leq Ca+C\delta^{\frac 12}\tilde{\mathcal O}'+C\mathcal R'.$$
Therefore, together with the estimates for $\nab^2\trch'$ and $\nab^2\trchb'$ in Proposition \ref{Thetap}, we have
$$\tilde{\mathcal O}\leq Ca+C\delta^{\frac 12}\tilde{\mathcal O}'+C\mathcal R'.$$
By choosing $\delta$ sufficiently small depending on $C$ so that $C\delta^{\frac 12}$, we have
$$\tilde{\mathcal O}\leq Ca+C\mathcal R'.$$
Notice that $C$ and $\delta$ depend only on the a priori estimates in Theorem \ref{timeofexistence} and are both independent of $a$.
\end{proof}
\subsection{Estimates for the Curvature Difference}\label{convsec5}
In order to finish the proof of Theorem \ref{convergencethm}, we need to estimate $\mathcal R'$ by $Ca$. This will be proved using an energy-type estimate. As in the bounds for the curvature components themselves, we will derive the energy estimates directly from the Bianchi equations without using the Bel Robinson tensor. This method is especially useful in controlling the difference of the curvature components since these difference quantities do not arise from a Weyl field with respect to either of the background spacetime metrics.
\begin{proposition}
\begin{equation*}
\begin{split}
&\sum_{i\leq 1}\left(\sum_{\Psi\in\{\beta,\rho,\sigma,\betab\}}\int_{H_u} (\nabla^i\Psi')^2+\sum_{\Psic\in\{\rhoc,\sigmac,\betab,\alphab\}}\int_{\underline{H}_{\underline{u}}} (\nabla^i\Psi')^2\right)\\
\leq& Ca^2+C\delta(\mathcal R')^2+C||(K\Psi)'||_{L^2_u L^2_{\ub}L^2(S)}^2\\
&+C\sum_{i_1+i_2+i_3\leq 1}||(\psi^{i_1}\nabla^{i_2}\psi\nabla^{i_3}\Psi)'||_{L^2_u L^2_{\ub}L^2(S)}^2+\sum_{i_1+i_2+i_3\leq 2}||(\psi^{i_1}\nabla^{i_2}\psi\nabla^{i_3}\psi)'||_{L^2_u L^2_{\ub}L^2(S)}^2.
\end{split}
\end{equation*}
\end{proposition}
\begin{proof}
The energy estimates for the curvature difference follows from the difference Bianchi equations. Once these are derived, the energy estimates will be obtained as in Section \ref{estimates} by integration by parts in appropriate subsets of the difference Bianchi system. In the proof below, we concentrate on the most difficult case which involves renormalized curvature components. The estimates for the other components can be derived in a similar fashion.

Recall that
\begin{equation*}
\begin{split}
&\nabla_3\beta+tr\chib\beta=\nabla\rho + 2\omegab \beta +^*\nabla\sigma +2\chih\cdot\betab+3(\eta\rho+^*\eta\sigma),\\
&\nabla_4\sigma+\frac 32tr\chi\sigma=-\div^*\beta+\frac 12\chibh\cdot ^*\alpha-\zeta\cdot^*\beta-2\etab\cdot
^*\beta,\\
&\nabla_4\rho+\frac 32tr\chi\rho=\div\beta-\frac 12\chibh\cdot\alpha+\zeta\cdot\beta+2\etab\cdot\beta,\\
\end{split}
\end{equation*}
As before, we renormalize the equations. Define
$$\sigmac:=\sigma+\frac{1}{2}\chibh\wedge\chih,$$
$$\rhoc:=\rho-\frac{1}{2}\chibh\cdot\chih.$$
Schematically, we have
\begin{equation*}
\begin{split}
&\nabla_3\beta=\nabla\rhoc+^*\nabla\sigmac +\psi\nabla\psi+\psi\Psi+\psi\psi\psi,\\
&\nabla_4\sigmac=-\div^*\beta+\psi\nabla\psi+\psi\Psi+\psi\psi\psi,\\
&\nabla_4\rhoc=\div\beta+\psi\nabla\psi+\psi\Psi+\psi\psi\psi.
\end{split}
\end{equation*}
From these we can derive equations for $\beta'$, $\sigma'$ and $\rhoc'$. By Proposition \ref{peqn}, 
\begin{equation*}
\begin{split}
&\nabla_3\beta'=(\nabla_3\beta)'+\frac{(\Omega^{-1})'}{\Omega^{-1}}\nabla_3\beta+\Omega^{-1}(b^A)'\nabla_{\frac{\partial}{\partial\th^A}}\beta+\Omega^{-1}(\gamma^{-1}\chi)'\beta+\Omega^{-1}b\Gamma'\beta,\\
&\nabla_4\sigmac=(\nabla_4\sigmac)'+\frac{(\Omega^{-1})'}{\Omega^{-1}}\nabla_4\sigmac+\Omega^{-1}(\gamma^{-1}\chi)'\sigmac,\\
&\nabla_4\rhoc=(\nabla_4\rhoc)'+\frac{(\Omega^{-1})'}{\Omega^{-1}}\nabla_4\rhoc+\Omega^{-1}(\gamma^{-1}\chi)'\rhoc.
\end{split}
\end{equation*}
Moreover, we can write schematically
\begin{equation*}
\begin{split}
&(\nab(\rhoc,\sigmac))'=\nabla(\rhoc',\sigmac')+\Gamma'(\rhoc,\sigmac)\\
&(\div\beta)'=\div\beta'+\gamma'\nabla\beta+\Gamma'\beta.
\end{split}
\end{equation*}
Therefore,
\begin{equation*}
\begin{split}
&\nabla_3\beta'=\nabla\rhoc'+^*\nabla\sigmac' +(\psi\nabla\psi+\psi\Psi+\psi\psi\psi)'+F_{\beta},\\
&\nabla_4\sigmac'=-\div^*\beta'+(\psi\nabla\psi+\psi\Psi+\psi\psi\psi)'+F_{\sigmac},\\
&\nabla_4\rhoc'=\div\beta'+(\psi\nabla\psi+\psi\Psi+\psi\psi\psi)'+F_{\rhoc},
\end{split}
\end{equation*}
where $F_{\beta}$, $F_{\sigmac}$ and $F_{\rhoc}$ satisfy the bound
$$\sum_{i\leq 1}||\nab^iF||_{L^2_uL^2_{\ub}L^2(S)}\leq C\delta^{\frac 12}(a+\mathcal R').$$
Therefore, multiplying the first equation by $\beta'$, integrating by parts and using the other equations, we have
\begin{equation*}
\begin{split}
&\int_{H_u} |\beta'|^2+\int_{\Hb_{\ub}} (|\rhoc'|^2+|\sigmac'|^2)\\
\leq& Ca^2 +||(\beta',\rhoc',\sigma')(\psi\nabla\psi+\psi\Psi+\psi\psi\psi)'||_{L^1_u L^1_{\ub}L^1(S)}+||(\beta',\rhoc',\sigma')F_{\beta,\rhoc,\sigmac}||_{L^1_u L^1_{\ub}L^1(S)}\\
\leq &Ca^2 +||(\beta',\rhoc',\sigma')||_{L^2_u L^2_{\ub}L^2(S)}+||(\psi\nabla\psi+\psi\Psi+\psi\psi\psi)'||_{L^2_u L^2_{\ub}L^2(S)}^2++||F||_{L^2_u L^2_{\ub}L^2(S)}^2\\
\leq &Ca^2 +C\delta(\mathcal R')^2+||(\psi\nabla\psi+\psi\Psi+\psi\psi\psi)'||_{L^2_u L^2_{\ub}L^2(S)}^2.
\end{split}
\end{equation*}
We now look at the equation the first covariant angular derivative for the 
curvature components:
\begin{equation*}
\begin{split}
&\nabla_3\nabla\beta=\nabla^2\rhoc+^*\nabla^2\sigmac+K(\rhoc,\sigmac) +\psi^2\nabla\Psi+\nabla\psi\Psi+\psi\nabla\Psi+\psi\nabla^2\psi+\psi^2\nabla\psi+\psi^4,\\
&\nabla_4\nabla\sigmac=-\div^*\nabla\beta+K\beta+\psi^2\nabla\Psi+\nabla\psi\Psi+\psi\nabla\Psi+\psi\nabla^2\psi+\psi^2\nabla\psi+\psi^4,\\
&\nabla_4\nabla\rhoc=\div\nabla\beta+K\beta+\psi^2\nabla\Psi+\nabla\psi\Psi+\psi\nabla\Psi+\psi\nabla^2\psi+\psi^2\nabla\psi+\psi^4.
\end{split}
\end{equation*}
A similar argument as before involving writing the equations for $\nabla(\beta',\rhoc',\sigmac')$ and integrating by parts yield
\begin{equation*}
\begin{split}
&\int_{H_u} |\nabla\beta'|^2+\int_{\Hb_{\ub}} (|\nabla\rhoc'|^2+|\nabla\sigmac'|^2)\\
\leq &Ca^2 +C\delta(\mathcal R')^2\\
&+||(K\Psi+\psi^2\nabla\Psi+\nabla\psi\Psi+\psi\nabla\Psi+\psi\nabla^2\psi+(\nabla\psi)^2+\psi^2\nabla\psi+\psi^4)'||_{L^2_u L^2_{\ub}L^2(S)}^2.
\end{split}
\end{equation*}
\end{proof}

From this one can conclude
\begin{proposition}\label{curvaturep}
\begin{equation*}
\begin{split}
\mathcal R'\leq Ca.
\end{split}
\end{equation*}
\end{proposition}
\begin{proof}
By the previous Proposition, we need to estimate
$$||(K\Psi)'||_{L^2_u L^2_{\ub}L^2(S)}^2,$$
$$\sum_{i_1+i_2+i_3\leq 1}||(\psi^{i_1}\nabla^{i_2}\psi\nabla^{i_3}\Psi)'||_{L^2_u L^2_{\ub}L^2(S)}^2,$$
$$\sum_{i_1+i_2+i_3\leq 2}||(\psi^{i_1}\nabla^{i_2}\psi\nabla^{i_3}\psi)'||_{L^2_u L^2_{\ub}L^2(S)}^2$$
By Proposition \ref{Kp} and the definition of the norm $\mathcal R'$ and the a priori estimates from Theorem \ref{timeofexistence}, we have
$$||(K\Psi)'||_{L^2_u L^2_{\ub}L^2(S)}^2\leq ||K'\Psi||_{L^2_u L^2_{\ub}L^2(S)}^2+||K \Psi'||_{L^2_u L^2_{\ub}L^2(S)}^2\leq C\delta^2(\mathcal R')^2.$$
To estimate the remaining terms, we will repeatedly invoke Proposition \ref{angularpL4} to exchange $(\nabla\Psi)'$ and $\nabla(\Psi')$ (or $(\nabla\psi)'$ and $\nabla(\psi')$ etc.), at the expense of an error term $\delta^2(\mathcal R')^2$.
\begin{equation*}
\begin{split}
&\sum_{i_1+i_2+i_3\leq 1}||(\psi^{i_1}\nabla^{i_2}\psi\nabla^{i_3}\Psi)'||_{L^2_u L^2_{\ub}L^2(S)}^2\\
\leq &C\delta^2(\mathcal R')^2+C(\sum_{i_1\leq 1}\sum_{i_2\leq 2}||\nabla^{i_1}\psi||_{L^\infty_u L^\infty_{\ub}L^\infty(S)}^{i_2})^2(\sum_{i_3\leq 1}||\nabla^{i_3}\Psi'||_{L^2_u L^2_{\ub}L^2(S)}^2)\\
&+C(\sum_{i_1\leq 1}||\nabla^{i_1}\psi||_{L^\infty_u L^\infty_{\ub}L^\infty(S)}^2)(\sum_{i_2\leq 1}||\nabla^{i_3}\Psi||_{L^2_u L^2_{\ub}L^4(S)}^2)(\sum_{i_3\leq 1}||\nabla^{i_1}\psi'||_{L^\infty_u L^\infty_{\ub}L^4(S)}^2)\\
\leq &C\delta(a^2+(\mathcal R')^2),
\end{split}
\end{equation*}
where the second term is estimated by Theorem \ref{timeofexistence} and the definition of $\mathcal R'$; and the last term is estimated by Theorem \ref{timeofexistence} and Proposition \ref{Riccip}. Finally, we estimate using Theorem \ref{timeofexistence} and Proposition \ref{Riccip}
\begin{equation*}
\begin{split}
&\sum_{i_1+i_2+i_3\leq 1}||(\psi^{i_1}\nabla^{i_2}\psi\nabla^{i_3}\psi)'||_{L^2_u L^2_{\ub}L^2(S)}^2\\
\leq &C\delta^2(\mathcal R')^2+C(\sum_{i_1\leq 1}\sum_{i_2\leq 2}||\nabla^{i_1}\psi||_{L^\infty_u L^\infty_{\ub}L^\infty(S)}^{i_2})^2(\sum_{i_3\leq 2}||\nabla^{i_3}\psi'||_{L^2_u L^2_{\ub}L^2(S)}^2)\\
\leq &C\delta(a^2+(\mathcal R')^2),
\end{split}
\end{equation*}
Putting these together, we have
$$(\mathcal R')^2\leq Ca^2+C\delta(\mathcal R')^2.$$
The conclusion follows by choosing $\delta$ sufficiently small depending only on $C$.
\end{proof}
This implies all the quantities that we have estimated in this section can be estimated by $Ca$:
\begin{proposition}\label{p}
$$||g'||_{L^\infty_u L^\infty_{\ub} L^\infty(S)},||\frac{\partial}{\partial\th}g'||_{L^\infty_u L^\infty_{\ub} L^4(S)},||\frac{\partial^2}{\partial\th^2}g'||_{L^\infty_u L^\infty_{\ub} L^2(S)},\mathcal O',\tilde{\mathcal O'},\mathcal R'\leq Ca.$$
\end{proposition}
\begin{proof}
This is a direct consequence of Propositions \ref{Omegap}, \ref{gammap}, \ref{dgammap}, \ref{dOmegap}, \ref{ddgammap}, \ref{ddOmegap}, \ref{bp}, \ref{Riccip}, \ref{Thetap}, \ref{Ricciellipticp} and \ref{curvaturep}.
\end{proof}
This concludes the proof of Theorem \ref{convergencethm}.

\subsection{Additional Estimates on the Difference of the Metrics}\label{AddEstMetric}

In this Subsection, we derive some additional estimates on the difference of the metrics. They follow directly from the bounds already derived. This will be used in Section \ref{limit} to show that the limiting spacetime metric is in the space asserted in Theorem \ref{convergencethm2} and satisfies the Einstein equations.

\begin{proposition}\label{ugp}
$$||\frac{\partial}{\partial u}(\gamma',\Omega',b')||_{L^4(S_{u,\ub})}\leq Ca.$$
\end{proposition}
\begin{proof}
$\gamma$ satisfies a transport equation in the $\Lb$ direction:
$$\Ls_{\Lb}\gamma=2\Omega\chib.$$
From this we can derive an equation for $\gamma'$:
$$\frac{\partial}{\partial u}\gamma'=-b^A\frac{\partial}{\partial\th^A}\gamma'-(b^A)'\frac{\partial}{\partial\th^A}\gamma+2(\Omega\chi)'.$$
Estimating the right hand side using Proposition \ref{p}, we have
$$||\frac{\partial}{\partial u}\gamma'||_{L^4(S_{u,\ub})}\leq Ca.$$
$\Omega$ also satisfies a transport equation in the $\Lb$ direction:
$$\omegab=-\frac 12\nabla_3\log\Omega=\frac 12\Omega\nab_3\Omega^{-1}=\frac 12(\frac{\partial}{\partial u}+b^A\frac{\partial}{\partial\th^A})\Omega^{-1}.$$
As for $\gamma'$, we can derive an equation for $\frac{\partial}{\partial u}(\Omega^{-1})'$. By direct estimates using Proposition \ref{p}, we have
$$||\frac{\partial}{\partial u}\Omega'||_{L^4(S_{u,\ub})}\leq Ca.$$
Finally, we move to $\frac{\partial}{\partial u}b'$. $b$ does not satisfy a transport equation. We therefore resort to the equation
$$\frac{\partial b^A}{\partial \ub}=-4\Omega^2\zeta^A.$$
Applying $\frac{\partial}{\partial u}$, we get
$$\frac{\partial^2 b^A}{\partial \ub\partial u}=-4\frac{\partial}{\partial u}(\Omega^2\zeta^A).$$
Since $(b')^A=0$ on $\Hb_0$, $\frac{\partial b'}{\partial u}=0$ on $\Hb_0$.
Thus
$$||\frac{\partial}{\partial u}b'||_{L^4(S_{u,\ub})}\leq C\int_0^{\ub} ||\frac{\partial}{\partial u}(\Omega^2\zeta^A)'||_{L^4(S_{u,\ub'})}d\ub'.$$
By Proposition \ref{p}, the right hand side $\leq Ca$.
\end{proof}

\begin{proposition}\label{u2gp}
$$||\frac{\partial^2}{\partial u^2}(\gamma',\Omega',b')||_{L^2(S_{u,\ub})}\leq Ca.$$
\end{proposition}
\begin{proof}
We prove this Proposition by taking an extra $\frac{\partial}{\partial u}$ derivative of the equations in the proof of Proposition \ref{ugp}. By
$$\Ls_{\Lb}\gamma=2\Omega\chib,$$
we have
$$\frac{\partial^2}{\partial u^2}\gamma=2\frac{\partial}{\partial u}(\Omega\chib)-\frac{\partial}{\partial u}(b^A\frac{\partial}{\partial\th^A}\gamma).$$
Notice that 
$$\Omega\frac{\partial}{\partial u}\chib=\nab_3\chib-\chib\chib.$$
Thus using the null structure equation for $\nab_3\chib$, and estimating directly, we get
$$||\frac{\partial^2}{\partial u^2}\gamma'||_{L^2(S_{u,\ub})}\leq Ca.$$
To estimate $\frac{\partial^2}{\partial u^2}\Omega$, we use
$$\omegab=\frac 12(\frac{\partial}{\partial u}+b^A\frac{\partial}{\partial\th^A})\Omega^{-1}.$$
Differentiating with respect to $\frac{\partial}{\partial u}$, using the estimates for $\frac{\partial}{\partial u}b^A$ and the null structure equation for $\nab_3\omega$, we get
$$||\frac{\partial^2}{\partial u^2}\Omega'||_{L^2(S_{u,\ub})}\leq Ca.$$
Finally, differentiating 
$$\frac{\partial b^A}{\partial \ub}=-4\Omega^2\zeta^A$$
twice in the $u$ direction, we get
$$\frac{\partial^3 b^A}{\partial \ub\partial u^2}=-4\frac{\partial^2}{\partial u^2}(\Omega^2\zeta^A).$$
Using the null structure equations as well as the Bianchi equation for $\nab_3\betab$, we see that all the terms on the right hand side can be bounded in $L^2(S_{u,\ub})$. Thus we can integrate to get
$$||\frac{\partial^2}{\partial u^2}b'||_{L^2(S_{u,\ub})}\leq Ca.$$
\end{proof}
It is also possible to prove estimates for the mixed second derivatives
\begin{proposition}\label{utgp}
$$||\frac{\partial^2}{\partial u\partial\th^A}(\gamma',\Omega',b')||_{L^2(S_{u,\ub})}\leq Ca.$$
\end{proposition}
\begin{proof}
This can be proved in a similar way as Proposition \ref{ugp}, taking $\frac{\partial}{\partial\th^A}$ instead of $\frac{\partial}{\partial u}$ derivative.
\end{proof}
On the other hand, $\ub$ derivatives can only be taken once and can only be estimated after taking the $L^{p_0}_{\ub}$ norm:
\begin{proposition}\label{ubgp}
$$||\frac{\partial}{\partial \ub}(\gamma',\Omega',b')||_{L^{p_0}_{\ub}L^\infty(S_{u,\ub})}\leq Ca.$$
\end{proposition}
\begin{proof}
We can directly estimate the right hand side of the equations
$$\frac{\partial}{\partial \ub}\gamma'=2(\Omega\chi)',$$
$$\frac 12\frac{\partial}{\partial \ub}(\Omega^{-1})'=\omega',$$
$$\frac{\partial}{\partial \ub}(b^A)'=-4(\Omega^2\zeta^A)'.$$
\end{proof}
The mixed $\ub$ and $\th$ derivatives can also be estimates:
\begin{proposition}\label{ubtgp}
$$||\frac{\partial^2}{\partial \ub \partial \th^A}(\gamma',\Omega',b')||_{L^{p_0}_{\ub}L^4(S_{u,\ub})}\leq Ca.$$
\end{proposition}\label{ubugp}
\begin{proof}
Take $\frac{\partial}{\partial \th^A}$ derivatives of the equations
$$\frac{\partial}{\partial \ub}\gamma'=2(\Omega\chi)',$$
$$\frac 12\frac{\partial}{\partial \ub}(\Omega^{-1})'=\omega',$$
$$\frac{\partial}{\partial \ub}(b^A)'=-4(\Omega^2\zeta^A)',$$
and apply the estimates in Proposition \ref{p}.
\end{proof}
Similarly for the mixed $\ub$ and $u$ derivatives:
\begin{proposition}
$$||\frac{\partial^2}{\partial \ub \partial u}(\gamma',\Omega',b')||_{L^{p_0}_{\ub}L^2(S_{u,\ub})}\leq Ca.$$
\end{proposition}
\begin{proof}
Take $\frac{\partial}{\partial u}$ derivatives of the equations
$$\frac{\partial}{\partial \ub}\gamma'=2(\Omega\chi)',$$
$$\frac 12\frac{\partial}{\partial \ub}(\Omega^{-1})'=\omega',$$
$$\frac{\partial}{\partial \ub}(b^A)'=-4(\Omega^2\zeta^A)',$$
and use the null structure equations.
\end{proof}
We do not have estimates for the second $\ub$ derivatives of the metric, since they can only be understood as measures even for the initial data. However, the following substitute is useful in showing that the limit spacetime that we construct satisfies the Einstein equations, in particular, $R_{u\ub}=R_{A\ub}=R_{\ub\ub}=0$.
\begin{proposition}\label{trchp}
$$||\frac{\partial^2}{\partial \ub^2}(b^A)'||_{L^{p_0}_{\ub}L^\infty(S_{u,\ub})}\leq Ca$$
and
$$||\frac{\partial}{\partial \ub}((\gamma^{-1})^{AB}\frac{\partial}{\partial \ub}\gamma_{AB})'||_{L^{p_0}_{\ub}L^4(S_{u,\ub})}\leq Ca.$$
\end{proposition}
\begin{proof}
The first estimate follows directly from 
$$\frac{\partial}{\partial \ub}(b^A)'=-4(\Omega^2\zeta^A)'$$
and the null structure equations.

For the second estimate, notice that in coordinates,
$$\trch=(\gamma^{-1})^{AB}\frac{\partial}{\partial \ub}\gamma_{AB}.$$
Thus we want to estimate $\frac{\partial}{\partial\ub}\trch'$. Using the null structure equation
$$\nab_4\trch+\frac 12(\trch)^2=-|\chih|^2-2\omega\trch,$$
we have
$$\frac{\partial}{\partial\ub}\trch'=(-\frac \Omega 2(\trch)^2-\Omega|\chih|^2-2\Omega\omega\trch)'.$$
The conclusion follows from estimating the right hand side directly using Proposition \ref{p}.
\end{proof}

\subsection{The Limiting Spacetime Metric}\label{limit}

In this Subsection, we prove Theorems \ref{convergencethm2} and \ref{Einstein}. We therefore assume that the conditions of Theorem \ref{convergencethm2} hold. Consider the sequence of metric
$$g_n=-2(\Omega_n)^2(du\otimes d\ub+d\ub\otimes du)+(\gamma_n)_{AB}(d\th^A-(b_n)^Adu)\otimes (d\th^B-(b_n)^Bdu).$$ 
By Proposition \ref{p}, $\gamma_n$, $b_n$ and $\Omega_n$ converge uniformly to their limiting values $\gamma_{\infty}$, $b_\infty$ and $\Omega_\infty$ and are therefore continuous functions of $(u,\ub,\th^1,\th^2)$ and define a continuous limiting spacetime metric
\begin{equation}\label{limitmetric}
g_\infty=-2(\Omega_\infty)^2(du\otimes d\ub+d\ub\otimes du)+ (\gamma_{\infty})_{AB}(d\th^A-(b_\infty)^Adu)\otimes(d\th^B-(b_\infty)^Bdu).
\end{equation}
To complete the proof of Theorem \ref{convergencethm2}, we need to demonstrate the desired regularity of the limiting spacetime. This follows easily from the estimates in the previous Subsections.
\begin{proposition}
$$(\frac{\partial}{\partial \th}g_n,\frac{\partial}{\partial u}g_n)\mbox{ converge to }(\frac{\partial}{\partial \th}g_\infty,\frac{\partial}{\partial u}g_\infty)\mbox{ in }L^\infty_u L^\infty_{\ub} L^4(S),$$
\begin{equation*}
\begin{split}
(\frac{\partial^2}{\partial \th^2}g_n,\frac{\partial^2}{\partial u\partial\th}g_n,\frac{\partial^2}{\partial u^2}g_n)\mbox{ converge to }\\
(\frac{\partial^2}{\partial \th^2}g_\infty,\frac{\partial^2}{\partial u\partial\th}g_\infty,\frac{\partial^2}{\partial u^2}g_{\infty})\mbox{ in }L^\infty_u L^\infty_{\ub} L^2(S),
\end{split}
\end{equation*}
\begin{equation*}
\begin{split}
(\frac{\partial}{\partial \ub}g_n, \frac{\partial}{\partial\ub}((\gamma_n^{-1})^{AB}\frac{\partial}{\partial\ub}(\gamma_n)_{AB}))\mbox{ converge to }\\
(\frac{\partial}{\partial \ub}g_{\infty}, \frac{\partial}{\partial\ub}((\gamma_{\infty}^{-1})^{AB}\frac{\partial}{\partial\ub}(\gamma_{\infty})_{AB}))\mbox{ in }L^\infty_u L^{p_0}_{\ub} L^\infty(S),
\end{split}
\end{equation*}
\begin{equation*}
\begin{split}
(\frac{\partial^2}{\partial \th \partial \ub}g_n,\frac{\partial^2}{\partial u\partial\ub}g_n,\frac{\partial^2}{\partial \ub^2} (b^A)_n)\mbox{ converge to }\\
(\frac{\partial^2}{\partial \th \partial \ub}g_\infty,\frac{\partial^2}{\partial u\partial\ub}g_\infty,\frac{\partial^2}{\partial \ub^2} (b^A)_{\infty})\mbox{ in }L^\infty_u L^{p_0}_{\ub} L^4(S).
\end{split}
\end{equation*}
\end{proposition}
\begin{proof}
That
$$\frac{\partial}{\partial \th}g_n,\frac{\partial}{\partial u}g_n\mbox{ converge in }L^\infty_u L^\infty_{\ub} L^4(S)$$
follows from Propositions \ref{p} and \ref{ugp} respectively.

That
$$\frac{\partial^2}{\partial \th^2}g_n,\frac{\partial^2}{\partial u^2}g_n,\frac{\partial^2}{\partial u\partial\th}g_n\mbox{ converge in } L^\infty_u L^\infty_{\ub} L^2(S)$$
follows from Propositions \ref{p}, \ref{u2gp} and \ref{utgp} respectively.

That
$$\frac{\partial}{\partial \ub}g_n, \frac{\partial}{\partial\ub}((\gamma_n^{-1})^{AB}\frac{\partial}{\partial\ub}(\gamma_n)_{AB})\mbox{ converge in }L^\infty_u L^{p_0}_{\ub} L^\infty(S)$$
follows from Propositions \ref{ubgp} and \ref{trchp} respectively.

That
$$\frac{\partial^2}{\partial \th \partial \ub}g_n,\frac{\partial^2}{\partial u\partial\ub}g_n,\frac{\partial^2}{\partial \ub^2} (b^A)_n\mbox{ converge in }L^\infty_u L^{p_0}_{\ub} L^4(S)$$
follows from Propositions \ref{utgp}, \ref{ubugp} and \ref{trchp} respectively.
\end{proof}
\begin{proposition}\label{gspace}
$$\frac{\partial}{\partial \th}g_\infty,\frac{\partial}{\partial u}g_\infty\in C^0_u C^0_{\ub} L^4(S),$$
$$\frac{\partial^2}{\partial \th^2}g_\infty,\frac{\partial^2}{\partial u\partial\th}g_\infty,\frac{\partial^2}{\partial u^2}g_\infty\in C^0_u C^0_{\ub} L^2(S),$$
$$\frac{\partial}{\partial \ub}g_\infty, \frac{\partial}{\partial\ub}((\gamma_\infty^{-1})^{AB}\frac{\partial}{\partial\ub}(\gamma_\infty)_{AB}) \in L^\infty_u L^\infty_{\ub} L^\infty(S),$$
$$\frac{\partial^2}{\partial \th \partial \ub}g_\infty,\frac{\partial^2}{\partial u\partial\ub}g_\infty,\frac{\partial^2}{\partial \ub^2} (b^A)_{\infty}\in L^\infty_u L^\infty_{\ub} L^4(S).$$
\end{proposition}
\begin{proof}
To prove the first two statements, notice that since $g_n$ are smooth, the convergence in $L^\infty_u L^\infty_{\ub}L^p(S)$ implies that the limit is in $C^0_u C^0_{\ub} L^p(S)$.

For the latter two statements, we use the fact that for $p_0<\infty$, if $f_n\to f$ in $L^{p_0}$ and $f_n$ is uniformly bounded in $L^\infty$, then $f\in L^\infty$.
\end{proof}

Proposition \ref{gspace} allows us to conclude that all first derivatives of $g_{\infty}$ can be defined and belong to an appropriate space. These norms also allow us to conclude that the product of any two first derivatives of the metric belongs to $L^\infty_u L^\infty_{\ub}L^2(S)$. For the second derivatives, Proposition \ref{gspace} shows that all second derivatives of the metric except $\frac{\partial^2}{\partial \ub^2}(\gamma,\Omega)$ are defined as functions in $L^\infty L^{p_0}_{\ub}L^2(S)$. The second derivative $\frac{\partial^2}{\partial \ub^2}(\gamma,\Omega)$ is merely a distribution. As we will see below, its definition is not necessary to make sense of the Einstein equations. 

We now show that all components of the curvature for the limiting spacetime, except for $R_{\ub A\ub B}$, $R_{\ub u \ub A}$ and $R_{\ub u \ub u}$, can be defined at worst as functions in $L^\infty_uL^\infty_{\ub}L^2(S)$. For this we use the coordinate definition of the Riemann curvature tensor:
\begin{equation}\label{curvdef}
R_{\delta\sigma\mu\nu}=g_{\delta\rho}(\frac{\partial}{\partial x^{\mu}}\Gamma^\rho_{\nu\sigma}-\frac{\partial}{\partial x^{\nu}}\Gamma^\rho_{\mu\sigma}+\Gamma^{\rho}_{\mu\lambda}\Gamma^\lambda_{\nu\sigma}-\Gamma^\rho_{\nu\lambda}\Gamma^\lambda_{\mu\sigma}).
\end{equation}
\begin{proposition}
In the limiting spacetime, all components of the Riemann curvature tensor, defined by the coordinate expression (\ref{curvdef}) expect for $R_{\ub A\ub B}$, $R_{\ub u \ub A}$ and $R_{\ub u \ub u}$, can be defined as functions in $L^\infty_uL^\infty_{\ub}L^2(S)$. Moreover, $R(L,\Lb,L,e_A)$ and $R(L,\Lb,L,\Lb)$ can be defined as functions in $L^\infty_uL^\infty_{\ub}L^2(S)$.
\end{proposition}
\begin{proof}
Using (\ref{curvdef}), it is easy to see that if $\delta,\sigma,\mu,\nu\neq \ub$, the expression has at most one $\ub$ derivative. Thus, by the regularity properties of $g_\infty$, these components of the Riemann curvature tensor are well-defined as $L^\infty_uL^\infty_{\ub}L^2(S)$ functions. 

In the case that exactly one of $\delta,\sigma,\mu,\nu$ is equal to $\ub$, we can assume, by the symmetry properties of the Riemann curvature tensor, that $\delta=\ub$. Then, the formula for the component of the Riemann curvature tensor has at most one $\ub$ derivative and by the regularity properties, it is well-defined as $L^\infty_uL^\infty_{\ub}L^2(S)$ functions. Notice, however, that strictly speaking, at the level of regularity that we have, the symmetry properties of the Riemann curvature tensor may not hold. Consider, for example, the component $R_{12\ub 1}$. Using the formula above, we would have a term $g_{1A}(g^{A\ub}g_{12,\ub})_{,\ub}$, which is not defined. However, the sum of \emph{all} terms containing two $\ub$ derivatives is equal to the expression
$$g_{1\alpha}\frac{\partial}{\partial \ub}(g^{\ub\alpha}\frac{\partial}{\partial \ub}g_{12}).$$
which can be rearranged, up to terms with at most one $\ub$ derivative, in the form
$$\frac{\partial}{\partial \ub}(g_{1\alpha}g^{\ub\alpha}\frac{\partial}{\partial \ub}g_{12}).$$
This expression vanishes since $g_{1\alpha}g^{\ub\alpha}=0$. The above calculation is a reflection of the symmetry properties of the Riemann curvature tensor and provides an appropriate distributional definition of $R_{\sigma\delta\mu\nu}$ for which the symmetry properties formally hold.

We have therefore shown that all components of curvature with the exception of $R_{\ub A\ub B}$, $R_{\ub u \ub A}$ and $R_{\ub u \ub u}$ are defined as $L^\infty_uL^\infty_{\ub}L^2(S)$ functions. However, replacing the double null coordinate system $(u,\ub,\th^1,\th^2)$ by the null frame $(L,\Lb,e^1,e^2)$, we can show that additional components of curvature are defined as functions in $L^\infty_uL^\infty_{\ub}L^2(S)$. To see this, we can write
$$R(L,\Lb,L,e_A)=R_{\ub u \ub A}+ b^B R_{\ub B\ub A},$$
and
$$R(L,\Lb,L,\Lb)=R_{\ub u \ub u}+ b^B R_{\ub u \ub B}+b^A R_{\ub A \ub u}+b^A b^B R_{\ub A\ub B}=R_{\ub u \ub u}+2b^A R_{\ub A \ub u}+b^A b^B R_{\ub A\ub B}.$$
Note that
$$R_{\ub u \ub u}=-\frac 12 g_{uu,\ub\ub}+\mbox{good terms}=-\frac 12 b^A b^B (\frac{\partial}{\partial \ub})^2 g_{AB}+\mbox{good terms},$$
$$R_{\ub u \ub A}=-\frac 12 g_{Au,\ub\ub}+\mbox{good terms}=\frac 12 b^B (\frac{\partial}{\partial \ub})^2 g_{AB}+\mbox{good terms},$$
$$R_{\ub B\ub A}=-\frac 12 g_{AB,\ub\ub}+\mbox{good terms},$$
where good terms denote terms with at most one $\ub$ derivative. Therefore, the terms with two $\ub$ derivatives cancel in $R(L,\Lb,L,A)$ and $R(L,\Lb,L,\Lb)$ and they can be defined as $L^\infty_uL^\infty_{\ub}L^2(S)$ functions.
\end{proof}
We now verify that the limiting spacetime that we have constructed satisfies the vacuum Einstein equations in the sense that relative to the system of double null coordinates $(u,\ub,\th^1,\th^2)$,
$$R_{\mu\nu}=0\quad\mbox{in }L^\infty_uL^\infty_{\ub}L^2(S).$$
For this, we use the coordinate definition of the Ricci curvature:
\begin{equation*}
R_{\mu\nu}=\frac{\partial}{\partial x^{\rho}}\Gamma^\rho_{\mu\nu}-\frac{\partial}{\partial x^{\nu}}\Gamma^\rho_{\rho\mu}+\Gamma^\rho_{\rho\lambda}\Gamma^{\lambda}_{\mu\nu}-\Gamma^\rho_{\mu\lambda}\Gamma^{\lambda}_{\rho\nu}.
\end{equation*}
\begin{proposition}
The limiting spacetime satisfies the vacuum Einstein equations in the sense that relative to the system of double null coordinates $(u,\ub,\th^1,\th^2)$,
$$R_{\mu\nu}=0\quad\mbox{in }L^\infty_uL^\infty_{\ub}L^2(S).$$.
\end{proposition}
\begin{proof}
Now, notice that with the metric given by (\ref{limitmetric}), we have
$$(g_{\infty}^{-1})^{\ub\ub}=(g_{\infty}^{-1})^{uu}=(g_{\infty}^{-1})^{u1}=(g_{\infty}^{-1})^{u2}=0.$$
Therefore,
$\Gamma^{\ub}_{\mu\nu}$ does not contain a term with $g_{\mu\nu,\ub}$. In other words, the expressions for 
$$R_{uu}, R_{uA}, R_{AB}\quad\mbox{for $A,B=1,2$}$$
do not contain two $\ub$ derivatives. Therefore, we have 
$$R_{uu}= R_{uA}= R_{AB}= 0$$
as functions in $L^\infty_uL^\infty_{\ub}L^2(S)$. For the component $R_{u\ub}$, there are terms involving two $\ub$ derivatives of the metric. In particular, we have
\begin{equation*}
\begin{split}
R_{u\ub}=&\frac{\partial}{\partial \ub}\Gamma^{\ub}_{u\ub}-\frac{\partial}{\partial \ub}\Gamma^{\ub}_{\ub u}-\frac{\partial}{\partial\ub}\Gamma^A_{Au}-\frac{\partial}{\partial\ub}\Gamma^u_{uu}\\
&+\mbox{terms involving at most one $\ub$ derivative of the metric}\\
=&\frac{-1}{4\Omega^2}\gamma_{AB}b^A\frac{\partial^2}{\partial\ub^2} b^B+\mbox{terms involving at most one $\ub$ derivative of the metric}.
\end{split}
\end{equation*}
From this expression, we see that except for the second $\ub$ derivatives of $b$, the terms involving two $\ub$ derivatives of the metric cancel. Recall from Proposition \ref{gspace} that $\frac{\partial^2}{\partial\ub^2} b^A$ is in $L^\infty_uL^\infty_{\ub}L^2(S)$. Thus we also have
$$R_{u\ub}=0$$
as functions in $L^\infty_uL^\infty_{\ub}L^2(S)$. Similarly, for $R_{A\ub}$, we have
\begin{equation*}
\begin{split}
R_{A\ub}=&\frac{\partial}{\partial \ub}\Gamma^{\ub}_{A\ub}-\frac{\partial}{\partial \ub}\Gamma^{\ub}_{\ub A}-\frac{\partial}{\partial\ub}\Gamma^B_{BA}-\frac{\partial}{\partial\ub}\Gamma^u_{uA}\\
&+\mbox{terms involving at most one $\ub$ derivative of the metric}\\
=&\frac{1}{4\Omega^2}\gamma_{AB}\frac{\partial^2}{\partial\ub^2} b^B+\mbox{terms involving at most one $\ub$ derivative of the metric}.
\end{split}
\end{equation*}
As for $R_{u\ub}$, because of the cancellation, we have
$$R_{\ub A}=0$$
as functions in $L^\infty_{u}L^\infty_{\ub}L^2(S)$.

It now remains to study $R_{\ub\ub}$. Indeed this term involves two $\ub$ derivatives of the metric and there are no cancellations to remove this term. We have
\begin{equation*}
\begin{split}
R_{\ub\ub}=&\frac{\partial}{\partial \ub}\Gamma^{\ub}_{\ub\ub}-\frac{\partial}{\partial \ub}\Gamma^{\ub}_{\ub\ub}+\frac{\partial}{\partial \th^A}\Gamma^{A}_{\ub\ub}+\frac{\partial}{\partial \ub}\Gamma^A_{A\ub}\\
&+\mbox{terms involving at most one $\ub$ derivative of the metric}.
\end{split}
\end{equation*}
The first two terms cancel each other. Note that
$$\Gamma^A_{\ub\ub}=0$$
since
$$g_{\ub\ub}=g_{\ub B}=0,\quad g^{Au}=0.$$
We are thus left with
$$\frac{\partial}{\partial \ub}\Gamma^A_{A\ub}=\frac{\partial}{\partial \ub}(g^{AB}\frac{\partial}{\partial \ub}g_{AB}).$$
It might seem that this term behaves like $\frac{\partial^2}{\partial \ub^2}\gamma$ which we cannot control. However, this particular combination of derivatives, by Proposition \ref{gspace}, is in fact in $L^\infty_uL^\infty_{\ub}L^4(S)$. Thus $R_{\ub\ub}=0$ in $L^\infty_uL^\infty_{\ub}L^2(S)$.
\end{proof}

\subsection{Uniqueness}\label{uniquenesssec}

In the Subsection, we show that the spacetime constructed in Theorem \ref{convergencethm2} is the unique spacetime among the class of continuous spacetimes which arise as $C^0$ limits of smooth spacetime solutions to the vacuum Einstein equations. In view of Theorem \ref{convergencethm}, uniqueness follows easily if the spacetime solution is assumed to be in the class of spacetimes satisfying the conclusions of \ref{convergencethm}. Nevertheless, we prove uniqueness within a larger class of spacetimes for which it is not a priori assumed that the spacetime is Lipschitz, or that it is more regular in the $u$ and the angular directions.

The following is a precise formulation of the uniqueness theorem (Theorem \ref{uniquenessthm}) and is the main result in this Subsection:
\begin{proposition}\label{uniquenessprop}
Let $(\mathcal M^{(1)}, g^{(1)})$ and $(\mathcal M^{(2)}, g^{(2)})$ be two $C^0$ Lorentzian spacetimes in double null coordinates, defined in $0\leq u\leq u_*, 0\leq \ub\leq \ub_*$ and taking the form
$$g^{(i)}=-2(\Omega^{(i)})^2(du\otimes d\ub+d\ub\otimes du)+(\gamma^{(i)})_{AB}(d\th^A-(b^{(i)})^Adu)\otimes (d\th^B-(b^{(i)})^Bdu)$$ 
such that
\begin{itemize}
\item $\Omega=1$ and $b^A=0$ on $H_0$ and $\Hb_0$ and $\gamma^{(1)}|_{H_0}=\gamma^{(2)}|_{H_0}$ and $\gamma^{(1)}|_{\Hb_0}=\gamma^{(2)}|_{\Hb_0}$
\item for $i=1,2$, there exists a sequence of smooth spacetimes $(\mathcal M_n^{(i)}, g^{(i)}_n)$ such that for every $n$, $g^{(i)}_n$ is defined in $0\leq u\leq u_*$ and $0\leq \ub\leq \ub_*$ and takes the following form in double null coordinates:
$$g^{(i)}_n=-2(\Omega^{(i)}_n)^2(du\otimes d\ub+d\ub\otimes du)+(\gamma^{(i)}_n)_{AB}(d\th^A-(b^{(i)}_n)^Adu)\otimes (d\th^B-(b^{(i)}_N)^Bdu)$$ 
with the property that
\begin{enumerate}
\item $$g^{(i)}_n \to g^{(i)} \quad\mbox{in }C^0\mbox{ in }0\leq u\leq u_*, 0\leq \ub\leq \ub_*,$$

\item The initial data for $g^{(i)}_n$ satisfy the assumptions of Theorem \ref{timeofexistence} uniformly,

\item The initial data for $g^{(i)}_n$ converges to the initial data for $g^{(i)}$, i.e.,
$$g^{(i)}_n|_{H_0} \to g^{(i)}|_{H_0} \quad\mbox{and }g^{(i)}_n|_{\Hb_0} \to g^{(i)}|_{\Hb_0}$$
in the norms in the assumptions of Theorem \ref{convergencethm}.
\end{enumerate}
\end{itemize}
Then, if $u_*, \ub_*\leq \epsilon$, where $\epsilon$ is as given in Theorem \ref{timeofexistence},
$$g^{(1)}=g^{(2)}\quad\mbox{in    }0\leq u\leq u_*, 0\leq \ub\leq \ub_* .$$
\end{proposition}
\begin{proof}
By Property $(3)$ in the assumptions, for every $i$, there exists $n_i$ such that 
$$\sup_u|(\frac{\partial}{\partial\th})^i(\gamma_{AB})_{n_i}'(\ub=0)|+\sup_{\ub}||(\frac{\partial}{\partial\th})^i(\gamma_{AB})_{n_i}'(u=0)|\leq 2^{-i},$$
$$\sum_{\psi\neq\chih,\omega}(\sum_{i\leq 1} \sup_{u}||\nabla^i\psi_{n_i}'||_{L^2(S_{u,0})}+\sum_{i\leq 1} \sup_{\ub}||\nabla^i\psi_{n_i}'||_{L^2(S_{0,\ub})})\leq 2^{-i},$$
$$\sum_{i\leq 1} ||\nabla^i(\chih_{n_i}',\omega_{n_i}')||_{L^{p_0}_{\ub}L^2(S_{0,\ub})}\leq 2^{-i}\quad\mbox{for some fixed }2\leq p_0<\infty,$$
$$||\nabla^2(\chih_{n_i}',\omega_{n_i}',\eta_{n_i}',\etab_{n_i}')||_{L^2(H_0)}+\sup_{\ub} ||\nabla^2(\chibh_{n_i}',\omegab_{n_i}',\eta_{n_i}',\etab_{n_i}')||_{L^2(\Hb_{0})}\leq 2^{-i},$$
$$\sup_{u} ||\nabla^2(\trch_{n_i}',\trchb_{n_i}')||_{L^2(S_{u,0})}+\sup_{\ub} ||\nabla^2(\trch_{n_i}',\trchb_{n_i}')||_{L^2(S_{0,\ub})}\leq 2^{-i},$$
$$\sum_{\Psi\in\{\beta,\rho,\sigma,\betab\}}\sum_{i\leq 1} ||\nabla^i\Psi_{n_i}'||_{L^{2}_{\ub}L^2(S_{0,\ub})}\leq 2^{-i}.$$
By Theorem \ref{convergencethm}, there exists $C$ depending only on the uniform bounds in the assumption $(2)$ such that
$$|g^{(1)}_{n_i}-g^{(2)}_{n_i}|\leq C2^{-i}.$$
Therefore, in $C^0$ norm,
$$|g^{(1)}-g^{(2)}|\leq C|g^{(1)}_{n_i}-g^{(2)}_{n_i}|\leq C2^{-i}.$$
Since this holds for every $i$, we have
$$g^{(1)}=g^{(2)}.$$
\end{proof}

\subsection{Propagation of Regularity}\label{regularityp}
Using Proposition \ref{propagationregularity}, we show that if the initial data is more regular in the $\nab_3$ and $\nab$ directions, then so is the limiting spacetime.
\begin{proposition}\label{regularityprop}
Suppose, in addition to the assumptions of Theorem \ref{convergencethm}, we have the following estimates for the initial data:
$$\sum_{j\leq J}\sum_{i\leq I}(\sum_{\Psi\in\{\rho,\sigma,\betab,\alphab\}}||\nab_3^j \nab^i\Psi||_{L^2(\Hb_0)}+\sum_{\Psi\in\{\beta,\rho,\sigma,\betab\}}||\nab_3^i \nab^j\Psi||_{L^2(H_0)})\leq C. $$
Then
$$\sum_{j\leq J}\sum_{i\leq I}(\sum_{\Psi\in\{\rho,\sigma,\betab,\alphab\}}\sup_{\ub}||\nab_3^j \nab^i\Psi||_{L^2(\Hb_{\ub})}+\sum_{\Psi\in\{\beta,\rho,\sigma,\betab\}}\sup_u||\nab_3^i \nab^j\Psi||_{L^2(H_u)})\leq C'.$$
\end{proposition}
\begin{proof}
By Proposition \ref{propagationregularity}, for an approximating sequence of spacetimes, the following bound holds independent of $n$:
$$\sum_{j\leq J}\sum_{i\leq I}(\sum_{\Psi\in\{\rho,\sigma,\betab,\alphab\}}\sup_{\ub}||\nab_3^j \nab^i\Psi_n||_{L^2(\Hb_{\ub})}+\sum_{\Psi\in\{\beta,\rho,\sigma,\betab\}}\sup_u||\nab_3^i \nab^j\Psi_n||_{L^2(H_u)})\leq C'.$$
The conclusion thus follows.
\end{proof}
Using the equations
$$\frac{\partial}{\partial \ub}\gamma=2\Omega\chi,$$
$$\frac 12\frac{\partial}{\partial \ub}\Omega^{-1}=\omega,$$
$$\frac{\partial}{\partial \ub}b^A=-4\Omega^2\zeta^A,$$
together with Proposition \ref{regularityprop}, we can show that $g$ is in the desired space as indicated in Theorem \ref{regularitythm}. The details are straightforward and will be omitted.

\subsection{The Limiting Spacetime for an Impulsive Gravitational Wave}\label{limitgiw}
In the Subsection, we show that for the case of an impulsive gravitational wave, the spacetime is more regular than a general spacetime with merely bounded $\chih$, thus concluding the proof of Theorem \ref{giwthmv1}. In particular, we show that $\alpha_\infty$ is well defined as a measure which is the weak limit of $\alpha_n$. All other curvature components can be defined in $L^\infty$. Moreover, the singularity of $\alpha_\infty$ is supported on a null hypersurface and the constructed spacetime is smooth away from this null hypersurface.

\subsubsection{Limit of the Curvature Component $\alpha$}

We first show that $\alpha_n$ has a well defined limit as a measure.
\begin{proposition}\label{alphalimit}
In the case of an impulsive gravitational wave, for every $u$ and every $\vartheta=(\th^1,\th^2)\in\mathbb S^2$, $\alpha_\infty(u,\vartheta)$ is a measure on $[0,\ub_*]$, defined as the limit of $\alpha_n(u,\vartheta)$
\end{proposition}
\begin{proof}
By Proposition \ref{totalvariation}, for every $u$ and $\vartheta$, $\alpha_n$ has uniformly bounded $L^1_{\ub}$ norms. Thus to show convergence of $\alpha_n$, we only need to show that for every $u$ and every $\vartheta=(\th^1,\th^2)\in\mathbb S^2$, 
$$\int_0^{\ub} \alpha_n(u,\ub',\vartheta) d\ub' $$
converges to a limit $\alpha_{\infty}(u,\vartheta)([0,\ub))$ as $n\to\infty$ for all $\ub$ such that $\alpha_{\infty}(u,\vartheta)([0,\ub))$ is continuous.
By the equation
$$\nab_4\chih+\trch\chih=-2\omega\chih-\alpha,$$
we have
$$\int_0^{\ub} \alpha_n(u,\ub',\vartheta) d\ub'=\int_0^{\ub} (\Omega^{-1}_n\frac{\partial}{\partial\ub}\chih_n+\trch_n\chih_n+2\omega_n\chih_n)(u,\ub',\vartheta) d\ub'.$$
Integrating by parts and using 
$$\frac{\partial}{\partial\ub}\Omega^{-1}_n=2\omega_n,$$
we derive
\begin{equation}\label{intalpha}
\int_0^{\ub} \alpha_n(u,\vartheta) d\ub=(\Omega^{-1}_n\chih_n)(u,\ub,\vartheta)-(\Omega^{-1}_n\chih_n)(u,\ub=0,\vartheta)+\int_0^{\ub} (\trch_n\chih_n)(u,\ub',\vartheta) d\ub'.
\end{equation}
By Proposition \ref{p}, 
$$\int_0^{\ub} (\trch_n\chih_n)(u,\ub',\vartheta) d\ub'$$
converges. Thus, in order to show weak convergence of $\alpha_n$, we need to show pointwise convergence for $\Omega^{-1}_n\chih_n$ for all $\ub$ such that $\alpha_{\infty}(u,\vartheta)([0,\ub))$ is continuous. By the construction of the data for an impulsive gravitational wave in Section \ref{initialcondition}, we know that for the initial data
$$(\chih'_n)_{AB}(u=0,\ub,\vartheta)\leq C2^{-n}\quad\mbox{for }\ub\leq\ub_s\mbox{ or }\ub\geq\ub_s+2^{-n}.$$
Using the equation
$$\nabla_3\chih+\frac 12\trchb\chih=\nabla\hot\eta+2\omegab\chih-\frac 12\trch\chibh+\eta\hot\eta,$$
one sees that 
$$(\chih'_n)_{AB}(u,\ub,\vartheta)\leq C2^{-n}\quad\mbox{for }\ub\leq\ub_s\mbox{ or }\ub\geq\ub_s+2^{-n}.$$
Therefore, $\chih_n$ converges for all $\ub\neq\ub_s$. By Proposition \ref{convergencethm2}, $\Omega^{-1}_n\chih_n$ also converges for all $\ub\neq\ub_s$.

It remains to show that the limit $\alpha_{\infty}(u,\vartheta)([0,\ub))$ is discontinuous at $\ub=\ub_s$. In view of (\ref{intalpha}), it suffices to show that $\chih$ has a jump discontinuity across $\Hb_{\ub_s}$. This follows from the equation
$$\nabla_3\chih+\frac 12\trchb\chih=\nabla\hot\eta+2\omegab\chih-\frac 12\trch\chibh+\eta\hot\eta.$$
Rewriting in coordinates and expressing $\chih$, $\chibh$, $\eta$, $\etab$ in terms of the coordinate vector fields $\frac{\partial}{\partial\th^1}$ and $\frac{\partial}{\partial\th^2}$ on the spheres, we have
\begin{equation*}
\begin{split}
&\Omega^{-1}(\frac{\partial}{\partial u}+b^A\frac{\partial}{\partial\th^A})\chih-\gamma^{-1}\chib\chih+\frac 1 2 \trchb \chih\\
=&\nab\widehat{\otimes} \eta+2\omegab \chih-\frac 12 \trch \chibh +\eta\widehat{\otimes} \eta.
\end{split}
\end{equation*}
Consider the coordinate system $(u,\ub,\tilde{\th}^1,\tilde{\th}^2)$ such that 
$$\frac{d}{d\ub}\tilde{\th}^A(\ub;u,\th)=b^A(u,\ub,\tilde\th^1,\tilde\th^2),$$
with initial data
$$\tilde{\th}^A(0;u,\th)=\th^A.$$
According to the proven estimates for $b$, this change of variable is $W^{2,\infty}$. Therefore, the vector fields $\frac{\partial}{\partial\tilde{\th}^1}$ and $\frac{\partial}{\partial\tilde{\th}^2}$ associated to the new coordinate system are $W^{1,\infty}$ with respect to the differentiable structure given by the original coordinate system $(u,\ub,\th^1,\th^2)$. In the new coordinate system, we can rewrite the transport equation for $\chih$ as
$$\Omega^{-1}\frac{\partial}{\partial u}\chih-\gamma^{-1}\chib\chih+\frac 1 2 \trchb \chih
=\nab\widehat{\otimes} \eta+2\omegab \chih-\frac 12 \trch \chibh +\eta\widehat{\otimes} \eta,$$
where $\chih$, $\chibh$, $\eta$, $\etab$ are now expressed in terms of the new coordinate vector fields $\frac{\partial}{\partial\tilde{\th}^1}$ and $\frac{\partial}{\partial\tilde{\th}^2}$.
Notice that $\chib$ is continuous and the expression
$$\nab\widehat{\otimes} \eta+2\omegab \chih-\frac 12 \trch \chibh +\eta\widehat{\otimes} \eta$$
is also continuous. Since for the initial data, $\chih(\tilde{u}_0,\ub,\theta)$ has a jump discontinuity for $\ub=\ub_s$, $\chih_{\tilde{A}\tilde{B}}$ also has a jump discontinuity across $\ub=\ub_s$. As noted before, the the vector fields $\frac{\partial}{\partial\tilde{\th}^1}$ and $\frac{\partial}{\partial\tilde{\th}^2}$ associated to the new coordinate system are $W^{1,\infty}$ with respect to the differentiable structure given by the original coordinate system $(u,\ub,\th^1,\th^2)$. Therefore, $\chih_{AB}$ also has a jump discontinuity.
\end{proof}

\subsubsection{Control of the Curvature Components $\beta,\rho,\sigma,\betab,\alphab$}

Except for the curvature component $\alpha$ which can only be defined as a measure, all other curvature components can be defined as $L^\infty_uL^\infty_{\ub}L^\infty(S)$ functions:

\begin{proposition}
$\beta,\rho,\sigma,\betab,\alphab$ can be defined as functions in $L^\infty_uL^\infty_{\ub}L^\infty(S)$.
\end{proposition}
\begin{proof}
That $\rho,\sigma,\betab,\alphab$ can be defined as functions in $L^\infty_uL^\infty_{\ub}L^\infty(S)$ follows from Proposition \ref{regularityprop} and the Sobolev Embedding Theorem. To show that $\beta$ can be defined as a $L^\infty_uL^\infty_{\ub}L^\infty(S)$ function, we consider the following Bianchi equation:
$$\nabla_3\beta=\sum_{\Psi\in\{\rho,\sigma\}}\nabla\Psi+\psi\sum_{\Psi\in\{\beta,\rho,\sigma,\betab\}}\Psi.$$
Since $\nab\rho, \psi, \rho, \sigma, \betab$ can be bounded in $L^\infty_uL^\infty_{\ub}L^\infty(S)$, it follows from Proposition \ref{transport} and Gronwall's inequality that $\beta$ can be estimated in $L^\infty_uL^\infty_{\ub}L^\infty(S)$, as desired.
\end{proof}

\subsubsection{Smoothness of Spacetime away from $\Hb_{\ub_s}$}

We now show that the spacetime is smooth away from the null hypersurface $\ub=\ub_s$. In the region $\ub<\ub_s$ before the impulse, this follows from standard theory of local-wellposedness. We focus on the region after the impulse, where $\ub>\ub_s$. In the following Proposition, we prove uniform estimates for the $\nab_4$ and $\nab$ derivatives of $\alpha_\infty$ in the region $\ub>\ub_s$. The bounds for all derivatives of the curvature components follow from a combination of the estimates on $\alpha_\infty$, Theorem \ref{regularitythm} and the Bianchi equations.
\begin{proposition}\label{giwsmoothness}
Suppose the data for the impulsive gravitational wave satisfy the estimates as given in Proposition \ref{dataprop} for some $I\geq 2$, $K\geq 0$. Then
$$\sum_{i\leq I-1}\sum_{k\leq \min\{K,\lfloor\frac{I-i-1}{2}\rfloor\}} \limsup_{\tilde{\ub}\to \ub_s+}||\nab_4^k\nab^i\alpha_{\infty}||_{L^2(H_u(\tilde{\ub},\ub_*))}\leq C_{I,K}'.$$
\end{proposition}
\begin{proof}
By Proposition \ref{alphaapriori},
\begin{equation}\label{alphaaprioricon}
\sum_{i\leq I-1}\sum_{k\leq \min\{K,\lfloor\frac{I-i-1}{2}\rfloor\}}\sup_{\ub}||\nab_4^k\nabla^i\alpha_n||_{L^2(H_u(\ub_s+2^{-n},\epsilon))}\leq C_{I,K}'.
\end{equation}
Fix $\delta>0$. Take $n$ large enough such that $2^{-n}< \delta$. Then (\ref{alphaaprioricon}) implies that 
$$\sum_{i\leq I-1}\sum_{k\leq \min\{K,\lfloor\frac{I-i-1}{2}\rfloor\}}\sup_{\ub}||\nab_4^k\nabla^i\alpha_n||_{L^2(H_u(\ub_s+\delta,\epsilon))}\leq C_{I,K}'.$$
The conclusion follows from the fact that this bound is uniform in $\delta$.
\end{proof}
As a consequence, we show
\begin{proposition}
The limiting spacetime is smooth away from the the null hypersurface $\Hb_{\ub_s}$.
\end{proposition}

\begin{remark}
In Proposition \ref{giwsmoothness}, we have only used the fact that for $\ub>\ub_s$, the initial data set is smooth. Therefore, suppose we have a initial data set satisfying the assumptions of Theorem \ref{rdthmv2}, with the additional assumption that it is smooth for $\ub>\tilde{\ub}$ for some $\tilde{\ub}$, we can also prove that the spacetime is smooth in $0\leq u\leq \epsilon$, $\tilde{\ub} <\ub\leq\epsilon$.
\end{remark}

\subsubsection{Decomposition of $\alpha_\infty$ into Singular and Regular Parts}

Finally, we show that $\alpha_\infty$ has a delta singularity on the incoming null hypersurface $\Hb_{\ub_s}$.
\begin{proposition}
$\alpha_{\infty}$ can be decomposed as
$$\alpha_{\infty}=\delta(\ub_s)\alpha_s+\alpha_r,$$
where $\delta(\ub_s)$ is the scalar delta function supported on the null hypersurface $\Hb_{\ub_s}$, $\alpha_s=\alpha_s(u,\vartheta)\neq 0$ belongs to $L^\infty_uL^\infty(S)$ and $\alpha_r$ belongs to $L^\infty_uL^\infty_{\ub}L^\infty(S)$.
\end{proposition}
\begin{proof}
Define 
$$\alpha_s(u,\vartheta):=\lim_{\ub\to\ub_s^+} \Omega^{-1}\chih (u,\ub,\vartheta)-\lim_{\ub\to\ub_s^-} \Omega^{-1}\chih (u,\ub,\vartheta),$$
and 
$$\alpha_r:=\alpha_{\infty}-\delta(\ub_s)\alpha_s.$$
We now show that $\alpha_s$ and $\alpha_r$ have the desired property. By Theorem \ref{convergencethm2}, $\alpha_s$ belongs to $L^\infty_uL^\infty(S)$. That $\alpha_s\neq 0$ follows from the fact that $\chih$ has a jump discontinuity across $\ub=\ub_s$, which is proved in Proposition \ref{alphalimit}.

It remains to show that $\alpha_r$ belongs to $L^\infty_u L^\infty_{\ub} L^\infty(S)$. To show this, we consider the measure of the half open interval $[0,\ub)$ using the measure $\alpha_r(u,\vartheta)$:
\begin{equation*}
\begin{split}
&(\alpha_r(u,\vartheta))([0,\ub))\\
=&(\Omega^{-1}\chih)(u,\ub,\vartheta)-\lim_{\tilde{\ub}\to\ub_s^+}(\Omega^{-1}\chih)(u,\tilde{\ub},\vartheta)+\lim_{\tilde{\ub}\to\ub_s^-}(\Omega^{-1}\chih)(u,\tilde{\ub},\vartheta)-(\Omega^{-1}\chih)(u,\ub=0,\vartheta)\\
&+\int_0^{\ub} (\trch\chih)(u,\tilde{\ub},\vartheta) d\tilde{\ub}\\
=&\lim_{\tilde{\ub}\to\ub_s^-}\int_0^{\tilde{\ub}} \frac{\partial}{\partial\ub}(\Omega^{-1}\chih)(u,\tilde{\ub}',\vartheta)d\tilde{\ub}'+\lim_{\tilde{\ub}\to\ub_s^+}\int_{\tilde{\ub}}^{\ub} \frac{\partial}{\partial\ub}(\Omega^{-1}\chih)(u,\tilde{\ub}',\vartheta)d\tilde{\ub}'\\
&+\int_0^{\ub} (\trch\chih)(u,\tilde{\ub},\vartheta) d\tilde{\ub}.
\end{split}
\end{equation*}
By Proposition \ref{giwsmoothness}, $\frac{\partial}{\partial\ub}(\Omega^{-1}\chih)(u,\ub,\vartheta)$ is a bounded function away from the the hypersurface $\Hb_{\ub_s}$. Thus $(\alpha_r(u,\vartheta))([0,\ub))$ can be expressed as an integral over $[0,\ub)$ of an $L^\infty_u L^\infty_{\ub}L^\infty(S)$ function, as desired.
\end{proof}

\bibliographystyle{hplain}
\bibliography{NewShock2}

\end{document}